\documentclass{elsarticle}
\biboptions{sort&compress}

\usepackage[fleqn]{amsmath}
\usepackage{amsmath}
\usepackage{amssymb}
\usepackage{amsthm}
\usepackage[margin=0.6in, paper=a4paper]{geometry}
\usepackage[colorlinks = true]{hyperref}
\usepackage[nameinlink]{cleveref}
\usepackage{graphicx}
\graphicspath{{figures/}}
\usepackage{subcaption}
\usepackage{wrapfig}
\usepackage{multirow}
\usepackage{makecell}
\usepackage{stackengine}

\usepackage[most]{tcolorbox}
\newcommand\coloredcomponent[2]
{
  \tcolorboxenvironment{#1}
  {
    breakable,
    enhanced,
    colback = white,
    colframe = #2,
    boxrule = 1pt,
    left = 2pt,
    right = 2pt,
    top = 2pt,
    bottom = 2pt,
    sharp corners,
    before skip = \topsep,
    after skip = \topsep,
  }
}

\counterwithin{equation}{section}
\theoremstyle{definition}

\newtheorem{theorem}{Theorem}[section]
\AddToHook{env/theorem/begin}{\crefalias{section}{theorem}}
\coloredcomponent{theorem}{blue}

\newtheorem{proposition}[theorem]{Proposition}
\AddToHook{env/proposition/begin}{\crefalias{theorem}{proposition}}
\coloredcomponent{proposition}{blue}

\newtheorem{corollary}[theorem]{Corollary}
\AddToHook{env/corollary/begin}{\crefalias{theorem}{corollary}}
\coloredcomponent{corollary}{blue}

\newtheorem{notation}[theorem]{Notation}
\AddToHook{env/notation/begin}{\crefalias{theorem}{notation}}
\coloredcomponent{notation}{green}

\newtheorem{definition}[theorem]{Definition}
\AddToHook{env/definition/begin}{\crefalias{theorem}{definition}}
\coloredcomponent{definition}{green}

\newtheorem{formulation}[theorem]{Formulation}
\AddToHook{env/formulation/begin}{\crefalias{theorem}{formulation}}
\coloredcomponent{formulation}{brown}

\newtheorem{discussion}[theorem]{Discussion}
\AddToHook{env/discussion/begin}{\crefalias{theorem}{discussion}}
\coloredcomponent{discussion}{yellow}

\newtheorem{example}[theorem]{Example}
\AddToHook{env/example/begin}{\crefalias{theorem}{example}}
\coloredcomponent{example}{purple}

\newtheorem{remark}[theorem]{Remark}
\AddToHook{env/remark/begin}{\crefalias{theorem}{remark}}
\coloredcomponent{remark}{orange}

\coloredcomponent{proof}{cyan}

\newcommand{\norm}[1]{\left\lVert#1\right\rVert}
\newcommand{\abs}[1]{\left\lvert#1\right\rvert}
\newcommand{\inner}[2]{\langle#1,#2\rangle}

\newcommand{\tr}{\mathop{\rm tr}\nolimits}
\newcommand{\Ker}{\mathop{\rm Ker}\nolimits}
\newcommand{\Hom}{\mathop{\rm Hom}\nolimits}
\newcommand{\vol}{\mathop{\rm vol}\nolimits}
\newcommand{\rel}{\varepsilon^\perp}
\newcommand{\OR}{\mathop{\rm or}\nolimits}

\newcommand{\N}{\mathbb{N}}

\newcommand{\R}{\mathbb{R}}

\newcommand{\set}[2]{\{#1 \mid #2\}}
\newcommand{\restrict}[2]{\left.#1\right|_{#2}}

\newcommand{\precdot}{%
  \mathrel{%
    \ensurestackMath{%
      \stackinset{c}{0.6ex}{c}{-0.3pt}{\cdot}{\prec}%
    }%
  }%
}

\newcommand{\succdot}{%
  \mathrel{%
    \ensurestackMath{%
      \stackinset{c}{-0.6ex}{c}{-0.3pt}{\cdot}{\succ}%
    }%
  }%
}

\newcommand{\newterm}[1]{\textbf{#1}}
\newcommand{\amount}{\mathsf{X}}
\newcommand{\potential}{\mathsf{Y}}
\newcommand{\mass}{\mathsf{M}}
\newcommand{\length}{\mathsf{L}}
\renewcommand{\time}{\mathsf{T}}
\newcommand{\temperature}{\theta}
\newcommand{\charge}{\mathsf{C}}

\newcommand{\diffusive}{\mathcal{D}}
\newcommand{\advective}{\mathcal{A}}

\newcommand{\Dirichlet}{\mathbf{D}}
\newcommand{\Neumann}{\mathbf{N}}

\newcommand{\topStrut}{\rule{0pt}{2.6ex}}

\begin{document}

\hypersetup{
  pdfauthor = {Kiprian Berbatov, Andrey Jivkov},
  pdfcreator = {pdflatex},
  pdfproducer = {Latex2e with hyperref},
  colorlinks = true,
  linkcolor = blue,
  citecolor = green,
  urlcolor = cyan,
  bookmarksnumbered = true,
  bookmarksopen = false,
}

\begin{frontmatter}

\title{Variational formulations of transport phenomena on combinatorial meshes}
\author[1]{Kiprian Berbatov\corref{fn1}}
\author[1]{Andrey P. Jivkov\corref{fn2}}
\address[1]{Department of Mechanical and Aerospace Engineering,
  The University of Manchester, Oxford Road, Manchester M13 9PL, UK}
\cortext[fn1]{Corresponding author: kiprian.berbatov@manchester.ac.uk}
\cortext[fn2]{Corresponding author: andrey.jivkov@manchester.ac.uk}

\begin{abstract}
  We develop primal and mixed variational formulations of transport phenomena on
  cell complexes with simple polytope connectivity.
  This framework addresses materials with internal structures comprising
  components of different topological dimensions, where cells of each dimension
  may possess distinct physical properties.
  The approach, which we call Combinatorial Mesh Calculus (CMC), extends
  Forman's combinatorial differential forms, previously used to formulate strong
  conservation laws.
  CMC operates directly on meshes without requiring smooth
  embeddings, using discrete analogues of the exterior derivative, Hodge star,
  and co-differential operators.
  Our mixed formulation leads to a block-diagonal mass-like matrix arising
  from inner products weighted by material coefficients, enabling efficient
  local elimination strategies within the mixed system.
  CMC differs from Discrete Exterior Calculus, which requires circumcentric
  duality and well-centred meshes, and from Finite Element Exterior Calculus,
  which constructs polynomial spaces on smooth domains.
  Our framework applies to general cell complexes, including curved cells and
  irregular meshes; nonetheless irregularity leads to worse numerical
  performance.
  The mathematical development proceeds in parallel between the smooth and
  discrete settings, establishing correspondences between continuous and
  discrete operators.
  Initial boundary value problems are formulated for mass diffusion, heat
  conduction, charge transport, and fluid flow through porous media.
  Numerical examples on regular and irregular meshes in two and three dimensions
  demonstrate agreement with analytical solutions.
  The framework enables modelling of transport in materials where
  microstructural topology influences macroscopic behaviour, with applications
  to polycrystalline materials, composites, and porous media.
\end{abstract}


\end{frontmatter}

\section{Introduction}

\subsection{Motivation}
\label{sec:introduction/motivation}

\noindent Modern experimental techniques are revealing, with increasing detail, the complexity of internal structures in materials \cite{kumar2023MC,hayashi2024MC} and their evolution under various forces \cite{stock2020XCT,woodruff2021SCT}.

These observations demonstrate that \textit{material structures at any length scale consist of finite-size components appearing with different topological dimensions}. Relative to the bulk, materials contain point defects, line defects, and surface defects. For example, at the polycrystalline level, metals and alloys appear as collections of 3D grains, 2D grain or interphase boundaries, and 1D boundary junctions. At sub-grain level, they contain regions with perfect atomic lattice (3D), stacking faults (2D), and dislocations (1D). This multi-dimensional view extends to the atomic scale, where 3D Wigner-Seitz cells around atoms connect through 2D faces and 1D edges \cite{hunklinger2022SSP}.

The behaviour of these structures emerges from interactions between components of all dimensions. Critically, \textit{any given component affects all components of lower dimensions on its boundary and all components of higher dimensions containing it}. Transport phenomena, including mass diffusion, heat conduction, and charge transport, occur differently through bulk regions, along grain boundaries and junction lines, with each pathway characterised by distinct material properties \cite{trovalusci2016MSM,brancherie2017MSM}. The evolution of these structures involves discrete, finite changes in arrangement, shape, size, and nature of components through processes such as shear, separation, diffusion, dissolution, and chemical reactions. These microstructural features localise and channel transport, making their accurate representation essential for predictive modelling \cite{burczynski2022MSM,baniassadi2023MCM}.

From a computational perspective, these multi-scale structures present unique challenges: (i) the need for structure-preserving discretisations that maintain conservation properties exactly, (ii) efficient solution of the resulting algebraic systems with heterogeneous material properties, and (iii) numerical methods that handle arbitrary polytopes without geometric quality constraints. Moreover, the discrete topological changes in material evolution require numerical frameworks that can naturally accommodate discontinuities and interface phenomena without artificial smoothing or enrichment techniques.

Current modelling approaches fail to capture this multi-dimensional complexity adequately. Continuum methods based on smooth topology \cite{morro2023Continuum} operate with averaged intensive properties and lead to partial differential equations that require numerical approximation. The most popular approaches - finite element methods \cite{brenner2008FEM}, finite difference methods \cite{thomas1995FDM}, and finite volume methods \cite{leveque2012FVM} - effectively capture macroscopic properties but struggle with discontinuities at interfaces and defects. These are addressed through various approximations including cohesive zone models \cite{ruiz2001CZM,park2011CZM}, extended finite element methods \cite{mohammadi2012XFEM,vellwock2024XFEM}, and phase-field methods \cite{chen2022phase,chen2024phase}. However, phase-field approaches use artificial smoothing for sharp transitions, potentially compromising accuracy for localised phenomena. Furthermore, continuum constitutive laws may inadequately capture the complexities of real materials, particularly under extreme deformations, multi-phase interactions, or non-equilibrium processes.

Discrete topology methods offer an alternative foundation through particle-based approaches including molecular dynamics \cite{santamaria2023MD}, peridynamics \cite{madenci2014PD}, smoothed particle hydrodynamics \cite{filho2019SPH}, and discrete element methods \cite{jebahi2015DEM}. These methods naturally handle atomic and mesoscopic phenomena but lack explicit topological connectivity, making it difficult to enforce constraints and predict emergent behaviours. Their absence of intrinsic volume concepts hampers effective treatment of long-range interactions. A significant limitation is their reliance on heuristic force laws that may lack physical justification.

The constraints of both discrete and continuous topologies can be overcome by adopting cell complexes \cite{kozlov2008,knudson2022} as the topological foundation for modelling. Cell complexes are collections of cells of different topological dimensions arranged following specific rules, studied in algebraic topology. Unlike continuum approaches relying on smooth differentiable manifolds, or discrete particle methods lacking explicit topological structure, cell complexes naturally encode the connectivity and geometry of complex material domains. This makes them particularly suited for modelling materials with intricate microstructures, such as polycrystalline metals, composite materials, and biological tissues.

In previous work \cite{berbatov2022diffusion}, we demonstrated a formulation of scalar conservation laws (mass, energy, charge, volume) on cell complexes analogous to the strong (differential) form in continuum settings. The formulation used combinatorial differential forms, introduced by Forman \cite{forman2002combinatorial}, extended with metric-dependent operations for physics applications. We showed construction of boundary value problems on cell complexes and applied the method to diffusion in composites with complex component arrangements. However, the strong form's analogue approximated Neumann boundary conditions. Weak (integral) formulations avoid this approximation and provide the natural setting for numerical implementation.

\subsection{Overview and contributions}
\label{sec:introduction/contribution}

\noindent
In \Cref{sec:exterior_calculus} we recall the essential elements of exterior calculus on smooth manifolds. The presentation is self-contained, offering abstract definitions of key terms used throughout the paper. While this overview is not a substitute for a textbook on differential geometry (e.g., \cite{lee2012introduction}), it introduces the mathematical machinery essential for the developments that follow. It also lays the foundation for the discrete mesh calculus introduced in \Cref{sec:combinatorial}, where all operations have natural counterparts in exterior calculus.

In \Cref{sec:continuum}, we formulate the continuum model for transport phenomena using exterior calculus. Two variational formulations on a $D$-manifold $M$ are derived: a primal weak formulation with the potential as a $0$-form (\Cref{sec:continuum/primal}), and a mixed weak formulation with the flow rate as a $(D-1)$-form and the dual potential as a $D$-form (\Cref{sec:continuum/mixed}).
Crucially, all physical quantities are defined on their natural geometric domains as differential forms, and the governing laws are expressed in the language of exterior calculus. To the best of our knowledge, such modelling of transport phenomena has not previously appeared in the literature. While the formulations could be rephrased in vector calculus, their expression in exterior calculus better captures the underlying physical structure. We note that the formulation highlights structural parallels between the continuous and discrete settings, as shown in \Cref{sec:discrete}, offering conceptual clarity for further theoretical development, and we believe it has significant pedagogical value.

Although the variational derivations conceptually mirror those in vector calculus, we derive them entirely within the language of exterior calculus. The resulting formulations -- particularly the mixed one -- are related to known results, such as the mixed variational formulation for the Dirichlet problem, a special case of the more general weak Hodge Laplacian problem in Finite Element Exterior Calculus (FEEC) \cite[Chapter 8]{arnold2018finite}.
While FEEC naturally accommodates essential and natural boundary conditions, our
emphasis is on their explicit interpretation in terms of primal and dual
combinatorial forms and associated physical quantities.

While the algebraic structure of the present framework is closely related to
Discrete Exterior Calculus and Finite Element Exterior Calculus, the intent of
this work is not to discretise a continuum problem posed on a smooth domain.
Instead, physical balance laws are formulated directly on a given cell complex,
which is treated as a primary representation of the system.
Cell measures and material parameters are introduced subsequently to encode
geometry and physics.
The framework should therefore be interpreted as a native discrete physical
model rather than as a numerical approximation of a continuous one.

In \Cref{sec:combinatorial} we introduce the formalism of combinatorial meshes and their Forman subdivisions. This builds on our previous work in \cite{berbatov2022diffusion}, which developed meshes of simple polytopes and combinatorial differential forms \cite{forman2002combinatorial}, interpreted as cochains on subdivisions of quasi-cubical meshes (meshes with cube-like connectivity). Notable differences from \cite{berbatov2022diffusion} include: (1) we now work with combinatorial meshes -- intrinsically defined discrete structures that can be embedded as manifold subdivisions, including with non-flat cells; (2) we adopt the inner product introduced in later work \cite{berbatov2023discrete}; (3) we drop the convexity requirement on mesh cells, leading to adjustments in the inner product and Hodge star definitions; and (4) we adopt the standard sign convention for the Hodge star.

In \Cref{sec:discrete}, we present the primal and mixed variational formulations for transport phenomena on combinatorial meshes. Although this section contains the core contribution of the paper, it is also the most straightforward. With the continuum variational formulations established in \Cref{sec:continuum} and the discrete machinery developed in \Cref{sec:combinatorial}, the discrete formulations follow by direct translation. The paper is structured to build the theory step by step - from the smooth to the discrete setting - ensuring all necessary mathematical foundations are established before their synthesis.

Finally, in \Cref{sec:simulations}, we verify our approach with several steady-state problems across different spatial dimensions and with varying domain and mesh geometries. While our long-term goal is to apply the method to complex microstructures (as in \cite[Section 4]{berbatov2022diffusion}), here we test it on manufactured solutions of continuum problems and report discretisation errors. The method can be viewed both as a discretisation of continuum models and as an intrinsically discrete modelling framework. The accompanying code \cite{berbatov2026cmc}, \url{https://github.com/kipiberbatov/cmc}, includes additional steady-state and transient examples, though only steady-state results are reported in this article for clarity.

\subsection{Style and notations}
\label{sec:introduction/style}

\noindent The sections are broken down into items named discussion, definition, notation, remark, proposition, corollary, formulation, and example, which are numbered to facilitate cross-referencing.
All unary operators (excluding taking the additive inverse) have higher precedence than binary operators (excluding taking exponents).
For instance, $f v \wedge w$ will mean $(f v) \wedge w$.
Further, we record extensively the physical dimensions of the quantities, operators, and equations we discuss.
If a quantity $v$ has dimension $X$, we will write $v [X]$ when we want to record the dimension of $v$, and $[[v]] = X$ when we do calculations with dimensions.
Physical dimensions will be represented by the SI quantities for mass $\mass$, length $\length$, time $\time$, temperature $\temperature$, and charge $\charge$.
(All except the last one are base SI quantities; for historical reasons electric current $\mathsf{I} = \charge \time^{-1}$ is used instead of $\charge$ as a base quantity, but we prefer charge since it is a conserved property of matter.)
If an operator $\varphi$ takes quantities $v_1 [X_1], ..., v_n[X_n]$ and returns a quantity $\varphi(v_1, ..., v_n) [Y X_1 ... X_n]$, we will say that $\varphi$ is of dimension $Y$, and write $\varphi [Y]$ or $[[\varphi]] = Y$.
Finally, for equations, if $v [X]$ and $w [X]$, we will emphasise their dimensions in equalities, writing $v = w [X]$.

\section{Overview of exterior calculus}
\label{sec:exterior_calculus}

\begin{discussion}
  We present only the elements of exterior calculus needed in the rest of the work.
  Some basic knowledge of general topology is required.
  Experience with smooth manifolds is also beneficial, although we introduce the terminology we use anyway.
  A comprehensive book on smooth manifolds is, for instance, \cite{lee2012introduction}.
  Exterior calculus offers several advantages over vector calculus: its notation is succinct and dimension-independent, and the underlying mathematical objects are coordinate-free and globally defined.
  The benefits are:
  \begin{enumerate}
    \item
      A clear distinction between manifold's intrinsic (non-metric) calculus (\Cref{sec:exterior_calculus/topological}) and (additional) calculus depending on a structure such as a Riemannian metric (\Cref{sec:exterior_calculus/metric}).
      This elucidates the distinction between different types of physical relations and the mapping of these relations straightforwardly to their discrete counterparts, for which we build a similar calculus.
    \item
      A representation of physical quantities by differential forms ready to be integrated over their respective domains.
      In contrast, in $3$D vector calculus, there are only scalar and vector fields, and it is impossible to differentiate between ``scalar'' quantities living on $0$- and $3$-cells, and ``vector'' quantities living on $1$-cells and $2$-cells.
      Terms like \newterm{pseudo-scalar fields} and \newterm{pseudo-vector fields} are often used, but they are related to sign invariance under change of orientation. Because sign invariance is an important notion, we will use the term \newterm{pseudo-forms}.
    \item
      Differential forms can express extensive physical quantities, not just densities (intensive ones).
    \item
      The post-processing of non-primary quantities in their respective formulations, such as flow rate in primal formulations and potential in mixed formulations, is clear in the respective discrete formulations.
  \end{enumerate}
\end{discussion}
\subsection{Topological operations}
\label{sec:exterior_calculus/topological}
\begin{definition}
  Let $D \in \N$, $M$ be a topological space.
  We say that $M$ is a \newterm{locally Euclidean} of \newterm{dimension} $D$ (write it $\dim M = D$) if every point $ x \in M$ has
  a neighbourhood $U$ homeomorphic to $\R^D$ ($x$ is an \newterm{interior point}),
  or a neighbourhood $U$ homeomorphic to the half-space $\R_{\geq 0} \times \R^{D - 1}$ ($x$ is a \newterm{boundary point}).
  A \newterm{chart} on $M$ is a pair $(U, \varphi)$, where $U$ is an open subset of $M$ and
  $\varphi \colon U \to \R^D$ (interior chart),
  or $\varphi \colon U \to \R_{\geq 0} \times \R^{D - 1}$ (boundary chart)
  is a homeomorphism.
  A chart $(U, \varphi)$ generates a \newterm{local coordinate system}, denoted by $\{x^p\}_{p = 1}^D$, where for $y \in U$, $x^p(y) := (\varphi(y))^p$ (here, for $z \in \R^D$, $z^p$ is the $p^{\rm th}$ component of $z$).

  We say that $M$ is a \newterm{topological manifold with boundary} if it is locally Euclidean, Hausdorff and second countable; Hausdorffness and second countability are technical requirements that will not be discussed.
  An \newterm{atlas} on $M$ is a collection of charts whose domains cover $M$.
  A \newterm{smooth atlas} is an atlas such that all \newterm{transition maps} between two charts $(U, \varphi)$ and $(V, \psi)$ are smooth in the sense of multivariable calculus, i.e.,
  \begin{equation}
    \restrict{\varphi}{U \cap V} \circ \psi^{-1}
  \end{equation}
  are smooth maps; their domain and codomain are subsets of $\R^D$.
  A \newterm{smooth manifold with boundary} is a topological manifold with boundary supplied with a smooth atlas.
  The \newterm{boundary} $\partial M$ of $M$ consists of all boundary points of $M$.
  It is a manifold without boundary with $\dim(\partial M) = D - 1$.
\end{definition}
\begin{definition}
  Let $M$ and $N$ be smooth manifolds with boundary, $\dim M = D$, $\dim N = D'$.
  A function $f \colon M \to \R$ is \newterm{smooth}
  if for every smooth charts $(U, \varphi)$ of $M$ and $(V, \psi)$ of $N$
  the multivariable function
  \begin{equation}
    \restrict{\psi}{V \cap f(U)} \circ \restrict{f}{U} \circ \varphi^{-1}
  \end{equation}
  is smooth in the sense of multivariable calculus; its domain and codomain are subsets of $R^D$ and $R^{D'}$, respectively.
  The set of smooth functions between $M$ and $N$ will be denoted by $\mathcal{C}^\infty(M, N)$.
  A smooth function $f \in \mathcal{C}^\infty(M, N)$ is called a \newterm{diffeomorphism} if it is invertible and its inverse is also smooth.
  $f$ is called \newterm{embedding}, if it restricts to a diffeomorphism to its image.
\end{definition}
\begin{remark}
  When $N = \R$ we will use the notation
  \begin{equation}
    \mathcal{F} M := \mathcal{C}^\infty(M, \R).
  \end{equation}
  Under pointwise addition, multiplication, and scalar multiplication, $\mathcal{F} M$ forms an algebra over $\R$.
  Disregarding scalar multiplication, $\mathcal{F} M$ is a \newterm{commutative ring with unity}, and as such will serve as a base ring for the $(\mathcal{F} M)$-\newterm{modules} of vector fields and differential forms.
  (Recall that a \newterm{module over commutative ring with unity} is the straightforward generalisation of a vector space over field, the base being the more general structure of a commutative ring with unity instead of a field.)
\end{remark}
\begin{definition}
  Let $M$ be a smooth manifold, $S \subseteq M$ that is a smooth manifold in its own.
  $S$ is called a \newterm{submanifold} of $M$ if the \newterm{inclusion map} $\iota_S \colon S \to M,\ \iota(s) = s$ is an embedding.

  For instance, $\partial M$ is a submanifold of $M$.
  Also, any open subset of $M$ is a submanifold of $M$.
\end{definition}
\begin{definition}
  Let $M$ be a smooth manifold, $x \in M$.
  A \newterm{tangent vector at $x$} is a map
  \begin{equation}
    v \colon \mathcal{F} M \to \R
  \end{equation}
  with the following property: for any $f, g \in \mathcal{F} M$,
  \begin{equation}
    v(f g) = v(f) \cdot g(x) + f(x) \cdot v(g).
  \end{equation}
  The space of all tangent vectors at $x$ is called the \newterm{tangent space at $x$}, denoted by $T_x M$.
  It is a real vector space.
\end{definition}
\begin{definition}
  Let $M$ be a smooth manifold of dimension $D$.
  A \newterm{vector field} on $M$ is a \newterm{derivation} on $\mathcal{F} M$, i.e., a map
  \begin{equation}
    X \colon \mathcal{F} M \to \mathcal{F} M
  \end{equation}
  that is $\R$-linear and for all $f, g \in \mathcal{F} M$,
  \begin{equation}
    X(f g) = (X f) \cdot g + f \cdot (X g).
  \end{equation}
  The space of vector fields over $M$ will be denoted by $\mathfrak{X} M$.
  It is a real vector space but also an $(\mathcal{F} M)$-module with the following multiplication: for every $f \in \mathcal{F} M,\ X \in \mathfrak{X} M$, $f X \in \mathfrak{X} M$ is defined so that for each $g \in \mathcal{F} M$,
  \begin{equation}
    (f X)(g) = f \cdot (X g).
  \end{equation}
\end{definition}
\begin{remark}
  Any vector field $X$ can be restricted to a tangent vector $\restrict{X}{x}$ at $x$ as follows: for any $f \in \mathcal{F} M$,
  \begin{equation}
    \restrict{X}{x}(f) := (X f)(x).
  \end{equation}
  Conversely, if a tangent vector is defined in a \emph{continuous fashion} at every point $x$, this gives rise to a vector field on $M$. The precise definition requires the theory of bundles and sections which will not be discussed here.
\end{remark}
\begin{remark}
  It can be proven \cite[Corollary 3.3]{lee2012introduction} that for any $y \in M$, $T_y M$ is a $D$-dimensional real vector space,
  and $\mathfrak{X} M$ is a $D$-dimensional module over $\mathcal{F} M$.
  In local coordinates (Cartesian coordinates when $M = \R^D$) a vector field $X$ and a tangent vector $\restrict{X}{y}$ have the representation
  \begin{equation}
    X = \sum_{p = 1}^D f_p \frac{\partial}{\partial x^p},\
    \restrict{X}{y} = \sum_{p = 1}^D f_p(y) \restrict{\frac{\partial}{\partial x^p}}{y}
  \end{equation}
  for some smooth functions $f_1, ..., f_D$. The operators $\partial / (\partial x^p)$ are the partial derivative operators.
\end{remark}
\begin{definition}
  Let $R$ be a commutative ring with unity (for our purposes, $\R$ or $\mathcal{F} M$ for some manifold $M$), $V$ be a module over $R$.

  The \newterm{dual space} $V^*$ of $V$ is defined as $\Hom(V, R)$, the space of linear functions on $V$.
  If $V$ is finite dimensional, $\dim V = D$, then also $\dim V^* = D$.
  To a basis $(e_1, ..., e_D)$ of $V$ corresponds the \newterm{dual basis} $(e^1, ..., e^D)$ of $V$, so that for $1 \leq p, q \leq D$, $e^p(e_q) = \delta^p_q$, where $\delta^p_q$ is the \newterm{Kronecker delta} ($1$ if $p = q$ and $0$ if $p \neq q$).

  The \newterm{exterior algebra} $\Lambda^\bullet V$ of $V$ is the freest associative and alternating algebra containing $V$, i.e.,
  it is the algebra generated by an associative operation $\wedge$, the \newterm{wedge (exterior) product},
  subject to the \newterm{alternating property}: for all $v \in V$, $v \wedge v = 0$.
  From the alternating property it follows the antisymmetric property:
  for any $v, w \in V$, $v \wedge w = - w \wedge v$.
  When $V$ is of dimension $D$ with basis $e_1, ..., e_D$, $\Lambda^\bullet V$ is decomposed into
  \begin{equation}
    \Lambda^\bullet V = \bigoplus_{p = 0}^D \Lambda^p V.
  \end{equation}
  Here, for $p \in \{0, ..., D\}$, $\Lambda^p V$ is the space of \newterm{$p$-vectors}, having dimension $\binom{D}{p}$, and basis $e_{i_1} \wedge ... \wedge e_{i_p}$, where
  $1 \leq i_1 < i_2 < ... < i_p \leq D$
  (when $p = 0$, the unit element $1 \in R$ forms a basis of $\Lambda^0 V$).
  It follows that the wedge product restricts to
  \begin{equation}
    \wedge_{p, q} \colon \Lambda^p V \times \Lambda^q V \to \Lambda^{p + q} V\quad
      (0 \leq p, q,\ p + q \leq D).
  \end{equation}
  The antisymmetric property is generalised to the \newterm{graded-commutative rule}:
  $\wedge_{q, p} = (-1)^{p q} \wedge_{p, q}$.
\end{definition}
\begin{definition}
  Let $M$ be a smooth manifold of dimension $D$, $0 \leq p \leq D$.
  The space $\Omega^p M$ of \newterm{differential $p$-forms} on $M$ is the space $\Lambda^p((\mathfrak{X} M)^*)$ of $p$-vectors on the $(\mathcal{F} M)$-dual of $\mathfrak{X} M$.
  The space of differential forms $\Omega^\bullet M$ is defined by the following decomposition:
  \begin{equation}
    \Omega^\bullet M := \bigoplus_{p = 0}^D \Omega^p M,
  \end{equation}
  and is an exterior algebra of over $\mathcal{F} M$, together with $\wedge$, the \newterm{exterior product of differential forms}.
  In a local coordinate system $\{x^q\}_{q = 1}^D$, for any $p \in \{0, ..., D\}$, the space $\Omega^p M$ of $p$-forms is spanned by
  \begin{equation}
    d x^{i_1} \wedge ... \wedge d x^{i_p},\
    1 \leq i_1 < i_2 < ... < i_p \leq D.
  \end{equation}
  Here, the $1$-forms $\{d x^q\}_{q = 1}^D$ form the basis of $\Omega^1 M = (\mathfrak{X} M)^*$ dual to the basis $\{\partial / (\partial x^q)\}_{q = 1}^D$ of $\mathfrak{X} M$.

  If $y \in M$, then the \newterm{fibre} $\Omega^p_y$ can be formed as the exterior algebra (over $\R$) of the \newterm{cotangent space} $(T_y M)^*$,
  and the above $(\mathcal{F} M)$-basis of $\Omega^p M$ restricts to the $\R$-basis of
  $\Omega^p_y = \Lambda^p((T_y M)^*)$ spanned by:
  \begin{equation}
    \restrict{d x^{i_1}}{y} \wedge ... \wedge \restrict{d x^{i_p}}{y},\
    1 \leq i_1 < i_2 < ... < i_p \leq D.
  \end{equation}
  The wedge product $\wedge$ is dimensionless.
\end{definition}
\begin{definition}
  Let $M$ be a smooth manifold of dimension $D$.
  The \newterm{exterior derivative} of forms,
  \begin{equation}
    \begin{split}
      & d_p \colon \Omega^p M \to \Omega^{p + 1} M
      & (0 \leq p < D), \\
      & d \colon \Omega^\bullet M \to \Omega^\bullet M
      & (d = \bigoplus_{p = 0}^{D - 1} d_p)
    \end{split}
  \end{equation}
  is an analytic operation on forms, and is the unique linear operation on forms that satisfies the following conditions:
  \begin{enumerate}
    \item
      $(\Omega^\bullet, d)$ is a \newterm{cochain complex}, i.e.,
      \begin{equation}
        d \circ d = 0\ (d_p \circ d_{p - 1} = 0,\ p = 1, ..., D - 1);
      \end{equation}
    \item
      \newterm{graded Leibniz rule}: for $\omega \in \Omega^p M$ and $\eta \in \Omega^q M$,
      \begin{equation}
        \label{eq:exterior_calculus/leibniz}
        d_{p + q}(\omega \wedge \eta) = d_p \omega \wedge \eta + (-1)^p \omega \wedge d_q \eta;
      \end{equation}
    \item
      $d_0$ is the \newterm{differential} on functions, i.e., for a function $f \in \mathcal{F} M$, and a vector field $X \in \mathfrak{X} M$,
      \begin{equation}
        (d f)(X) = X(f).
      \end{equation}
  \end{enumerate}
  The above equation in a coordinate system $\{x^q\}_{q = 1}^D$ takes the form
  \begin{equation}
    d f = \sum_{q = 1}^D \frac{\partial f}{\partial x^q} d x^q.
  \end{equation}
  In algebraic terms all conditions except the last one mean that $(\Omega^\bullet M, \wedge, d)$ is a \newterm{differential graded algebra}.

  For any $p$ the operator $d_p$ is dimensionless, hence $d$ is dimensionless as well.
\end{definition}
\begin{definition}
  Let $M$ and $N$ be smooth manifolds, $f \colon N \to M$ be a smooth function.
  The \newterm{pullback} of $f$ is the unique linear map
  \begin{equation}
    f^* \colon \Omega^\bullet M \to \Omega^\bullet N,\ f^*_p \colon \Omega^p M \to \Omega^p N,
  \end{equation}
  satisfying the following conditions:
  \begin{enumerate}
    \item
      it is a cochain map:
      \begin{equation}
        d \circ f^* = f^* \circ d,\ d_p \circ f^*_p = f^*_{p + 1} \circ d_p;
      \end{equation}
    \item
      it respects wedge products:
      for any $\omega, \eta \in \Omega^\bullet M$,
      \begin{equation}
        f^*(\omega \wedge \eta) = f^* \omega \wedge f^* \eta;
      \end{equation}
    \item
      pullback coincides with composition for functions:
      for any $g \in \Omega^0 M = \mathcal{F} M$,
      \begin{equation}
        f^* g = g \circ f.
      \end{equation}
  \end{enumerate}
\end{definition}
\begin{example}
  Let $M = N = \R^2$,
  \begin{equation}
    f \colon N \to M,\ f(r, \varphi) = (r \cos \varphi, r \sin \varphi),
  \end{equation}
  and take an arbitrary $2$-form
  \begin{equation}
    \omega \in \Omega^2 N,\ \omega = g\, d x \wedge \, d y.
  \end{equation}
  Denote $\tilde{g}(r, \varphi) = g(r \cos \varphi, r \sin \varphi)$.
  Then
  \begin{equation}
    \begin{split}
    (f^* \omega)(r, \varphi)
    & = (f^* g)(r, \varphi)\, f^*(d x) \wedge f^*(d y) \\
    & = g(r \cos \varphi, r \sin \varphi)\, d(r \cos \varphi) \wedge d(r \sin \varphi) \\
    & = \tilde{g}(r, \varphi) (\cos \varphi\, d r - r \sin \varphi\, d \varphi) \wedge (\sin \varphi\, d r + r \cos \varphi\, d \varphi) \\
    & = r \tilde{g}(r, \varphi)\, d r \wedge d \varphi.
    \end{split}
  \end{equation}
\end{example}
\begin{definition}
  Let $M$ be a smooth manifold, $S$ be a submanifold of $\partial M$, $\iota_S \colon S \to M$ be the inclusion map, i.e., $\iota_S(s) = s$ for all $s \in S$.
  The \newterm{trace} operator is defined by
  \begin{equation}
    \tr_S := \iota_S^* \colon \Omega^\bullet M \to \Omega^\bullet S.
  \end{equation}
\end{definition}
\begin{example}
  Let $M = \set{(x, y) \in \R^2}{x^2 + y^2 \leq 1}$ be the unit disk,
  whose boundary is the unit circle
  $\partial M = \set{(x, y) \in \R^2}{x^2 + y^2 = 1}$,
  $\iota_{\partial M} \colon \partial M \to M$ be the inclusion map,
  \begin{equation}
    \omega \in \Omega^1 M,\ \omega(x, y) = -y\, d x + x\, d y.
  \end{equation}
  We will express $\tr_{\partial M} \omega = \iota_{\partial M}^* \omega \in \Omega^1(\partial M)$.
  Consider an open subset $U$ of $\partial M$ that does not contain the point $(1, 0)$.
  A local parametrisation of $U$ is then given by
  \begin{equation}
    f \colon I \to U,\ f(\varphi) = (\cos \varphi, \sin \varphi),
  \end{equation}
  where $I$ is a subinterval of $(0, 2 \pi)$.
  Then $\tr_U \omega$ is expressed over $I$ as
  \begin{equation}
    f^*(\tr_U \omega)
    = -\sin \varphi\, d(\cos \varphi) + \cos \varphi\, d(\sin \varphi)
    = \sin^2 \varphi\, d \varphi + \cos^2 \varphi\, d \varphi
    = d \varphi
    \in \Omega^1 I.
  \end{equation}
\end{example}
\begin{definition}
  Let $M$ be a connected smooth manifold of dimension $D$,
  $\omega, \eta \in \Omega^D M$ be non-vanishing forms (that is, for any $x \in M$, $0 \neq \restrict{\omega}{x} \in \Omega^D_x$), $f \in \mathcal{F} M$ be the unique function satisfying $\omega = f \eta$.
  Its existence and uniqueness follow from $\dim_{\mathcal{F} M}(\Omega^D M) = 1$.
  We say that $\omega$ and $\eta$ are in the same \newterm{orientation class} if $f > 0$, and are in the opposite orientation class if $f < 0$.
  There are no other possibilities since $f$ is continuous, and $\omega$ and $\eta$ are nowhere zero.
  Being in the same orientation class is an equivalence relation on the set of non-vanishing top-dimensional forms.
  There are two possibilities, depending on the topology of $M$:
  \begin{enumerate}
    \item
      there are no non-vanishing $D$-forms which makes $M$ \newterm{non-orientable};
    \item
      there are non-vanishing $D$-forms on $M$, which gives two orientation classes on $M$, and making $M$ \newterm{orientable}.
      A choice of such a class $\OR$ is called an orientation on $M$, and the pair $(M, \OR)$ is an \newterm{oriented manifold}.
  \end{enumerate}
If $M$ is disconnected, each of its connected components can be oriented separately.
\end{definition}
\begin{example}
  Let $D \in \N$.
  The flat manifold $\R^D$ has non-vanishing top-dimensional forms of the form
  \begin{equation}
    \omega = f\, d x^1 \wedge ... \wedge d x^n,
  \end{equation}
  for a smooth function $f \colon \R^D \to \R \setminus \{0\}$.
  The standard orientation of $\R^D$ is characterised by $f > 0$ (with canonical representative $\omega = d x^1 \wedge ... \wedge d x^n$, i.e., $f \equiv 1$), and the opposite orientation is characterised by $f < 0$.
\end{example}
\begin{remark}
  Orientation is the first additional structure imposed on a smooth manifold.
  Indeed, a smooth manifold is a topological space with a smooth structure, from which we can canonically construct vector fields, differential forms, the wedge product and the exterior derivative.
  Orientation, however, is a choice of a sign that is not made canonically.
  It will be required for the definitions of integration of forms and the Hodge star operator (the latter requiring a metric structure as well).

  With potential exceptions not considered in this article, invariant physical quantities should not depend on the choice of an orientation.
  For this reason we introduce the informal notion of \newterm{pseudo-forms}, that is differential forms that change sign when orientation changes.
  More generally, we can talk about \newterm{pseudo-objects}, that is objects that change sign when orientation changes.
  For instance, the amount form that we will define later is a pseudo-form -- integrating it leads to a positive amount, but since integration is also orientation-dependent, it must change sign when orientation is changed.
  Note that orientation-independent operators (like wedge product and exterior derivative) will not change the orientation type of an object, while orientation-dependent ones (like integration or Hodge star) will reverse it.

  We can do the treatment of pseudo-objects formally, but that would require the introduction of orientation-aware analogues of all the operators we use, which will unnecessarily complicate our formalism.
  Note that the distinction between objects and pseudo-objects is usually not done extensively in the mathematical literature. It is usually clear when orientation is used but a careful analysis of differentiating between objects and pseudo-objects is lacking.
  We, however, believe it is worth mentioning, for the following reasons:
  \begin{enumerate}
    \item
      the analogical terms ``pseudo-scalars'' and ``pseudo-vectors'', which may be considered special cases in our formalism,  are notably used by physicists to distinguish between orientation-dependent and orientation-independent quantities, the latter having more explicit physical meaning;
    \item
      just like the notion of \textit{physical dimension} gives a ``reality check'' for invariance under scalar multiplication,
      the distinction between forms and pseudo-forms checks invariance under orientation reversal. Note, that we will systematically analyse the physical dimensions of all objects and operators in this article.
      For instance, just like two non-zero objects of different dimensions cannot be equal, a non-zero form and a non-zero pseudo-form could not be equal as well.
  \end{enumerate}
\end{remark}
\begin{definition}
  Let $M$ be a smooth manifold with boundary,
  $a \in \partial M$,
  $v \in T_a M \setminus T_a({\partial M})$.
  We say that $v$ is \newterm{inward-pointing} if there exists $\delta > 0$ and a smooth curve $\gamma \colon [0, \delta) \to M$, such that $\gamma(0) = a$ and $\gamma'(0) = v$.
  The equality $\gamma'(0) = v$ is defined as follows: for any smooth function $f \colon M \to \R$ ($f \in \mathcal{F} M$),
  \begin{equation}
    v(f) = (f \circ \gamma)'(0),
  \end{equation}
  where the right hand side is the standard derivative of the real-valued function
  $f \circ \gamma \colon [0, \delta) \to \R$.
  We say that $v$ is \newterm{outward-pointing} if $-v$ is inward pointing.
\end{definition}
\begin{example}
  Let $M := \set{(x, y) \in \R^2}{x^2 + y^2 \leq 1}$ be the closed unit disk with boundary the unit circle $S := \partial M = \set{(x, y) \in \R^2}{x^2 + y^2 = 1}$.
  Consider a point $p_0 := (\cos \varphi_0, \sin \varphi_0) \in S$.
  Then the (Euler) vector field
  \begin{equation}
    v := \cos \varphi_0 \restrict{\frac{\partial}{\partial x}}{p_0} + \sin \varphi_0 \restrict{\frac{\partial}{\partial y}}{p_0}
  \end{equation}
  is outward-pointing.
  Indeed, the curve $\gamma(t) = (1 - t)(\cos \varphi_0, \sin \varphi_0)$ witnesses the fact that $-v$ is inward-pointing.
\end{example}
\begin{definition}
  Let $R$ be a commutative ring with unity,
  $D \in \N$,
  $V$ be an $R$-module of dimension $D$.
  The \newterm{interior product}
  \begin{equation}
    i \colon V \to \Hom(\Lambda^p V^*, \Lambda^{p - 1} V^*),\ p = 1, ..., D,
  \end{equation}
  is the unique bilinear map such that:
  \begin{enumerate}
    \item
      for any $v \in V$, $i_v$ satisfies the graded Leibniz rule, i.e.,
      for any $\omega \in \Lambda^p V^*,\ \eta \in \Lambda^\bullet V^*$,
      \begin{equation}
        i_v(\omega \wedge \eta) = (i_v \omega) \wedge \eta + (-1)^p \omega \wedge (i_v \eta);
      \end{equation}
    \item
      $i$ is pairing on covectors: for any $v \in V$ and $\omega \in \Lambda^1 V^* \simeq V^*$,
      \begin{equation}
        i_v \omega = \omega(v) \in R \simeq \Lambda^0 V.
      \end{equation}
  \end{enumerate}
\end{definition}
\begin{definition}
  \label{def:exterior_calculus/relative_orientations}
  Let $M$ be an oriented smooth manifold of dimension $D$ with boundary and orientation $\OR_M$, and $S$ be an oriented $(D - 1)$-dimensional submanifold of $\partial M$ with orientation $\OR_S$.
  Consider an outward-pointing vector field $X$ on $S$.
  Since $\dim_{\mathcal{F} M}(\Omega^{D - 1} S) = 1$, there exists unique smooth function $f \colon S \to \R \setminus \{0\}$ such that
  \begin{equation}
    i_X \OR_M = f \OR_S.
  \end{equation}
  Since $X$ may not be defined outside $S$, the above equation is understood locally at each point $x$ of $S$, i.e.,
  \begin{equation}
    i_{\restrict{X}{x}}(\restrict{\OR_M}{x}) = f(x) \restrict{\OR_S}{x} \in \Omega^{D - 1}_x.
  \end{equation}
  Define the \newterm{relative orientation} between $M$ and $S$ as
  \begin{equation}
    \varepsilon(M, S) := {\rm sign}(f) =
    \begin{cases}
      1, & f > 0 \\
      -1, & f < 0
    \end{cases}.
  \end{equation}
  If $\varepsilon(M, S) = 1$, we say that $S$ has the \newterm{induced (boundary) orientation}.
\end{definition}
\begin{example}
  \label{ex:exterior_calculus/half_space_boundary_orientation}
  Let $M = \R_{\geq 0} \times \R^{D - 1}$ be a half space whose boundary is
  $\partial M = \{0\} \times \R^{D - 1}$.
  Then an outward-pointing vector field is
  $X := - \frac{\partial}{\partial x_1}$.
  The standard orientation on $M$ is given by
  $\OR_M := d x^1 \wedge ... \wedge d x^D$.
  Then the induced orientation on $\partial M$ is given by
  \begin{equation}
    \OR_{\partial M}
    := i_{X} \OR_M
    = - dx^2 \wedge ... \wedge dx^D.
  \end{equation}
  Note that since all manifolds with boundary are represented locally by
  $\R_{\geq 0} \times \R^{D - 1}$,
  this example is the archetypal way for finding induced orientation.
\end{example}
\begin{discussion}
  \label{dsc:exterior_calculus/diamond_property}
  Consider an oriented manifold $M$ with boundary $N$, $D = \dim M$, and let $P$ be an oriented $(D - 2)$-submanifold of $N$.
  Locally $P$ divides $N$ into exactly two parts $N'$ and $N''$.
  (This is the \newterm{diamond property} -- there are exactly two ``cells'' between $M$ and $P$.
  A similar property holds for convex polytopes as discussed in \cite[Theorem 2.7 (iii)]{ziegler1995lectures}.)
  Moreover, if $N'$ and $N''$ have the same orientations as $N$ (the induced boundary orientation),
  then they will have opposite relative orientations with $P$, i.e.,
  $\varepsilon(N', P) = - \varepsilon(N'', P)$.
  To allow for possible non-induced orientations on $N'$ and $N''$, we can write the following equality \cite[Supplementary matrial, Proposition Appendix B.13]{berbatov2022diffusion} which is satisfied in general:
  \begin{equation}
    \label{eq:exterior_calculus/chain_complex_property}
    \varepsilon(M, N')\, \varepsilon(N', P) + \varepsilon(M, N'')\, \varepsilon(N'', P) = 0.
  \end{equation}
\end{discussion}
\begin{definition}
  Let $M$ be an oriented smooth manifold with boundary, $D = \dim M$,
  \begin{equation}
    \Omega^\bullet_c M := \{\omega \in \Omega^\bullet M \mid {\rm supp}(\omega)\ \textrm{is compact}\}
  \end{equation}
  be the space of differential forms with compact support.
  The \newterm{integration of forms on $M$}
  \begin{equation}
    \int_M \colon \Omega^\bullet_c M \to \R,
  \end{equation}
  is a linear map having the following properties:
  \begin{enumerate}
    \item
      additive property:
      if $M$ is disconnected and can be partitioned into $M = M_1 \cup M_2$, then for any $\omega \in \Omega^D_c M$,
      \begin{equation}
        \int_M \omega = \int_{M_1} \tr_{M_1} \omega + \int_{M_2} \tr_{M_2} \omega;
      \end{equation}
    \item
      change of variables formula: if $\varphi \colon N \to M$ is an orientation-preserving diffeomorphism (that is, $\varphi^*$ maps orientation forms to orientation forms), then for any $\omega \in \Omega^D_c M$,
      \begin{equation}
        \int_M \omega = \int_N \varphi^* \omega;
      \end{equation}
    \item
      Stokes-Cartan theorem: for any $\omega \in \Omega^{D - 1}_c M$,
      \begin{equation}
        \label{eq:exterior_calculus/stokes_cartan}
        \int_M d \omega = \int_{\partial M} \tr_{\partial M} \omega,
      \end{equation}
      where $\partial M$ is given the induced boundary orientation;
    \item
      (base case) integration is function evaluation for $0$-forms: for any singleton $0$-dimensional manifold $\{x\}$ with orientation $\varepsilon \in \{-1, 1\}$, and any $f \colon \{x\} \to \R$,
      \begin{equation}
        \int_{\{x\}} f = \varepsilon f(x).
      \end{equation}
  \end{enumerate}
  Integration is a dimensionless operation.
\end{definition}
\begin{corollary}
  Let $M$ be an oriented smooth manifold with boundary, $D = \dim M$.
  An important consequence of the graded Leibniz rule, \Cref{eq:exterior_calculus/leibniz}, and the Stokes-Cartan theorem, \Cref{eq:exterior_calculus/stokes_cartan}, is the \newterm{integration by parts formula}:
  for any $p \in \{0, ..., D - 1\},\ \omega \in \Omega^p_c M,\ \eta \in \Omega^{D - p - 1}_c M$,
  \begin{equation}
    \int_M (d \omega \wedge \eta) = \int_{\partial M} \tr_{\partial M}(\omega \wedge \eta) - (-1)^p \int_M (\omega \wedge d \eta).
  \end{equation}
\end{corollary}

\subsection{Metric operations}
\phantom{T}
\label{sec:exterior_calculus/metric}
\begin{definition}
  Let $M$ be a smooth manifold.
  A \newterm{metric tensor} on $M$ is a smooth map
  \begin{equation}
    g \colon \mathfrak{X} M \times \mathfrak{X} M \to \R
  \end{equation}
  that is $(\mathcal{F} M)$-linear, symmetric, and positive definite at each point of $M$.
  The resulting pair $(M, g)$ is called a \newterm{Riemannian manifold}.

  In our applications $M$ is a space manifold and so it is useful to assume that $g$ has physical dimensions $\length^2$.
  Indeed, let $x \in M$ and $v \in T_x M$ be a tangent vector, $\restrict{g}{x}$ be the fibre of $g$ at $x$.
  Then the length of $v$ is defined as $\sqrt{\restrict{g}{x}(v, v)}$, and it is natural for it to have physical dimension $[[v]] \cdot \length$.
  Hence, $g_x$, and $g$ as well, is of physical dimension $\length^2$.
\end{definition}
\begin{remark}
  The metric tensor is a symmetric tensor, and in a local coordinate system can be represented as an expression of the form
  \begin{equation}
    g = \sum_{p, q = 1}^D g_{p, q}\, d x^p \otimes d x^q,\ g_{p, q} = g_{q, p}.
  \end{equation}
  Here $g$ is represented as an element of $(\mathfrak{X} M)^* \otimes (\mathfrak{X} M)^*$.
  As with differential forms, we can take pullbacks with respect to smooth maps $f \colon N \to M$, for some smooth manifold $M$, leading to the \newterm{pullback metric} $f^* g$.
  Importantly, when $N$ is a submanifold of $M$ and $\iota_N \colon N \to M$ is the embedding map, the pullback metric $\tr_N g := \iota_N^* g$ is the \newterm{induced metric} on $N$.
\end{remark}
\begin{discussion}
  The metric tensor $g$ on a smooth manifold $M$ can be extended on $1$-forms as the dual of $g$ and, consequently, on $p$-forms for any $p \in \{0, ..., D\}$.
  The resulting map will be denoted by $g^*_p$,
  \begin{equation}
    g^*_p \colon \Omega^p M \times \Omega^p M \to \R,
  \end{equation}
  and it has a physical dimension $\length^{-2 p}$.
  Indeed, on $1$-forms it has dimension $\length^{-2}$ as it is the dual of $g$.
  On $p$-forms, the dimension when acting on $1$-forms, $\length^{-2}$, is raised to the $p^{\rm th}$ power, hence $\length^{-2 p}$.
  In local coordinates, the following expression holds for $g^*_1$:
  \begin{equation}
    g^*_1 = \sum_{p, q = 1}^D g^{p, q} \frac{\partial}{\partial x^p} \otimes \frac{\partial}{\partial x^p},
  \end{equation}
  where $\set{g^{p, q}}{p, q = 1, ..., D}$ is symmetric and the inverse matrix of $\set{g_{p, q}}{p, q = 1, ..., D}$.
\end{discussion}
\begin{definition}
  Let $(M, g)$ be an oriented Riemannian manifold of dimension $D$.
  The \newterm{volume (pseudo-)form}
  \begin{equation}
    \vol \in \Omega^D M
  \end{equation}
  is the unique form constructed as
  \begin{equation}
    \vol = e^1 \wedge ... \wedge e^D,
  \end{equation}
  where $(e^1, ..., e^D)$ is a positively oriented orthonormal basis of $\Omega^1 M$ with respect to $g^*_1$.
  It is not hard to prove that $\vol$ is independent of the choice of a basis; the proof is omitted here.
  Changing the orientation of $M$ switches the sign of $\vol$, i.e., $\vol$ is a pseudo-form.

  The volume form has a physical dimension $\length^D$, since the basis $1$-forms are of dimensions $\length$.
\end{definition}
\begin{definition}
  Let $(M, g)$ be an oriented compact Riemannian manifold of dimension $D$ with a volume form $\vol$, $0 \leq p \leq D$.
  The \newterm{Hodge star operator on $p$-forms}
  \begin{equation}
    \star_p \colon \Omega^p M \to \Omega^{D - p} M
  \end{equation}
  is the unique operator satisfying the following condition: for any
  $\omega \in \Omega^p M$,
  $\eta \in \Omega^{D - p} M$,
  \begin{equation}
    \label{eq:exterior_calculus/hodge_star}
    g^*_{D - p}(\star_p \omega, \eta)\, \vol = \omega \wedge \eta.
  \end{equation}

  The physical dimension of $\star_p$ is $\length^{D - 2 p}$.
  Indeed, the dimensions of $\omega$ and $\eta$ cancel each other, and so
  \begin{equation}
    [[g^*_{D - p}]]\, [[\star_p]]\, [[\vol]] = 1
    \Rightarrow \length^{-2 (D - p)}\, [[\star_p]]\, \length^D = 1
    \Rightarrow [[\star_p]] = \length^{D - 2 p}.
  \end{equation}
\end{definition}
\begin{proposition}
  Let $(M, g)$ be an oriented Riemannian manifold of dimension $D$ with a volume form $\vol$, $0 \leq p \leq D$.
  Then
  \begin{equation}
    \star_{D - p} \circ \star_p = (-1)^{p (D - p)} {\rm id}_{\Omega^p M}.
  \end{equation}
\end{proposition}
\begin{proof}
  Consider an oriented orthonormal basis $(e^1, ..., e^D)$ of $\Omega^1 M$, $\vol := e^1 \wedge ... \wedge e^D$ be the volume form.
  Let $I = (I_1, ..., I_p)$, $1 \leq I_1 < ... < I_p \leq D$ and $e_I := e^{I_1} \wedge ... \wedge e^{I_p}$.
  Let $J$ be the ordered complement of $I$ (with respect to $(1, ..., D)$),
  and $e^J := e^{J_1} \wedge ... \wedge e^{J_{D - p}}$.
  Then $\star_p e^I = s e^J$ for some $s \in \{-1, 1\}$ (since taking the wedge product gives $e^1 \wedge ... \wedge e^D = \vol$ after reordering, taken into account by $s$), and similarly $\star_{D - p} e^J = t e^I$ for some $t \in \{-1, 1\}$.
  Then
  \begin{equation}
    s\, \vol
    = g^*_{D - p}(\star_p e^I, e^J)\, \vol = e^I \wedge e^J
    = (-1)^{p (D - p)} e^J \wedge e^I
    = (-1)^{p (D - p)} g^*_p(\star_{D - p} e^J, e^I)\, \vol
    = (-1)^{p (D - p)} t\, \vol,
  \end{equation}
  from which it follows that $s = (-1)^{p (D - p)} t$ or $s t = (-1)^{p (D - p)}$.
  Hence,
  \begin{equation}
    \star_{D - p} \star_p e^I = s t e^I = (-1)^{p (D - p)} e^I.
  \end{equation}
  Since $\star$ is $(\mathcal{F} M)$-linear, it follows that for any $\omega \in \Omega^p M$,
  $\star_{D - p} \star_p \omega = (-1)^{p (D - p)} \omega$.
\end{proof}
\begin{proposition}
  Let $(M, g)$ be an oriented Riemannian manifold of dimension $D$,
  $0 \leq p \leq D$,
  $\omega, \eta \in \Omega^p M$.
  Then
  \begin{equation}
    g^*_{D - p}(\star_p \omega, \star_p \eta) = g^*_p(\omega, \eta).
  \end{equation}
\end{proposition}
\begin{proof}
  Let  $\vol$ be the volume form on $\omega$.
  Then
  \begin{equation}
    \begin{split}
    g^*_{D - p}(\star_p \omega, \star_p \eta) \wedge \vol
    & = \omega \wedge \star_p \eta \\
    & = (-1)^{p (D - p)} \star_p \eta \wedge \omega \\
    & = (-1)^{p (D - p)} g^*_p(\star_{D - p} \star_p \eta, \omega) \wedge \vol \\
    & = g^*_p(\eta, \omega) \wedge \vol \\
    & = g^*_p(\omega, \eta) \wedge \vol.
    \end{split}
  \end{equation}
  Since $\vol$ is an $(\mathcal{F} M)$-basis of $\Omega^D M$, we can cancel it, and get the required
  $g^*_{D - p}(\star_p \omega, \star_p \eta) = g^*_p(\omega, \eta)$.
\end{proof}
\begin{definition}
  Let $(M, g)$ be an oriented compact Riemannian manifold of dimension $D$ with a volume form $\vol$.
  The \newterm{inner product} of $p$-forms
  \begin{equation}
    \inner{\cdot}{\cdot}_p \colon \Omega^p M \times \Omega^p M \to \R
  \end{equation}
  is defined by
  \begin{equation}
    \label{eq:exterior_calculus/inner_product}
    \inner{\omega}{\eta}_p = \int_M (g^*_p(\omega, \eta) \wedge \vol).
  \end{equation}
  It is independent of the chosen orientation since switching orientations switches both the signs of integration and the volume form.
  $\inner{\cdot}{\cdot}_p$ has a physical dimension of $\length^{-2p}\, \length^D = \length^{D - 2 p}$.
  The inner product can be restricted to a submanifold $S$ of $M$.
  To avoid confusion, we will write $\inner{\omega}{\eta}_{S, p}$ in that case.
\end{definition}
\begin{remark}
  Let $(M, g)$ be an oriented compact Riemannian manifold of dimension $D$, $0 \leq p \leq D$,
  $\omega \in \Omega^p M$,
  $\eta \in \Omega^{D - p} M$.
  Using the expressions for Hodge star (\Cref{eq:exterior_calculus/hodge_star}) and inner product (\Cref{eq:exterior_calculus/inner_product}) leads to
  \begin{equation}
    \label{eq:exterior_calculus/inner_product_with_hodge_star}
    \inner{\star_p \omega}{\eta}_{D - p}
    = \int_M (g^*_p(\star_p \omega, \eta) \wedge \vol)
    =\int_M (\omega \wedge \eta).
  \end{equation}
  An analogue of this equation will be used as a definition of a discrete Hodge star (\Cref{eq:combinatorial/inner_product_with_hodge_star}) because we will introduce inner product directly without defining a metric tensor.
\end{remark}
\begin{definition}
  Let $M$ be a smooth manifold with boundary, $D = \dim M$, $\tr_{\partial M} \colon \Omega^\bullet M \to \Omega^\bullet(\partial M)$ be the trace operator, $p \in \{0, ..., D\}$.
  Define
  \begin{equation}
    \Omega^p_0 M := \Ker \tr_{\partial M, p} = \set{\omega \in \Omega^p M}{\tr_{\partial M} \omega = 0}.
  \end{equation}
  Since $\tr_{\partial M, D} = 0$, $\Omega^D_0 M = \Omega^D M$.
  Using
  \begin{equation}
    d_p \circ \tr_{\partial M, p}
    = d_p \circ \iota^*_p
    = \iota^*_{p + 1} \circ d_p
    = \tr_{\partial M, p + 1} \circ d_p,
  \end{equation}
  it follows that the image of $d_p$ on $\Omega^p_0 M$ is contained in $\Omega^{p + 1}_0 M$.
  Hence, we can define the operator
  \begin{equation}
    \tilde{d}_p \colon \Omega^p_0 M \to \Omega^{p + 1}_0 M
  \end{equation}
  being the restriction of $d_p$.
\end{definition}
\begin{definition}
  Let $(M, g)$ be an oriented compact Riemannian manifold of dimension $D$,
  $1 \leq p \leq D$.
  The \newterm{codifferential}
  \begin{equation}
    d^\star_p \colon \Omega^p_0 M \to \Omega^{p - 1}_0 M
  \end{equation}
  is the adjoint of $\tilde{d}_{p - 1}$ with respect to the inner products $\restrict{\inner{\cdot}{\cdot}_p}{\Omega^p_0 M}$ and $\restrict{\inner{\cdot}{\cdot}_{p - 1}}{\Omega^{p - 1}_0 M}$,
  that is, for any $\omega \in \Omega^p_0 M,\ \eta \in \Omega^{p - 1}_0 M$,
  \begin{equation}
    \inner{\omega}{\tilde{d}_{p - 1} \eta}_p = \inner{d^\star_p \omega}{\eta}_{p - 1}.
  \end{equation}
  The physical dimension of $d^\star_p$ is $\length^{-2}$.
  Indeed,
  \begin{equation}
    [[\inner{\cdot}{\cdot}_p]] = [[d^\star_p]] \, [[\inner{\cdot}{\cdot}_{p - 1}]]
    \Rightarrow \length^{D - 2 p} = [[d^\star_p]] \, \length^{D - 2 (p - 1)}
    \Rightarrow [[d^\star_p]] = \length^{-2}.
  \end{equation}
\end{definition}
\begin{proposition}
  Let $(M, g)$ be an oriented compact Riemannian manifold of dimension $D$, $p \in \{0, ..., D - 1\}$.
  Then
  \begin{equation}
    \label{eq:exterior_calculus/codifferential_formula}
    d^\star_{D - p} \circ \star_p
    = (-1)^{p + 1} \star_{p + 1} \circ\, \tilde{d}_p.
  \end{equation}
  In particular,
  \begin{equation}
    d^\star_D \circ \star_0 = - \star_1 \circ\, \tilde{d}_0.
  \end{equation}
\end{proposition}
\begin{proof}
  Let $\omega \in \Omega^p_0 M,\ \eta \in \Omega^{D - p - 1}_0 M$.
  Then
  \begin{equation}
    \label{eq:exterior_calculus/codifferential_formula_proof}
    \begin{split}
      \inner{d^\star_{D - p} \star_p \omega}{\eta}_{D - p - 1}
      & = \inner{\star_p \omega}{\tilde{d}_{D - p - 1} \eta}_{D - p} \\
      & = \int_M (\omega \wedge d_{D - p - 1} \eta) \\
      & = \int_{\partial M} (\tr_{\partial M, p}(\omega) \wedge \tr_{\partial M, D - p - 1}(\eta)) + (-1)^{p + 1} \int_{M} (d_p \omega \wedge \eta) \\
      & = \inner{(-1)^{p + 1} \star_{p + 1} \tilde{d}_p \omega}{\eta}_{D - p - 1}.
    \end{split}
  \end{equation}
  The proposition follows since $\omega$ and $\eta$ are arbitrary.
\end{proof}

\section{Exterior calculus variational formulations}
\label{sec:continuum}

\begin{notation}
  Throughout this section $D$ will denote a positive integer, specifying the space dimension.
  For real applications $D$ is expected to be $3$ but in simpler models $D = 1$ or $D = 2$ may suffice.
  Since we are not going to make any special assumptions about the ambient space, we will work with arbitrary $D$.
  Moreover, many of the operators in this article have signs dependent on the space dimension, and so with general $D$ we could capture these signs in a unified way.

  Our model will be built on a bounded open body $M$, with fixed orientation, and embedded in Euclidean space $\R^D$.
  We assume that we have a scalar extensive quantity $A$, representing the \newterm{amount} of conserved quantity on $M$, with a given physical dimension $\amount$; examples include mass, energy, charge, volume.

  We will also fix an initial time $t_0 [\time]$ and the time interval $I = [t_0, \infty)$.

  Time-dependent quantities will be mostly represented by objects whose domain is
  $C^\infty(I, S_1) \times ... \times C^\infty(I, S_n)$,
  where $S_1, ..., S_n$ are vector spaces representing the spatial part.
  Previously introduced operators (wedge product, exterior derivative, pullback/trace, integration, metric tensor, inner product, Hodge star, codifferential) will be applied pointwise to the spatial parts $S_1, ..., S_n$.
  In other words, if $S$ is some other vector space, and we have an operator
  $H \colon S_1 \times ... \times S_n \to S$
  (e.g., exterior derivative), we overload $H$ to a map
  \begin{equation}
    H \colon \mathcal{C}^\infty(I, S_1) \times ... \times \mathcal{C}^\infty(I, S_n) \to \mathcal{C}^\infty(I, S),\
    H(\omega_1, ..., \omega_n)(t) := H(\omega_1(t), ..., \omega_n(t)) \in S.
  \end{equation}
  On the other hand, partial derivative with respect to time is overloaded to an operator on temporal objects, i.e.,
  \begin{equation}
    \frac{\partial}{\partial t} \colon \mathcal{C}^\infty(I, S) \to \mathcal{C}^\infty(I, S),\
    \frac{\partial \omega}{\partial t}(t)  := \lim_{\tau \to 0} \frac{\omega(t + \tau) - \omega(t)}{\tau} \in S.
  \end{equation}
\end{notation}
\begin{definition}
  The \newterm{amount pseudo-form}
  \begin{equation}
    \label{eq:continuum/amount_definition}
    Q [\amount] \in \mathcal{C}^\infty(I, \Omega^D M)
  \end{equation}
  is characterised as follows: for any $D$-dimensional subregion $V$ of $M$, and any time interval $[t_1, t_2] \subset I$,
  \begin{equation}
    \label{eq:continuum/amount_difference}
    \textrm{``amount difference of $A$ inside $V$ in $[t_1, t_2]$''} = \int_V Q(t_2) - \int_V Q(t_1).
  \end{equation}
\end{definition}
\begin{definition}
  \label{def:continuum/flow_rate}
  The \newterm{flow rate pseudo-form}
  \begin{equation}
    \label{eq:continuum/flow_rate_definition}
    q [\amount \time^{-1}] \in \mathcal{C}^\infty(I, \Omega^{D - 1} M)
  \end{equation}
  is characterised as follows.
  Let $[t_1, t_2] \subset I$ be a time interval, $S$ be a $(D - 1)$-dimensional hypersurface on $M$.
  Consider a subregion $V$ of $M$ partitioned into two parts $V_{+}$ and $V_{-}$ sharing $S$ as a common boundary,
  with relative orientations
  $\varepsilon(V_{+}, S) = 1$ and $\varepsilon(V_{-}, S) = -1$.
  Define the \newterm{net flow} of the amount $A$ through $S$ in $[t_1, t_2]$ to be be the difference between the flow from $V_{+}$ to $V_{-}$ in $[t_1, t_2]$ and the flow from $V_{-}$ to $V_{+}$ in $[t_1, t_2]$.
  Then
  \begin{equation}
    \label{eq:continuum/net_flow_through_surface}
    \textrm{``net flow of $A$ through $S$ in $[t_1, t_2]$''} = \int_{t_1}^{t_2} \int_S q\, d t.
  \end{equation}
\end{definition}
\begin{remark}
  In \Cref{eq:continuum/net_flow_through_surface}
  the perspective is from the point of view of the hypersurface $S$.
  However, if we view it through the eyes of $V_{+}$, the aforementioned term is the \newterm{net outflow} of $A$ from $V_{+}$ to $V_{-}$ through $S$ in $[t_1, t_2]$,
  while from the perspective of $V_{-}$ it is the \newterm{net inflow} of $A$ in $V_{-}$ from $V_{+}$ through $S$ in $[t_1, t_2]$.
\end{remark}
\begin{proposition}
  Let $V$ be a subregion of $M$, $[t_1, t_2] \subset I$.
  Then
  \begin{equation}
    \label{eq:continuum/net_boundary_outflow}
    \textrm{``net outflow of $A$ from $V$ through its boundary in $[t_1, t_2]$''}
    = \int_{t_1}^{t_2} \int_{\partial V} q\, d t.
  \end{equation}
\end{proposition}
\begin{proof}
  By definition, the boundary $\partial V$ has the induced boundary orientation, i.e.,
  $\varepsilon(V, \partial V) = 1,\ \varepsilon(M \setminus \overline{V}, \partial V) = -1$.
  The net outflow of $A$ from $V$ through its boundary in $[t_1, t_2]$, hence, equals to the net flow of $A$ through $\partial V$ in $[t_1, t_2]$,
  which by \Cref{def:continuum/flow_rate} is
  $\int_{t_1}^{t_2} \int_{\partial V} q\, d t$.
\end{proof}
\begin{definition}
  The \newterm{internal production rate pseudo-form}
  \begin{equation}
    \label{eq:continuum/internal_production_rate_definition}
    f [\amount \time^{-1}] \in \mathcal{C}^\infty(I, \Omega^D M)
  \end{equation}
  is characterised as follows: for any $D$-dimensional subregion $V$ of $M$, and any time interval $[t_1, t_2] \subset I$,
  \begin{equation}
    \label{eq:continuum/internal_production}
    \textrm{``internal production of $A$ inside $V$ in $[t_1, t_2]$''} = \int_{t_1}^{t_2} \int_V f\, d t.
  \end{equation}
\end{definition}
\begin{discussion}[Derivation of conservation law for transport phenomena]
  The conservation law for an extensive quantity $A$ (e.g., energy) has the following intuitive formulation:
  for any $D$-dimensional subregion $V$ of $M$, and any time interval $[t_1, t_2] \subset I$,
  \begin{equation}
    \begin{split}
      & \text{``amount difference of $A$ inside $V$ between $t_2$ and $t_1$''} \\
      & \quad = \text{``net production of $A$ inside $V$ in $[t_1, t_2]$''} \\
      & \quad \quad - \text{``net outflow of $A$ through $\partial V$ in $[t_1, t_2]$''}.
    \end{split}
  \end{equation}
  Having defined variables for denoting these quantities, we arrive at the equation
  \begin{equation}
    \label{eq:continuum/conservation_law_integral_form}
    \int_V Q(t_2) - \int_V Q(t_1) = \int_{t_1}^{t_2} \int_V f \, dt - \int_{t_1}^{t_2} \int_{\partial V} q\, d t.
  \end{equation}
  Using the Newton-Leibniz theorem for the first integral and the Stokes-Cartan theorem
  (\Cref{eq:exterior_calculus/stokes_cartan})
  for the third integral, we get:
  \begin{equation}
    \int_{t_1}^{t_2} \left(\int_V \frac{\partial Q}{\partial t}\right)\, d t = \int_{t_1}^{t_2} \left(\int_V f\right)\, d t - \int_{t_1}^{t_2} \left(\int_V\, d_{D - 1} q\right)\, d t.
  \end{equation}
  Since $V$ and $[t_1, t_2]$ are arbitrary, we can remove the integrals over them.
  Thus, we arrive at the differential formulation of the conservation law:
  \begin{equation}
    \label{eq:continuum/conservation_law_differential_form}
    \frac{\partial Q}{\partial t} = f - d_{D - 1} q.
  \end{equation}
\end{discussion}
\begin{discussion}[Neumann boundary conditions]
  In formulations of initial-boundary value problems for transport phenomena, one type of boundary condition is the imposition of flow rate on part of the boundary, i.e., the prescription of flow rate between subregions touching the boundary and the external world.
  To formulate such constraint, let $\Gamma_{\Neumann} \subseteq M$ be the part of the boundary where the flow rate is imposed,
  \begin{equation}
    \label{eq:continuum/prescribed_neumann_condition_definition}
    g_{\Neumann} \colon I \to \Omega^{D - 1}_{\Neumann}
  \end{equation}
  be the \newterm{prescribed flow rate pseudo-form}.
  Then the Neumann boundary condition reads as
  \begin{equation}
    \label{eq:continuum/neumann_boundary_condition}
    \tr_{\Gamma_{\Neumann}, D - 1}(q) = g_{\Neumann}.
  \end{equation}
\end{discussion}
\begin{discussion}[Dirichlet boundary conditions and potential]
  The other common boundary condition, Dirichlet boundary condition, is related to the prescription of an intensive counterpart
  \begin{equation}
    \label{eq:continuum/potential_definition}
    u [\potential] \in \mathcal{C}^\infty(I, \Omega^0 M)
  \end{equation}
  of the extensive quantity $Q$, called \newterm{potential}.
  Precisely, if $\Gamma_{\Dirichlet}$ is the other part of the boundary of $M$, the \newterm{Dirichlet boundary}, so that
  \begin{equation}
    \label{eq:continuum/boundary_decomposition}
    \partial M = \Gamma_{\Dirichlet} \cup \Gamma_{\Neumann},
  \end{equation}
  and
  \begin{equation}
    \label{eq:continuum/prescribed_dirichlet_condition_definition}
    g_{\Dirichlet} \in \mathcal{C}^\infty(I, \Omega^0 \Gamma_{\Dirichlet}),
  \end{equation}
  then the \newterm{Dirichlet boundary condition} is described by the equation
  \begin{equation}
    \label{eq:continuum/dirichlet_boundary_condition}
    \tr_{\Gamma_{\Dirichlet}, 0}(u) = g_{\Dirichlet}.
  \end{equation}
\end{discussion}
\begin{remark}
  Potentials for mass (mass transfer), energy (heat transfer), charge (charge transport), volume (fluid flow through porous media) are concentration, temperature, electric potential, and pressure respectively.
\end{remark}
\begin{discussion}[Relation between potential and amount]
  Let
  \begin{equation}
    \label{eq:continuum/capacity_definition}
    \pi [\amount \potential^{-1} \length^{-D}] \colon \mathcal{C}^\infty(I, \Omega^D M) \to \mathcal{C}^\infty(I, \Omega^D M)
  \end{equation}
  be the \newterm{volumetric capacity} of the material.
  Define the \newterm{dual volumetric capacity}
  \begin{equation}
    \label{eq:continuum/dual_capacity_definition}
    \widetilde{\pi} [\amount \potential^{-1} \length^{-D}] \colon \mathcal{C}^\infty(I, \Omega^0 M) \to \mathcal{C}^\infty(I, \Omega^0 M)
  \end{equation}
  by
  \begin{equation}
    \label{eq:continuum/dual_capacity_formula}
    \widetilde{\pi}
    = \star_0^{-1} \circ \pi \circ \star_0
    = \star_D \circ \pi \circ \star_0.
  \end{equation}
  The potential $u$ and the amount $Q$ are related by the equation
  \begin{equation}
    \label{eq:continuum/potential_to_amount}
    \frac{\partial Q}{\partial t}
    = \pi \star_0 \frac{\partial u}{\partial t}
    = \star_0 \widetilde{\pi} \frac{\partial u}{\partial t}.
  \end{equation}
  For simplicity, assume that the volumetric capacity is time-independent.
  Then,
  \begin{equation}
    Q(t)
    = Q(t_0) + \int_{t_0}^t \frac{\partial Q}{\partial \tau}\, d \tau
    = Q(t_0) + \int_{t_0}^t \star_0
      \widetilde{\pi}\left(\frac{\partial u}{\partial \tau}\, d \tau\right)
    = Q(t_0) - \star_0 (\widetilde{\pi} u)(t_0) + \star_0 (\widetilde{\pi} u)(t)
    = \star_0 (\widetilde{\pi} u)(t) + S_0,
  \end{equation}
  where we have denoted
  \begin{equation}
    \label{eq:continuum/amount_potential_constant}
    S_0 := Q(t_0) - \star_0 (\widetilde{\pi} u)(t_0) \in \Omega^D M.
  \end{equation}
  For future use we also define
  \begin{equation}
    \label{eq:continuum/dual_amount_potential_constant}
    \widetilde{S_0} := \star_D S_0 \in \Omega^0 M
    = \star_D Q(t_0) - (\widetilde{\pi} u)(t_0) .
  \end{equation}
\end{discussion}
\begin{remark}
  In mass transfer the potential is the mass density, so no volumetric capacity is present (it is given by the number $1$).
  In heat transfer, charge transport and fluid flow through porous media the volumetric capacities are volumetric heat capacity, volumetric capacitance, and compressibility respectively.
\end{remark}
\begin{discussion}[Drivers for flow]
  The most common transport mechanisms of physical properties are diffusion and advection, which we will assume to be present in our model.
  \begin{enumerate}
    \item
      Diffusion is the spontaneous transport of $A$ from regions with higher potential to neighbouring regions with lower potential. The driver for diffusive flow is the potential difference.
      The \newterm{diffusive flow rate} is expressed by the pseudo-form
      \begin{equation}
        \label{eq:continuum/diffusive_flow_rate_definition}
        q_{\diffusive} [\amount \time^{-1}] \in \mathcal{C}^\infty(I, \Omega^{D - 1} M).
      \end{equation}
    \item
      Advection is the bulk transport of $A$ caused by the some prescribed volumetric flow rate, which is the driver for advective flow.
      The \newterm{advective flow rate} is expressed by the pseudo-form
      \begin{equation}
        \label{eq:continuum/advective_flow_rate_definition}
        q_{\advective} [\amount \time^{-1}] \in \mathcal{C}^\infty(I, \Omega^{D - 1} M).
      \end{equation}
  \end{enumerate}
  The total flow rate is the sum of diffusive and advective flow rates, i.e.,
  \begin{equation}
    \label{eq:continuum/flow_rate_decomposition}
    q = q_{\diffusive} + q_{\advective}.
  \end{equation}
  In the following paragraphs we will relate $q_{\diffusive}$ and $q_{\advective}$ to the other model's variables.
\end{discussion}
\begin{discussion}[Constitutive law for diffusive flow rate]
  Let
  \begin{equation}
    \label{eq:continuum/conductivity_definition}
    \kappa [\amount \potential^{-1} \length^{2 - D} \time^{-1}] \colon \mathcal{C}^\infty(I, \Omega^{D - 1} M) \to \mathcal{C}^\infty(I, \Omega^{D - 1} M)
  \end{equation}
  be the \newterm{conductivity} of the material, and
  \begin{equation}
    \label{eq:continuum/dual_potential_definition}
    \tilde{u} [\potential \length^{D}] \in \mathcal{C}^\infty(I, \Omega^D M), \quad
    \tilde{u} = \star_0 u
  \end{equation}
  be the \newterm{dual potential}.
  Then $q_{\diffusive}$ is calculated via the \newterm{constitutive law}
  \begin{equation}
    \label{eq:continuum/constitutive_law}
    q_{\diffusive}
    = \kappa d^\star_D \tilde{u}
    = \kappa d^\star_D \star_0 u.
  \end{equation}
\end{discussion}
\begin{discussion}
  Define the \newterm{dual conductivity}
  \begin{equation}
    \label{eq:continuum/dual_conductivity_definition}
   \widetilde{\kappa} [\amount \potential^{-1} \length^{2 - D} \time^{-1}] \colon \mathcal{C}^\infty(I, \Omega^1 M) \to \mathcal{C}^\infty(I, \Omega^1 M)
  \end{equation}
  by
  \begin{equation}
    \label{eq:continuum/dual_conductivity_formula}
   \widetilde{\kappa}
    = \star_1^{-1} \circ \kappa \circ \star_1
    = (-1)^{D - 1} \star_{D - 1} \circ \kappa \circ \star_1.
  \end{equation}
  This means that
  \begin{equation}
    q_\diffusive
    = \kappa\, d^\star_D \star_0 u
    = - \kappa \star_1 d_0 u
    = - \star_1\widetilde{\kappa}\, d_0 u.
  \end{equation}
\end{discussion}
\begin{remark}
  The conductivities for mass transfer, heat transfer, charge transport, and fluid flow through porous media are mass diffusivity, thermal conductivity, electrical conductivity, and hydraulic conductivity respectively.
  The respective constitutive laws are Fick's law of diffusion, Fourier's law, Ohm's law, and Darcy's law.

  Discussed transport phenomena, their respective material properties and constitutive laws, and physical dimensions are summarised in \Cref{tab:transport_phenomena}.
\end{remark}
\begin{table}[!ht]
  \caption{Examples of transport phenomena}
  \label{tab:transport_phenomena}
  \centering
  \begin{tabular}{|c|c|c|c|c|c|}
    \hline
    {\bf Phenomenon}
    & \makecell{{\bf Amount} \\ $[\amount]$}
    & \makecell{{\bf Potential} \\ $[\potential]$}
    & \makecell{{\bf Volumetric capacity} \\ $[\amount \potential^{-1} \length^{-D}]$}
    & \makecell{{\bf Conductivity} \\ $[\amount \potential^{-1} \length^{2 - D} \time^{-1}]$}
    & \makecell{{\bf Constitutive} \\ {\bf law}} \topStrut \\[2pt]
    \hline
    \hline
    Mass transfer
    & \makecell{Mass \\ $[\mass]$}
    & \makecell{Mass density \\ $[\mass \length^{-D}]$}
    & \makecell{1 \\ $[1]$}
    & \makecell{Mass diffusivity \\ $[\length^2 \time^{-1}]$}
    & Fick's law \topStrut \\[2pt]
    \hline
    Heat transfer
    & \makecell{Heat energy \\ $[\mass \length^2 \time^{-2}]$}
    & \makecell{Temperature \\ $[\temperature]$}
    & \makecell{Volumetric heat capacity \\ $[\mass \length^{2 - D} \time^{-2}\temperature^{-1}]$}
    & \makecell{Thermal conductivity \\ $[\mass \length^{4 - D} \time^{-3} \temperature^{-1}]$}
    & Fourier's law \topStrut \\[2pt]
    \hline
    Charge transport
    & \makecell{Charge \\ $[\charge]$}
    & \makecell{Electric potential \\ $[\mass \length^2 \time^{-2} \charge^{-1}]$}
    & \makecell{Volumetric capacitance \\$[\mass^{-1} \length^{-2 - D} \time^2 \charge^2]$}
    & \makecell{Electrical conductivity \\ $[\mass^{-1} \length^{-D} \time \charge^2]$}
    & Ohm's law \topStrut \\[2pt]
    \hline
    \makecell{Fluid flow through \\ porous media}
    & \makecell{Volume \\ $[\length^D]$}
    & \makecell{Pressure \\ $[\mass \length^{2 - D} \time^{-2}]$}
    & \makecell{Compressibility \\ $[\mass^{-1} \length^{D - 2} \time^2]$}
    & \makecell{Hydraulic conductivity \\$[\mass^{-1} \length^D \time]$}
    & Darcy's law \topStrut \\[2pt]
    \hline
  \end{tabular}
\end{table}
\begin{discussion}[Description of advective flow rate]
  Let
  \begin{equation}
    \label{eq:continuum/volumetric_flow_rate_definition}
    v [\length^D \time^{-1}] \in \mathcal{C}^\infty(I, \Omega^{D - 1} M)
  \end{equation}
  be the \newterm{volumetric flow rate pseudo-form}, which is a prescribed quantity, e.g., calculated from a volume transport problem.
  It is characterised by the following property: for any two subregions $V_{+}$ and $V_{-}$ with common boundary $S$ (and respective positive and negative relative orientations with $S$), and any time interval $[t_1, t_2] \subset I$,
  \begin{equation}
    \label{eq:continuum/volumetric_flow_rate_formula}
    \textrm{``net volume outflow from $V_{+}$ to $V_{-}$ through $S$ in $[t_1, t_2]$''} = \int_{t_1}^{t_2} \int_S v\, d t.
  \end{equation}
  The advective flow rate is, roughly speaking, the amount density multiplied by the volumetric flow rate.
  In exterior calculus terms,
  \begin{equation}
    \label{eq:continuum/advective_flow_rate_formula}
    q_{\advective} = (\star_D Q) \wedge v.
  \end{equation}
\end{discussion}
\begin{discussion}[Initial conditions]
  Since we are describing a transient phenomenon, we need initial conditions.
  We assume that the initial potential is prescribed to some initial value
  \begin{equation}
    \label{eq:continuum/initial_potential_definition}
    u_0 [\potential] \in \Omega^0 M,
  \end{equation}
  i.e.,
  \begin{equation}
    \label{eq:continuum/initial_condition_potential}
    u(t_0) = u_0.
  \end{equation}
  Similarly, if advection is present, we will need to know the initial amount
  \begin{equation}
    \label{eq:continuum/initial_amount_definition}
    Q_0 [\potential] \in \Omega^D M,
  \end{equation}
  with the initial condition
  \begin{equation}
    \label{eq:continuum/initial_condition_amount}
    Q_0 = Q(t_0).
  \end{equation}
\end{discussion}
\begin{notation}[Parameters participating in the exterior calculus model for transport phenomena]
  \label{notation:continuum/parameters}
  Before imposing the (strong) model, let us summarise all the parameters that will be used in all subsequent reformulations.
  Until the end of this section let:
  \begin{itemize}
    \item
      $D$ be a positive integer (space dimension);
    \item
      $M$ be a $D$-dimensional compact oriented smooth manifold with boundary;
    \item
      $t_0 [\time] \in \R$ be the initial time;
    \item
      $I := [t_0, \infty)$;
    \item
      $\Gamma_{\Dirichlet}, \Gamma_{\Neumann}$ form a partition of $\partial M$ into Dirichlet and Neumann boundary
      (\Cref{eq:continuum/boundary_decomposition}).
  \end{itemize}
  Input parameters, initial and boundary conditions are given in \Cref{tab:continuum/parameters}.
  The unknowns are given in \Cref{tab:continuum/unknowns}.
\end{notation}
\begin{table}[!ht]
  \caption{Transport phenomena parameters in the exterior calculus formulation}
  \label{tab:continuum/parameters}
  \centering
  \begin{tabular}{|l|l|l|l|l|l|}
    \hline
    Quantity
    & Symbol
    & Domain
    & Dimension
    & Pseudo-object?
    & Reference \topStrut \\[2pt]
    \hline
    \hline
    Initial potential
    & $u_0$
    & $\Omega^0 M$
    & $\potential$
    & No
    & \cref{eq:continuum/initial_potential_definition} \topStrut \\[2pt]
    \hline
    Initial amount
    & $Q_0$
    & $\Omega^D M$
    & $\amount$
    & Yes
    & \cref{eq:continuum/initial_amount_definition} \topStrut \\[2pt]
    \hline
    Internal production rate
    & $f$
    & $\mathcal{C}^\infty(I, \Omega^D M)$
    & $\amount \time^{-1}$
    & Yes
    & \cref{eq:continuum/internal_production_rate_definition} \topStrut \\[2pt]
    \hline
    Prescribed flow rate
    & $g_\Neumann$
    & $\mathcal{C}^\infty(I, \Omega^{D - 1} \Gamma_\Neumann)$
    & $\amount \time^{-1}$
    & Yes
    & \cref{eq:continuum/prescribed_neumann_condition_definition} \topStrut \\[2pt]
    \hline
    Prescribed potential
    & $g_\Dirichlet$
    & $\mathcal{C}^\infty(I, \Omega^0 \Gamma_\Dirichlet)$
    & $\potential$
    & No
    & \cref{eq:continuum/prescribed_dirichlet_condition_definition} \topStrut \\[2pt]
    \hline
    Volumetric flow rate
    & $v$
    & $\mathcal{C}^\infty(I, \Omega^{D - 1} M)$
    & $\length^D \time^{-1}$
    & Yes
    & \cref{eq:continuum/volumetric_flow_rate_definition} \topStrut \\[2pt]
    \hline
    Capacity
    & $\pi$
    & $\mathcal{C}^\infty(I, \Omega^D M) \to \mathcal{C}^\infty(I, \Omega^D M)$
    & $\amount \potential^{-1} \length^{-D}$
    & No
    & \cref{eq:continuum/dual_capacity_definition} \topStrut \\[2pt]
    \hline
    Dual capacity
    & $\widetilde{\pi}$
    & $\mathcal{C}^\infty(I, \Omega^0 M) \to \mathcal{C}^\infty(I, \Omega^0 M)$
    & $\amount \potential^{-1} \length^{-D}$
    & No
    & \cref{eq:continuum/capacity_definition} \topStrut \\[2pt]
    \hline
    Conductivity
    & $\kappa$
    & $\mathcal{C}^\infty(I, \Omega^{D - 1} M) \to \mathcal{C}^\infty(I, \Omega^{D - 1} M)$
    & $\amount \potential^{-1} \length^{2 - D} \time^{-1}$
    & No
    & \cref{eq:continuum/conductivity_definition} \topStrut \\[2pt]
    \hline
    Dual conductivity
    & $\widetilde{\kappa}$
    & $\mathcal{C}^\infty(I, \Omega^1 M) \to \mathcal{C}^\infty(I, \Omega^1 M)$
    & $\amount \potential^{-1} \length^{2 - D} \time^{-1}$
    & No
    & \cref{eq:continuum/dual_conductivity_definition} \topStrut \\[2pt]
    \hline
  \end{tabular}
\end{table}
\begin{table}[!ht]
  \caption{Transport phenomena unknowns in the exterior calculus formulation}
  \label{tab:continuum/unknowns}
  \centering
  \begin{tabular}{|l|l|l|l|l|l|}
    \hline
    Quantity
    & Symbol
    & Domain
    & Dimension
    & Pseudo-object?
    & Reference \topStrut \\[2pt]
    \hline
    \hline
    Amount
    & $Q$
    & $\mathcal{C}^\infty(I, \Omega^D M)$
    & $\amount$
    & Yes
    & \Cref{eq:continuum/amount_definition} \topStrut \\[2pt]
    \hline
    Flow rate
    & $q$
    & $\mathcal{C}^\infty(I, \Omega^{D - 1} M)$
    & $\amount \time^{-1}$
    & Yes
    & \Cref{eq:continuum/flow_rate_definition} \topStrut \\[2pt]
    \hline
    Diffusive flow rate
    & $q_\diffusive$
    & $\mathcal{C}^\infty(I, \Omega^{D - 1} M)$
    & $\amount \time^{-1}$
    & Yes
    & \Cref{eq:continuum/diffusive_flow_rate_definition} \topStrut \\[2pt]
    \hline
    Advective flow rate
    & $q_\advective$
    & $\mathcal{C}^\infty(I, \Omega^{D - 1} M)$
    & $\amount \time^{-1}$
    & Yes
    & \Cref{eq:continuum/advective_flow_rate_definition} \topStrut \\[2pt]
    \hline
    Potential
    & $u$
    & $\mathcal{C}^\infty(I, \Omega^0 M)$
    & $\potential$
    & No
    & \Cref{eq:continuum/potential_definition} \topStrut \\[2pt]
    \hline
  \end{tabular}
\end{table}
\begin{formulation}[Exterior calculus transient model for transport phenomena]
  Combining everything described in this section, under the assumptions of \Cref{notation:continuum/parameters}, we arrive at the model given in \Cref{tab:continuum/equations}.
\end{formulation}
\begin{table}[!ht]
  \caption{Governing equations for transport phenomena in the exterior calculus formulation}
  \label{tab:continuum/equations}
  \centering
  \begin{tabular}{|ll|l|l|l|}
    \hline
    Equation
    &
    & Dimension
    & Law
    & Reference \topStrut \\[2pt]
    \hline
    \hline
    $\frac{\partial Q}{\partial t}$
    & $= f - d q$
    & $\amount \time^{-1}$
    & Conservation law
    & \Cref{eq:continuum/conservation_law_differential_form} \topStrut \\[2pt]
    \hline
    $\frac{\partial Q}{\partial t}$
    & $= \star_0 \widetilde{\pi} \frac{\partial u}{\partial t} = \pi \star_0 \frac{\partial u}{\partial t}$
    & $\amount \time^{-1}$
    & Relation between potential and amount
    & \Cref{eq:continuum/potential_to_amount} \topStrut \\[2pt]
    \hline
    $q$
    & $= q_\diffusive + q_\advective$
    & $\amount \time^{-1}$
    & Flow rate decomposition
    & \Cref{eq:continuum/flow_rate_decomposition} \topStrut \\[2pt]
    \hline
    $q_\diffusive$
    & $= \kappa d^\star_D \star_0 u$
    & $\amount \time^{-1}$
    & Constitutive law
    & \Cref{eq:continuum/constitutive_law} \topStrut \\[2pt]
    \hline
    $q_\advective$
    & $= (\star_D Q) \wedge v$
    & $\amount \time^{-1}$
    & Advective flow rate formula
    & \Cref{eq:continuum/advective_flow_rate_formula} \topStrut \\[2pt]
    \hline
    $\tr_{\Gamma_{\Dirichlet, 0}}(u)$
    & $= g_\Dirichlet$
    & $\potential$
    & Dirichlet boundary condition
    & \Cref{eq:continuum/dirichlet_boundary_condition} \topStrut \\[2pt]
    \hline
    $\tr_{\Gamma_{\Neumann, D - 1}}(q)$
    & $= g_\Neumann$
    & $\amount \time^{-1}$
    & Neumann boundary condition
    & \Cref{eq:continuum/neumann_boundary_condition} \topStrut \\[2pt]
    \hline
    $u(t_0)$
    & $= u_0$
    & $\potential$
    & Initial condition for potential
    & \Cref{eq:continuum/initial_condition_potential} \topStrut \\[2pt]
    \hline
    $Q(t_0)$
    & $= Q_0$
    & $\amount$
    & Initial condition for amount
    & \Cref{eq:continuum/initial_condition_amount} \topStrut \\[2pt]
    \hline
  \end{tabular}
\end{table}

\subsection{Primal weak formulations of transport phenomena}
\label{sec:continuum/primal}

\begin{discussion}
  Consider a \newterm{test function} (in the jargon of finite elements)
  \begin{equation}
    w \in \Ker \tr_{\Gamma_{\Dirichlet}, 0}
    = \set{\omega \in \Omega^0 M}{\tr_{\Gamma_{\Dirichlet}, 0} \omega = 0}
    = \set{\omega \in \mathcal{F} M}{\restrict{\omega}{\Gamma_\Dirichlet} = 0}.
  \end{equation}
  On the one hand, multiplying $w$ with the conservation law (\Cref{eq:continuum/conservation_law_differential_form}) and integrating over $M$ gives:
  \begin{equation}
    \label{eq:continuum/primal_weak/conservation_law}
    \begin{split}
      \int_M w \wedge \frac{\partial Q}{\partial t}
      & = - \int_M (w \wedge d_{D - 1} q) + \int_M (w \wedge f) \\
      & = - \int_{\partial M} \tr_{\partial M, D - 1} (w \wedge q)
        + \int_M (d_0 w \wedge q)
        + \int_M (w \wedge f) \\
      & =
        - \int_{\Gamma_{\Neumann}}
          (\tr_{\Gamma_{\Neumann}, 0} w \wedge \tr_{\Gamma_{\Neumann}, D - 1} q)
        + \int_M (d_0 w \wedge (- \star_1 \widetilde{\kappa} d_0 u + \star_D Q \wedge v))
        + \int_M (w \wedge f) \\
      & = - \int_{\Gamma_{\Neumann}} (\tr_{\Gamma_{\Neumann}, 0} w \wedge g_{\Neumann})
        - \inner{d_0 w}{\widetilde{\kappa} d_0 u}_{M, 1}
        + \int_M (d_0 w \wedge ((\widetilde{S_0} + \widetilde{\pi} u) \wedge v))
        + \int_M (w \wedge f) \\
      & = - \int_{\Gamma_{\Neumann}} (\tr_{\Gamma_{\Neumann}, 0} w \wedge g_{\Neumann})
        - \inner{d_0 w}{\widetilde{\kappa} d_0 u}_{M, 1}
        + \int_M (d_0 w \wedge (\widetilde{S_0} \wedge v))
        + \int_M (d_0 w \wedge (\widetilde{\pi} u \wedge v))
        + \int_M (w \wedge f).
    \end{split}
  \end{equation}
  On the other hand, multiplying $w$
  with the relation of amount and potential (\Cref{eq:continuum/potential_to_amount}) and integrating over $M$ gives:
  \begin{equation}
    \label{eq:continuum/primal_weak/potential_to_amount}
    \int_M w \wedge \frac{\partial Q}{\partial t}
    = \int_M w \wedge
      \left(\star_0 \widetilde{\pi} \frac{\partial u}{\partial t}\right)
    =  \inner{w}{\widetilde{\pi} \frac{\partial u}{\partial t}}_{M, 0}.
  \end{equation}
  Equating the right-hand sides of
  \Cref{eq:continuum/primal_weak/conservation_law}
  and \Cref{eq:continuum/primal_weak/potential_to_amount},
  imposing the Dirichlet boundary condition
  (\Cref{eq:continuum/dirichlet_boundary_condition})
  and the initial conditions (\Cref{eq:continuum/initial_condition_potential}
  and \Cref{eq:continuum/initial_condition_amount}),
  leads to the primal weak formulation presented below.
\end{discussion}
\begin{formulation}[Exterior calculus transient primal weak model for transport phenomena]
  \label{formulation:continuum/primal_weak/transient}
  Under the assumptions of \Cref{notation:continuum/parameters}
  define the following operators:
  \begin{subequations}
    \label{eq:continuum/primal_weak/operators}
    \begin{alignat}{3}
      & A_{\diffusive} \colon \Omega^0 M \times \mathcal{C}^\infty(I, \Omega^0 M) \to \R, \quad
      && A_{\diffusive}(w, u) := \inner{d_0 w}{\widetilde{\kappa} d_0 u}_{M, 1} \qquad
      && [\amount \time^{-1} \potential^{-1}], \\
      & A_{\advective} \colon \Omega^0 M \times \mathcal{C}^\infty(I, \Omega^0 M) \to \R, \quad
      && A_{\advective}(w, u) := \int_M (d_0 w \wedge (\widetilde{\pi} u \wedge v)) \qquad
      && [\amount \time^{-1} \potential^{-1}], \\
      & B \colon \Omega^0 M \times \mathcal{C}^\infty(I, \Omega^0 M) \to \R, \quad
      && B(w, u) := \inner{w}{\widetilde{\pi} u}_{M, 0} \qquad
      && [\amount \potential^{-1}], \\
      & G \colon \Omega^0 M \to \R, \quad
      && G(w)
        := \int_{\Gamma_{\Neumann}} (\tr_{\Gamma_{\Neumann}} w \wedge g_{\Neumann}) \qquad
      && [\amount \time^{-1}], \\
      & F_1 \colon \Omega^0 M \to \R, \quad
      && F_1(w) := \int_M (w \wedge f) \qquad
      && [\amount \time^{-1}], \\
      & F_2 \colon \Omega^0 M \to \R, \quad
      && F_2(w) := \int_M (d_0 w \wedge (\widetilde{S_0} \wedge v)) \qquad
      && [\amount \time^{-1}].
    \end{alignat}
  \end{subequations}
  The unknown variable is the potential
  $u [\potential] \in \mathcal{C}^\infty(I, \Omega^0 M)$.
  We are solving the following system for $u$:
  \begin{subequations}
    \label{eq:continuum/primal_weak/transient_formulation}
    \begin{alignat}{4}
      & \forall w [\potential] \in \Ker \tr_{\Gamma_{\Dirichlet}, 0}, \quad
      && B(w, \frac{\partial u} {\partial t}) + A_{\diffusive}(w, u) - A_{\advective}(w, u)
      && = F_1(w) + F_2(w) - G(w) \qquad
      && [\amount \time^{-1} \potential], \\
      &
      && \tr_{\Gamma_{\Dirichlet}, 0}(u)
      && = g_{\Dirichlet} \qquad
      && [\potential], \\
      &
      && u(t_0)
      && = u_0 \qquad
      && [\potential].
    \end{alignat}
  \end{subequations}
  The flow rate $q [\amount \time^{-1}] \in \mathcal{C}^\infty(I, \Omega^{D - 1} M)$
  is calculated in the post-processing phase by the formula
  \begin{equation}
    \label{eq:continuum/primal_weak/transient_post_processing}
    q(t, x) =
    \begin{cases}
      (-\star_1\widetilde{\kappa} d_0 u + (\widetilde{S_0} + \widetilde{\pi} u) \wedge v)(t, x),
        & x \notin \Gamma_\Neumann \\
      g_{\Neumann}(t, x), & x \in \Gamma_{\Neumann}
    \end{cases},\ t \in I.
  \end{equation}
\end{formulation}
\begin{formulation}[Exterior calculus steady-state primal weak model for transport phenomena]
  \label{formulation:continuum/primal_weak/steady_state}
  In order to get the steady-state version ($\frac{\partial u}{\partial t} = 0)$ of \Cref{eq:continuum/primal_weak/transient_formulation}, we just drop the bilinear form $B$ and assume that all the quantities are time-independent.
  Hence, with time-independent assumptions (\Cref{notation:continuum/parameters}) and operators (\Cref{eq:continuum/primal_weak/operators}), we get the following problem for the unknown potential $u [\potential] \in \Omega^0 M$:
  \begin{subequations}
    \label{eq:continuum/primal_weak/steady_state_formulation}
    \begin{alignat}{4}
      & \forall w [\potential] \in \Ker \tr_{\Gamma_{\Dirichlet}, 0}, \quad
      && A_{\diffusive}(w, u) - A_{\advective}(w, u)
      && = F_1(w) + F_2(w) - G(w) \qquad
      && [\amount \time^{-1} \potential], \\
      &
      && \tr_{\Gamma_{\Dirichlet}, 0}(u)
      && = g_{\Dirichlet} \qquad
      && [\potential], \\
      &
      && u(t_0)
      && = u_0 \qquad
      && [\potential].
    \end{alignat}
  \end{subequations}
  The flow rate $q [\amount \time^{-1}] \in \Omega^{D - 1} M$
  is calculated in the post-processing phase by the formula
  \begin{equation}
    \label{eq:continuum/primal_weak/steady_state_post_processing}
    q(x) =
    \begin{cases}
      (-\star_1\widetilde{\kappa} d_0 u + \widetilde{Q_0} \wedge v)(x),
        & x \notin \Gamma_\Neumann \\
      g_{\Neumann}(x), & x \in \Gamma_{\Neumann}
    \end{cases}.
  \end{equation}
\end{formulation}
\begin{remark}
  If only diffusive flow is present (i.e., $v = 0$), the previous formulations (transient, \Cref{eq:continuum/primal_weak/transient_formulation}, and steady-state, \Cref{eq:continuum/primal_weak/steady_state_formulation}) can be further simplified by dropping the term $A_{\advective}(w, u)$.
  (In fact, this is the case in our implementation and numerical examples.)
  This leads to a symmetric positive-definite system in the steady-state case, and to a system of positive-definite systems for the time integration scheme (trapezoidal, or Crank-Nicolson, method) in the transient case.
\end{remark}

\subsection{Mixed weak formulations of transport phenomena}
\label{sec:continuum/mixed}

\begin{discussion}
  Choose a test function
  \begin{equation}
    r \in \Ker(\tr_{\Gamma_{\Neumann, D - 1}}) = \set{s \in \Omega^{D - 1} \Gamma_{\Neumann}}{\tr_{\Gamma_{\Neumann, D - 1}} s = 0}.
  \end{equation}
  Rewrite the Constitutive law (\Cref{eq:continuum/constitutive_law}) as
  \begin{equation}
    \kappa^{-1} q_{\diffusive}
    = d^\star_D \star_0 u
    = -\star_1 d_0 u.
  \end{equation}
  Then,
  \begin{equation}
    \label{eq:continuum/mixed_weak/constitutive_law_calculations}
    \begin{split}
      \inner{r}{\kappa^{-1} q_{\diffusive}}_{M, D - 1}
      & = \inner{r}{-\star_1 d_0 u}_{M, D - 1} \\
      & = - \int_M (d_0 u \wedge r) \\
      & = - \int_{\partial M}
        (\tr_{\partial M, 0} u \wedge \tr_{\partial M, D - 1} r )
        + \int_M (u \wedge d_{D - 1} r) \\
      & = - \int_{\Gamma_{\Dirichlet}} (g_{\Dirichlet} \wedge \tr_{\Gamma_{\Dirichlet}, D - 1} r )
        + \inner{\star_0 u}{d_{D - 1} r} \\
      & = - \int_{\Gamma_{\Dirichlet}} (g_{\Dirichlet} \wedge \tr_{\Gamma_{\Dirichlet}, D - 1} r)
        + \inner{\widetilde{u}}{d_{D - 1} r},
    \end{split}
  \end{equation}
  and therefore,
  \begin{equation}
    \label{eq:continuum/mixed_weak/constitutive_law}
    \begin{split}
      \inner{r}{\kappa^{-1} q}_{M, D - 1}
      & = \inner{r}{\kappa^{-1} q_{\diffusive}}_{M, D - 1}
        + \inner{r}{\kappa^{-1} q_{\advective}}_{M, D - 1} \\
      & = -\int_{\Gamma_{\Dirichlet}} (\tr_{\Gamma_{\Dirichlet}, D - 1} r \wedge g_{\Dirichlet})
        + \inner{d_{D - 1} r}{\tilde{u}}
        + \inner{r}{\kappa^{-1} (\star_D (S_0 + \pi \widetilde{u}) \wedge v)}_{M, D - 1} \\
      & = -\int_{\Gamma_{\Dirichlet}} (\tr_{\Gamma_{\Dirichlet}, D - 1} r \wedge g_{\Dirichlet})
        + \inner{d_{D - 1} r}{\tilde{u}}
        + \inner{r}{\kappa^{-1} (\widetilde{S_0} \wedge v)}_{M, D - 1}
        + \inner{r}{\kappa^{-1} (\star_D \pi \widetilde{u} \wedge v)}_{M, D - 1}.
    \end{split}
  \end{equation}
  Let $\widetilde{w} [\potential \length^D] \in \Omega^D X$ be a test function.
  Taking the inner product of the conservation law (\Cref{eq:continuum/conservation_law_differential_form}) with $\widetilde{w}$ gives
  \begin{equation}
    \label{eq:continuum/mixed_weak/conservation_law}
    \inner{\pi \frac{\partial \widetilde{u}}{\partial t}}{\widetilde{w}}_{M, D}
    = \inner{f}{\widetilde{w}}_{M, D} - \inner{d_{D - 1} q}{\widetilde{w}}_{M, D}.
  \end{equation}
  Combining the weak version of the constitutive law (\Cref{eq:continuum/mixed_weak/constitutive_law}),
  the weak version of the conservation law (\Cref{eq:continuum/mixed_weak/conservation_law}),
  the Neumann boundary condition (\Cref{eq:continuum/neumann_boundary_condition}),
  and the initial conditions (\Cref{eq:continuum/initial_condition_potential}, \Cref{eq:continuum/initial_condition_amount}),
  leads to the formulation presented below.
\end{discussion}
\begin{formulation}[Exterior calculus transient mixed weak model for transport phenomena]
  \label{formulation:continuum/mixed_weak/transient}
  Under the assumptions of \Cref{notation:continuum/parameters}
  define the following operators
  \begin{subequations}
    \label{eq:continuum/mixed_weak/operators}
    \begin{alignat}{3}
      & A \colon \Omega^{D - 1} M \times \mathcal{C}^\infty(I, \Omega^{D - 1} M) \to \R, \,
      && A(r, s)
        := \inner{r}{\kappa^{-1} s}_{M, D - 1} \,
      && [\amount^{-1} \time \potential], \\
      & B_{\diffusive} \colon \Omega^D M \times \mathcal{C}^\infty(I, \Omega^{D - 1} M) \to \R, \,
      && B_{\diffusive}(\widetilde{w}, r)
        := \inner{d_{D - 1} r}{\widetilde{w}}_{M, D} \,
      && [\length^{-D}], \\
      & B_{\advective} \colon \Omega^D M \times \mathcal{C}^\infty(I, \Omega^{D - 1} M) \to \R, \,
      && B_{\advective}(\widetilde{w}, r)
        := \inner{r}{\kappa^{-1} (\star_D \pi \widetilde{w} \wedge v)}_{M, D - 1} \,
      && [\length^{-D}], \\
      & C \colon \Omega^D M \times \mathcal{C}^\infty(I, \Omega^D M) \to \R, \,
      && C(\widetilde{w}, \widetilde{u}) := \inner{\pi \widetilde{u}}{\widetilde{w}}_{M, D} \,
      && [\amount \length^{-2 D} \potential^{-1}], \\
      & G_1 \colon \Omega^{D - 1} M \to \R, \,
      && G_1(r)
        := \int_{\Gamma_{\Dirichlet}} (\tr_{\Gamma_{\Dirichlet}, D - 1} r \wedge g_{\Dirichlet})
        \,
      && [\potential], \\
      & G_2 \colon \Omega^{D - 1} M \to \R, \,
      && G_2(r) := \inner{r}{\kappa^{-1} (\widetilde{S_0} \wedge v)}_{M, D - 1} \,
      && [\potential], \\
      & F \colon \Omega^D M \to \R, \,
      && F(\widetilde{w}) := \inner{f}{\widetilde{w}}_{M, D} \,
      && [\amount \time^{-1} \length^{-D}], \\
      & \mathtt{flow\_rate} \colon \Omega^0 M \to \Omega^{D - 2} M, \,
      && \mathtt{flow\_rate}(w) := \kappa\, d^\star_D \star_0 w + (\widetilde{S_0} + \tilde{\pi} w) \wedge v \,
      && [\amount \time^{-1} \potential^{-1}].
    \end{alignat}
  \end{subequations}
  The unknown variables are:
  \begin{itemize}
    \item
      $q [\amount \time^{-1}] \in \mathcal{C}^\infty(I, \Omega^{D - 1} M)$ (flow rate);
    \item
      $\widetilde{u} [\potential \length^D] \in \mathcal{C}^\infty(I, \Omega^D M)$ (dual potential).
  \end{itemize}
  We are solving the following problem for $q$ and $\widetilde{u}$:
  \begin{subequations}
    \label{eq:continuum/mixed_weak/transient_formulation}
    \begin{alignat}{4}
      & \forall r [\amount \time^{-1}] \in \Ker \tr_{\Gamma_{\Neumann}, D - 1}, \;
      && A(r, q) - B_{\diffusive}^T(r, \widetilde{u}) - B_{\advective}^T(r, \widetilde{u})
      && = G_2(r) - G_1(r) \;
      && [\amount \time^{-1} \potential], \\
      & \forall \widetilde{w} [\potential \length^D] \in \Omega^D M, \;
      && - B_{\diffusive}(\widetilde{w}, q) - C(\widetilde{w}, \frac{\partial \widetilde{u}}{\partial t})
      && = - F(\widetilde{w}) \;
      && [\amount \time^{-1} \potential], \\
      &
      && \tr_{\Gamma_{\Neumann}, D - 1} q
      && = g_{\Neumann} \;
      && [\amount \time^{-1}], \\
      &
      && \widetilde{u}(t_0)
      && = \star_0 u_0 \;
      && [\potential \length^D], \\
      &
      && q(t_0)
      && = \mathtt{flow\_rate}(u_0) \;
      && [\amount \time^{-1}].
    \end{alignat}
  \end{subequations}
  The potential $u [\potential] \in \mathcal{C}^\infty(I, \Omega^0 M)$ is calculated in the post-processing phase by the formula
  \begin{equation}
    \label{eq:continuum/mixed_weak/transient_post_processing}
    u(t, x) :=
    \begin{cases}
      u_0(x), & t = t_0 \\
      (\star_D \widetilde{u})(t, x), & t > t_0\ \text{and}\ x \notin \Gamma_{\Dirichlet} \\
      g_{\Dirichlet}(t, x), & t_0 > 0\ \text{and}\ x \in \Gamma_{\Dirichlet}
    \end{cases}.
  \end{equation}
\end{formulation}
\begin{formulation}[Exterior calculus steady-state mixed weak model for transport phenomena]
  \label{formulation:continuum/mixed_weak/steady_state}
  The steady-state version of \Cref{eq:continuum/mixed_weak/transient_formulation} is
  \begin{subequations}
    \label{eq:continuum/mixed_weak/steady_state_formulation}
    \begin{alignat}{4}
      & \forall r [\amount \time^{-1}] \in \Ker \tr_{\Gamma_{\Neumann}, D - 1}, \;
      && A(r, q) - B_{\diffusive}^T(r, \widetilde{u}) - B_{\advective}^T(r, \widetilde{u})
      && = G_2(r) - G_1(r) \;
      && [\amount \time^{-1} \potential], \\
      & \forall \widetilde{w} [\potential \length^D] \in \Omega^D M, \;
      && - B_{\diffusive}(\widetilde{w}, q)
      && = - F(\widetilde{w}) \;
      && [\amount \time^{-1} \potential], \\
      &
      && \tr_{\Gamma_{\Neumann}, D - 1} q
      && = g_{\Neumann} \;
      && [\amount \time^{-1}].
    \end{alignat}
  \end{subequations}
  The potential $u [\potential] \in \Omega^0 M$ is calculated in the post-processing phase by the formula
  \begin{equation}
    \label{eq:continuum/mixed_weak/steady_state_post_processing}
    u(x) :=
    \begin{cases}
      (\star_D \widetilde{u})(x), & x \notin \Gamma_{\Dirichlet} \\
      g_{\Dirichlet}(x), & x \in \Gamma_{\Dirichlet}
    \end{cases}.
  \end{equation}
\end{formulation}

\subsection{Domains for physical quantities}
\phantom{T}
\begin{discussion}
  We end this section with a short commentary on the domains used for the physical quantities appearing in the exterior calculus models.
  To simplify discussion we always assumed that all variables are infinitely smooth.
  This assumption is too strong condition and has to be relaxed in order to guarantee the existence of solutions.
  For weak formulations the appropriate domains for spaces of $p$-forms are the Sobolev spaces of weakly differentiable $p$-forms \cite{arnold2018finite}
  \begin{equation}
    H \Lambda^p M := \set{\omega \in L^2 \Lambda^p M}{d \omega \in L^2 \Lambda^{p + 1} M },
  \end{equation}
  where $L^2 \Lambda^q M$ denotes the space of $q$-forms whose components (in any coordinate system) are square-integrable.
  Domains for boundary conditions also need to be adjusted, and we refer to the cited book for details.

  Also, the requirement for space domain to be a smooth manifold boundary can be weakened so that the boundary is Lipschitz continuous, thus including general polytopes as valid domains.
\end{discussion}

\section{Combinatorial mesh calculus (CMC)}
\label{sec:combinatorial}
\begin{discussion}
  The smooth exterior calculus introduced in \Cref{sec:exterior_calculus} generalises the vector calculus and can replace it in modelling continuous media as presented in \Cref{sec:continuum}.
  As discussed in \Cref{sec:introduction/motivation}, media with internal structures are appropriately described by algebraic topological structures called cell complexes.
  The basics of these structures are given in \Cref{sec:combinatorial/meshes} and a discrete analogue of the smooth exterior calculus is presented in the rest of this section.
\end{discussion}

\subsection{Cell complexes (meshes)}
\label{sec:combinatorial/meshes}
\begin{definition}
  A \newterm{combinatorial cell complex} (or \newterm{combinatorial mesh}, or simply \newterm{mesh}) $\mathcal{M}$ is a collection of elements, called \newterm{cells}, that is a realisation of a subdivision of a topological space.
  More formally, it is a partially ordered set, with partial order $\preceq$, such that there exist a topological space $M$ and a function
  $\varphi \colon \mathcal{M} \to \mathcal{P}(M)$
  ($\mathcal{P}(M)$ is the power set of $M$),
  called \newterm{embedding} of $\mathcal{M}$ into $M$, satisfying the following conditions:
  \begin{enumerate}
    \item
      for any $a \in \mathcal{M}$, $\varphi(a) \subseteq M$ is homeomorphic to an open ball;
    \item
      for any $a, b \in \mathcal{M}$, if $a \neq b$, then $\varphi(a) \cap \varphi(b) = \emptyset$;
    \item
      for any $a \in \mathcal{M}$,
      \begin{equation}
        \overline{\varphi(a)} = \bigcup_{b \preceq a} \varphi(b),
      \end{equation}
      where $\overline{X}$ denotes the closure of $X$.
      Equivalently, if $\partial X = \overline{X} \setminus {\rm int(X)}$ is the topological boundary of $X$, then
      \begin{equation}
        \partial(\varphi(a)) = \bigcup_{b \prec a} \varphi(b),
      \end{equation}
  \end{enumerate}
  We say that $(M, \varphi)$ is a \newterm{realisation} of $\mathcal{M}$ if
  \begin{equation}
    M = \bigcup_{a \in \mathcal{M}} \varphi(a).
  \end{equation}
\end{definition}
\begin{remark}
  Let $(\mathcal{M}, \preceq)$ be a mesh.
  It has been constructed by the requirement of embedding but intrinsically it is a partial order, which we call the \newterm{topology of $\mathcal{M}$}.
  (A similar definition for CW complexes, requiring only the partial order, is given in \cite{bjorner1984posets}.)
  Henceforth, we say that a construction based on the partial order and the relative orientations defined shortly is a \newterm{topological concept on a mesh}.
  Such concepts will form the backbone of the discrete theory.
  Considering an intrinsic metric, such as positive numbers associated with the cells of the mesh, will introduce metric concepts, but the theory will still be intrinsic, e.g., independent of any embedding of the mesh in a manifold.
\end{remark}
\begin{notation}
  Let $(\mathcal{M}, \preceq)$ be a mesh.
  The elements of $\mathcal{M}$ are called \newterm{cells}.
  If $a, b \in \mathcal{M}$ and $a \preceq b$, we say that $a$ is a \newterm{subface} of $b$.
  The relations $\prec$, $\succeq$ and $\succ$ are respectively the strict, reverse, and strict reverse versions of $\preceq$.
  If $a \prec b$ and there does not exist $c \in \mathcal{M}$ such that $a \prec c \prec b$, we say that $a$ is a \newterm{hyperface} of $b$ and write it as $a \precdot b$ or $b \succdot a$.
  (The order theory terminology is ``$b$ covers a''.)
  We say that $a$ is a \newterm{node} if it is a minimal element, i.e., there does not exist $c \in \mathcal{M}$ such that $c \prec a$.
  For $a \in \mathcal{M}$ the \newterm{face lattice} of $a$ is the set of subfaces of $a$.
\end{notation}
\begin{proposition}
  Let $\mathcal{M}$ be a mesh, $M$ be a manifold, $\varphi \colon \mathcal{M} \to \mathcal{P}(M)$ be an embedding,
  $a \in \mathcal{M}$, and consider a maximal chain for $a$, i.e., a sequence $a_0 \precdot a_1 \precdot ... \precdot a_p = a$, where $a_0$ is a node.
  Then $\dim \varphi(a) = p$.
\end{proposition}
\begin{remark}
  The previous proposition has a direct consequence that any maximal chain for a cell $a$ will have the same number of elements,
  equal to the dimension of an embedding of $a$ plus $1$.
  This leads to the following definition of \newterm{dimension a cell}, $\dim a = p$: $p$ can defined extrinsically as the dimension of its embedding, or intrinsically as $p = l - 1$, where $l$ is the number of elements a maximal chain for $a$.

  Note that the order theory term ``chain'' will no longer be used and from now on ``chain'' will refer to an element of a chain complex.
\end{remark}
\begin{definition}
  Let $\mathcal{M}$ be a mesh.
  For any $p \in \N$, we define the set of \newterm{$p$-cells}
  \begin{equation}
    \mathcal{M}_p := \set{a \in \mathcal{M}}{\dim a = p}.
  \end{equation}
  The \newterm{dimension of $\mathcal{M}$} is defined by
  \begin{equation}
    \dim \mathcal{M} = \max(\set{p \in \N}{\mathcal{M}_p \neq \emptyset}).
  \end{equation}
  The dimension of $\mathcal{M}$ also equals the maximal dimension of the cells in an embedding.
  The following decomposition holds:
  \begin{equation}
    \mathcal{M} = \bigcup_{p = 0}^{\dim \mathcal{M}} \mathcal{M}_p.
  \end{equation}
\end{definition}
\begin{proposition}
  Let $\mathcal{M}$ be a mesh of dimension $D$.
  Then there exists a collection of functions
  \begin{equation}
    \set{\varepsilon_p \colon \mathcal{M}_p \times \mathcal{M}_{p - 1} \to \{-1, 0, 1\}}{p = 1, ..., D},
  \end{equation}
  called \newterm{relative orientations}, satisfying the following conditions:
  \begin{enumerate}
    \item
      nonzero values correspond to cell-hyperface pairs:
      for any $p \in \{1, ..., D\}$,
      $a \in \mathcal{M}_p$,
      $b \in \mathcal{M}_{p - 1}$,
      \begin{equation}
        \varepsilon_p(a, b) \neq 0 \Leftrightarrow b \precdot a;
      \end{equation}
    \item
      chain complex property:
      for any $p \in \{2, ..., D\}$,
      $a \in \mathcal{M}_p$,
      $c \in \mathcal{M}_{p - 2}$ with $c \prec a$,
      if $b, b' \in \mathcal{M}_{p - 1}$ are the only two cells between $c$ and $a$,
      \begin{equation}
        \varepsilon_p(a, b)\, \varepsilon_{p - 1}(b, c) + \varepsilon_p(a, b')\, \varepsilon_{p - 1}(b', c) = 0.
      \end{equation}
    \item
      (optional) all nodes are positively oriented:
      for any edge $a \in \mathcal{M}_1$ with nodes (hyperfaces) $b, c \precdot a$,
      \begin{equation}
        \varepsilon_1(a, b) = - \varepsilon_1(a, c);
      \end{equation}
  \end{enumerate}
\end{proposition}
\begin{proof}
  Take an embedding $\varphi \colon \mathcal{M} \to \mathcal{P}(M)$ into some manifold $M$.
  For any $a \in \mathcal{M}$ choose an orientation $\OR_{\varphi(a)}$ of $\varphi(a)$.
  For any two cells $a, b \in \mathcal{M}$ define
  \begin{equation}
    \varepsilon(a, b) :=
    \begin{cases}
      \varepsilon^M(\varphi(a), \varphi(b)), & b \precdot a \\
      0, & \text{else}
    \end{cases},
  \end{equation}
  where $\varepsilon^M$ are the relative orientations on $M$, see \Cref{def:exterior_calculus/relative_orientations}.
  The chain complex property follows from \Cref{eq:exterior_calculus/chain_complex_property}.
  The last, optional, requirement is satisfied if all embedded nodes, i.e., $\varphi(a)$ for $a \in \mathcal{M}_0$, have positive orientations.
\end{proof}
\begin{definition}
  Let $\mathcal{M}$ be a mesh of dimension $D$, $p \in \{0, ..., D\}$.
  A \newterm{$p$-chain} (with real coefficients) is a formal linear combination of $p$-cells in $\mathcal{M}$.
  The vector space of $p$-chains on $\mathcal{M}$, $C_p \mathcal{M}$ is the free real vector space (i.e., the space of formal linear combinations with real coefficients) of the set $\mathcal{M}_p$.
  The set of all chains is denoted by $C_\bullet \mathcal{M}$.
  The following decomposition holds:
  \begin{equation}
    C_\bullet \mathcal{M} := \bigoplus_{p = 0}^D C_p \mathcal{M}.
  \end{equation}
  For any cell $a \in \mathcal{M}$ by $a_\bullet \in C_\bullet \mathcal{M}$ we will denote the corresponding basis chain.
\end{definition}
\begin{definition}
  Let $\mathcal{M}$ be a mesh.
  The linear-algebraic dual of the chain space
  \begin{equation}
    C^\bullet \mathcal{M} := (C_\bullet \mathcal{M})^*
  \end{equation}
  is called the \newterm{space of cochains} and its elements are called \newterm{cochains}.
  The following decomposition holds
  \begin{equation}
    C^\bullet \mathcal{M} = \bigoplus_{p = 0}^D C^p \mathcal{M},\ C^p \mathcal{M} = (C_p \mathcal{M})^*\ (p = 0, ..., D).
  \end{equation}
  For any cell $a \in \mathcal{M}$ by $a^\bullet \in C^\bullet \mathcal{M}$ we will denote the corresponding basis cochain.
\end{definition}

\subsection{Topological operations}
\label{sec:combinatorial/topological}
\begin{definition}
  Let $\mathcal{M}$ be a mesh with relative orientations $\varepsilon = \{\varepsilon_p\}_{p = 0}^{D - 1}$.
  The \newterm{boundary operator} on $C_\bullet \mathcal{M}$ is the linear operator
  \begin{equation}
    \partial \colon C_\bullet \mathcal{M} \to C_\bullet \mathcal{M},
  \end{equation}
  defined for basis $p$-chains ($1 \leq p \leq D$) as follows: for any $a \in \mathcal{M}_p$
  \begin{equation}
    \partial a_\bullet
    := \partial_p a_\bullet
    = \sum_{b \in \mathcal{M}_{p - 1}} \varepsilon(a, b)\, b_\bullet
    = \sum_{b \precdot\, a} \varepsilon(a, b)\, b_\bullet
    \in C_{p - 1} \mathcal{M},
  \end{equation}
  where $\partial$ is decomposed into operators
  \begin{equation}
    \partial_p \colon C_p \mathcal{M} \to C_{p - 1} \mathcal{M},\ p = 1, ..., D.
  \end{equation}
\end{definition}
\begin{proposition}[Restatement of \Cref{eq:exterior_calculus/chain_complex_property}]
  Let $\mathcal{M}$ be a mesh with relative orientations
  $\varepsilon = \{\varepsilon_p\}_{p = 0}^{D - 1}$,
  and the corresponding boundary operator
  $\partial \colon C_\bullet \mathcal{M} \to C_\bullet \mathcal{M}$.
  Then
  \begin{equation}
    \partial \circ \partial = 0,
  \end{equation}
  or specifically,
  \begin{equation}
    \partial_p \circ \partial_{p + 1} = 0,\ p = 1, ..., D - 1,
  \end{equation}
  making $(\mathcal{M}, \partial)$ a chain complex.
\end{proposition}
\begin{definition}
  Let $\mathcal{M}$ be a mesh with boundary operator $\partial$.
  The \newterm{coboundary operator} $\delta$ is the dual of $\partial$:
  \begin{equation}
    \delta = \partial^* \colon C^\bullet \mathcal{M} \to C^\bullet M,
  \end{equation}
  i.e., for any $\gamma \in C_\bullet \mathcal{M}$, $\rho \in C^\bullet \mathcal{M}$,
  \begin{equation}
    (\delta \rho) \gamma = \rho(\partial \gamma).
  \end{equation}
  The coboundary operator can be decomposed into operators
  \begin{equation}
    \delta_p \colon C^p \mathcal{M} \to C^{p + 1} \mathcal{M},\ p = 0, ..., D - 1,
  \end{equation}
  and then
  \begin{equation}
    \delta_p = (\partial_{p + 1})^*,\ p = 0, ..., D - 1.
  \end{equation}
\end{definition}
\begin{proposition}
  Let $\mathcal{M}$ be a mesh with boundary operator $\partial$ and coboundary operator $\delta$.
  Then $(C^\bullet M, \delta)$ is a \newterm{cochain complex}, i.e.,
  \begin{equation}
    \delta \circ \delta = 0,
  \end{equation}
  or on components:
  \begin{equation}
    \delta_{p + 1} \circ \delta_p = 0,\ p = 0, ..., D - 2.
  \end{equation}
\end{proposition}
\begin{proof}
  Follows directly from one of the basic properties of dual maps:
  $(f \circ g)^* = g^* \circ f^*$,
  and from $\partial \circ \partial = 0$.
\end{proof}
\begin{definition}
  Let $\mathcal{M}$ be a mesh,
  $M$ be a smooth oriented manifold with boundary,
  $\varphi \colon \mathcal{M} \to \mathcal{P} M$ be a realisation of $\mathcal{M}$ onto $M$.
  Assume that orientations are given to the images of cells of $\mathcal{M}$, and hence relative orientations are induced on $\mathcal{M}$.
  Define the \newterm{de Rham map} $R$ as the unique bilinear map
  \begin{equation}
    R \colon \Omega^\bullet M \to C^\bullet \mathcal{M},\
    R_p \colon \Omega^p M \to C^p \mathcal{M}\ (p = 0, ..., D),
  \end{equation}
  such that for any $p \in \{0, ..., D\}$, $\omega \in \Omega^p M$, $c \in \mathcal{M}_p$,
  \begin{equation}
    (R_p \omega)(c_\bullet) := \int_{\varphi(c)} \tr_{\varphi(c)}(\omega).
  \end{equation}
\end{definition}
\begin{remark}
  Using Stokes-Cartan theorem it is not hard to see that for $p \in \{0, ..., D - 1\}$,
  \begin{equation}
    R_{p + 1} \circ d_p = \delta_p \circ R_p.
  \end{equation}
  In the language of cochain complexes, $R$ is a \newterm{cochain map}.
  We can interpret this result as follows: mesh cochains are a discrete analogue of differential forms, while the coboundary operator is a discrete version of the exterior derivative.
  This interpretation, together with other discrete operators we will introduce, will form the basis of the translation of exterior calculus formulations into CMC ones.
\end{remark}
\begin{definition}
  Let $\mathcal{M}$ be a mesh,
  $\mathcal{S}$ be a sub-mesh of $\mathcal{M}$, i.e., a subset of $\mathcal{M}$ that is also a mesh.
  Define the \newterm{discrete trace}
  \begin{equation}
    \tr_{\mathcal{S}} \colon C^\bullet \mathcal{M} \to C^\bullet \mathcal{S},\ \tr_{\mathcal{S}, p} \colon C^p \mathcal{M} \to C^p \mathcal{S}\ (p = 0, ..., \dim \mathcal{S})
  \end{equation}
  as follows: for any $\sigma \in C^\bullet \mathcal{M}$, $c \in \mathcal{S}$,
  \begin{equation}
    \tr_{\mathcal{S}}(\sigma)(c_\bullet) = \sigma(c_\bullet),
  \end{equation}
  where $c_\bullet$ on the left is a basis chain on $C_\bullet \mathcal{S}$, while $c_\bullet$ on the right is a basis chain on $C_\bullet \mathcal{M}$.
\end{definition}
\begin{remark}
  Consider a realisation $\varphi$ of $\mathcal{M}$ onto $M$, which induces a realisation $\restrict{\varphi}{\mathcal{S}}$ of $\mathcal{S}$ onto a submanifold $S$ of $M$.
  It is not hard to see that the de Rham maps commute with traces (smooth and discrete) as well, i.e., for any $p \in \{0, ..., D\}$,
  \begin{equation}
    R_{S, p} \circ \tr_{S, p} = \tr_{\mathcal{S}, p} \circ R_{M, p}.
  \end{equation}
\end{remark}
\begin{definition}
  Let $\mathcal{M}$ be a manifold-like mesh, i.e., a tessellation of a manifold, $\partial \mathcal{M}$ be its boundary.
  Denote
  \begin{equation}
    C^p_0 \mathcal{M} := \Ker \tr_{\partial \mathcal{M}, p} = \set{\pi \in C^p M}{\tr_{\partial \mathcal{M}, p}(\pi) = 0}.
  \end{equation}
  Like in the continuum, $\delta_p$ restricts to a well-defined operator
  \begin{equation}
    \tilde{\delta}_p \colon C^p_0 \mathcal{M} \to C^{p + 1}_0 \mathcal{M}.
  \end{equation}
\end{definition}
\begin{discussion}
  Until now we have not related meshes to manifolds.
  And indeed, when embedded, meshes can take shapes that do not resemble manifolds.
  We will require manifold-likeness in order to define a global or compatible orientation.
  The relative orientations work locally on cells and their neighbours, but they are not checked for global consistency, which we will do below.
\end{discussion}
\begin{definition}
  Let $\mathcal{M}$ be a mesh of dimension $D$
  that has the following property: any $(D - 1)$-cell is a hyperface of at most two $D$-cells. The $(D - 1)$-cells with $2$ adjacent $D$-cells are called \newterm{interior} cells, while those with $1$ adjacent $D$-cell are called \newterm{boundary} cells.)
  A \newterm{compatible orientation} on $\mathcal{M}$ is a relative orientation $\varepsilon_D$ between $D$-cells and $(D - 1)$-cells such that for any interior $c \in C^{D - 1} \mathcal{M}$ with adjacent $a, b \in C^D \mathcal{M}$,
  \begin{equation}
    \varepsilon_D(a, c) = - \varepsilon_D(b, c).
  \end{equation}
  If $\varepsilon$ is a compatible orientation on a finite mesh $\mathcal{M}$,
  we define the \newterm{fundamental class} of $\mathcal{M}$ by
  \begin{equation}
    [\mathcal{M}] := \sum_{c \in \mathcal{M}_D} c_\bullet.
  \end{equation}
\end{definition}
\begin{remark}
  The notion of a compatible orientation of a mesh is linked to the notion of an orientation of a mesh.
  If $M$ is a manifold and $\mathcal{M}$ is a decomposition of $M$, then an orientation on $M$ gives rise to a compatible orientation on $\mathcal{M}$ and vice-versa.
  In fact, just like integration is defined on an oriented manifold, a discrete version of integration can be defined on a compatibly oriented mesh.
\end{remark}
\begin{definition}
  Let $\mathcal{M}$ be a compatibly oriented finite mesh of dimension $D$ with fundamental class $[\mathcal{M}]$,
  $\sigma \in C^D M$.
  Denote
  \begin{equation}
    \text{``\newterm{discrete integral} of $\sigma$ over $\mathcal{M}$''} := \sigma [\mathcal{M}].
  \end{equation}
\end{definition}
\begin{remark}
  If the oriented manifold $M$ realises the oriented mesh $\mathcal{M}$ of dimension $D$,
  $R$ is the de Rham map, $\omega \in \Omega^D M$, then
  \begin{equation}
    \int_M \omega = R_D(\omega)[\mathcal{M}].
  \end{equation}
\end{remark}
\begin{discussion}
  Until now we have not introduced a discrete analogue to the wedge product, called a \newterm{cup product}, denoted by $\smile$.
  A cup product is a family of bilinear maps $\smile_{p, q} \colon C^p \mathcal{M} \times C^q \mathcal{M} \to C^{p + q} \mathcal{M}$, that satisfies the graded Leibniz rule with respect to the coboundary operator, and is local in the sense that for $\pi \in C^p \mathcal{M}, \rho \in C^q \mathcal{M},\ a \in \mathcal{M}_{p + q}$,
  \begin{equation}
    (\pi \smile \rho)(a_\bullet) = \sum_{b \in \mathcal{M}_p, b \preceq a}\, \sum_{c \in \mathcal{M}_q, c \preceq a} \lambda_{a, b, c}\, \pi(b_\bullet)\, \rho(c_\bullet),
  \end{equation}
  for some coefficients $\lambda_{a, b, c}$, depending only on the topology and the relative orientations of $\mathcal{M}$.
  In other words, when the cup product is applied to a cell, only the values of the cochains of the cell boundary are taken into account.

  Cup products exist for simplicial \cite{wilson2007cochain}, cubical \cite{arnold2012discrete}, and polygonal \cite{ptackova2017discrete} meshes.
  In all cases they are graded-commutative and posses the graded Leibniz rule, but are not associative (even in dimension $1$!).
  However, for simplices the cup product converges to the smooth wedge product \cite[Theorem 5.4]{wilson2007cochain} and cochains coming from de Rham maps are almost associative \cite[Theorem 5.9]{wilson2007cochain}.
  Cup product can also be defined easily for products of meshes using the tensor product of non-associative graded-commutative differential graded algebras, thus allowing to define it over prisms as well.
  However, we are unaware of a cup product for a generic polytopal mesh.
  Moreover, we want to use a notion of topological orthogonality which works well for cubical meshes.
  For this reason we will use a construction of a \newterm{quasi-cubical mesh} $\mathcal{K}$ (i.e., all its cells have face lattices isomorphic to those of cubes) from a generic mesh $\mathcal{M}$ of simple polytopes, called the \newterm{Forman subdivision}, outlined in detail in \cite[Section 2]{berbatov2022diffusion}.
  Another reason to use the Forman subdivision is to allow for more complex material properties (conductivities), as discussed in \cite[Section 4]{berbatov2022diffusion}.
  In this article's numerical examples we will use only constant conductivities but the intention is to use the discrete formulations (\Cref{sec:continuum/primal,sec:continuum/mixed}) with different conductivities of cells of different dimensions as dictated by practical applications.
\end{discussion}
\begin{discussion}
  The topology of the Forman subdivision $\mathcal{K}$ of a mesh $\mathcal{M}$ is defined as follows: the $p$-cells of $\mathcal{K}$ are the intervals
  $[a, b] := \set{c \in \mathcal{M}}{a \preceq c \preceq b}$,
  where $a \in \mathcal{M}_q,\ b \in \mathcal{M}_{q - p},\ q \geq p$.
  The partial order on $\mathcal{K}$ is the sub-interval relation on $\mathcal{M}$, i.e.,
  \begin{equation}
    [a, b] \preceq_{\mathcal{K}} [c, d] \Leftrightarrow c \preceq_{\mathcal{M}} a \preceq_{\mathcal{M}} b \preceq_{\mathcal{M}} d.
  \end{equation}
  For the embedding of $\mathcal{K}$ it was proposed in \cite{berbatov2022diffusion} to use barycentric coordinates for nodes and flat $p$-cells, if possible for $p \geq 1$.
  However, in this article we allow generic embeddings and so the natural embedding of $\mathcal{K}$ will be defined in case by case way.
  For instance, in curved meshes, as those used in \Cref{sec:simulations/disk,sec:simulations/hemisphere}, it will be more natural to use flat cells in the respective curvilinear parametrisations.
\end{discussion}
\begin{definition}
  Let $\mathcal{K}$ be a quasi-cubical mesh,
  $p, q \in \N,\ p + q \leq \dim \mathcal{K}$,
  $b \in \mathcal{K}_p,\ c \in \mathcal{K}_q,\ a \in \mathcal{K}_{p + q}$.
  We say that $b$ and $c$ are
  \newterm{topologically orthogonal} with respect to $a$, and write it as
  \begin{equation}
    b \perp_a c,
  \end{equation}
  if $b$ and $c$ share a single common $0$-cell, $b \preceq a$, and $c \preceq a$.
  Also, for a $(p + q)$-cell $a$, denote by $\perp_{p, q} a$ the set of all pairs of perpendicular $p$- and $q$-subfaces of $a$.
\end{definition}
\noindent
\begin{minipage}{0.7\textwidth}
\begin{remark}
  In a non-convex mesh (i.e., not all cells are convex polytopes) it is possible for two intersecting $1$-cells to be included in two $2$-cells.
  For instance, consider a mesh with nodes $N_0 = (0, 0),\ N_1 = (1, -1),\ N_2 = (2, 0),\ N_3 = (1, 1),\ N_4 = (1, 2)$, edges $E_0 = (N_0, N_1),\ E_1 = (N_1, N_2),\ E_2 = (N_2, N_3),\ E_3 = (N_3, N_0),\ E_4 = (N_2, N_4),\ E_5 = (N_4, N_0)$, and two quadrilaterals: a convex quadrilateral $F_0 = (E_0, E_1, E_2, E_3)$, and a concave quadrilateral $F_1 = (E_3, E_2, E_4, E_5)$.
  Then $F_0$ and $F_1$ share two common edges $E_2$ and $E_3$.
  Hence, $E_2 \perp_{F_0} E_3$, and $E_2 \perp_{F_1} E_3$.
\end{remark}
\end{minipage}
\begin{minipage}{0.3\textwidth}
  \vspace{0pt}
  \centering
  \includegraphics[width = .5\linewidth, keepaspectratio]{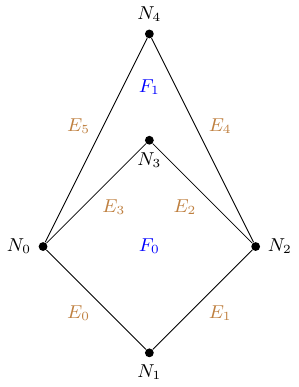}
\end{minipage}
\begin{discussion}
  Let $\mathcal{K}$ be a quasi-cubical mesh of dimension $D$ with relative orientations $\varepsilon$,
  $0 \leq p, q$ with $p + q \leq D$,
  $a \in \mathcal{K}_{p + q}$, $(b, c) \in \perp_{p, q} a$.
  For the definition of cup product we will need a way to relate the orientation of $a$ with the orientations of $b$ and $c$, which we will denote by
  \begin{equation}
    \rel(a, b, c) \in \{-1, 1\}.
  \end{equation}
  In \cite{berbatov2022diffusion} we worked on a convex mesh $\mathcal{K}$ and relative orientations were defined by vector space orientations of the embedding of $\mathcal{K}$ in $\R^D$.
  More precisely, the expression there was equivalent to
  \begin{equation}
    \label{eq:discrete/exterior_calculus/perpendicular_orientations}
    \OR(a) = \rel(a, b, c) \OR(b) \wedge \OR(c).
  \end{equation}
  However, this expression does not work for non-convex shapes, and requires embedding.
  For this reason we will derive an intrinsic one, using only the relative orientations.
  We will first consider the case of a cube with given embedding orientations, and isolate an embedding-independent expression that will be used as the definition for abstract quasi-cubical meshes.

  Let $\mathcal{N}$ be the common node of the cubes $a$, $b$ and $c$ so that
  $a$ be a cube spanned by basis vectors $e_1, ..., e_{p + q}$
  with orientation
  $s_a e_1 \wedge ... \wedge e_{p + q}$,
  $b$ be spanned by $e_1, ..., e_p$ with orientation
  $s_b e_1 \wedge ... \wedge e_p$,
  $c$ be spanned by $e_{p + 1}, ..., e_{p + q}$ with orientation
  $s_c e_{p + 1} \wedge ... \wedge e_{p + q}$,
  where $s_a, s_b, s_c \in \{-1, 1\}$.
  Then, substituting in \Cref{eq:discrete/exterior_calculus/perpendicular_orientations} leads to $\rel(a, b, c) = s_a s_b s_c$.
  If $b$ is the $0$-cell $\mathcal{N}$ ($s_b = 1$), then $a = c$ and $s_A = s_c$, and hence $\rel(a, b, c) = 1$.
  Analogously, if $c$ is the $0$-cell $\mathcal{N}$, then $\rel(a, b, c) = 1$.

  Now consider the case when $b$ is a $1$-cell spanned by $e_1$.
  From the calculation of a relative orientation in
  \Cref{ex:exterior_calculus/half_space_boundary_orientation}
  it follows that
  $\varepsilon(a, c) = - s_a s_c$.
  On the other hand, $\varepsilon(b, \mathcal{N}) = - s_b$.
  Hence we get the following expression, depending only on the relative orientations:
  \begin{equation}
    \rel(a, b, c) = \varepsilon(a, c)\, \varepsilon(b, \mathcal{N}).
  \end{equation}
  Analogously, if $c$ is a $1$-cell, using the graded commutativity of the wedge product, it follows that
  \begin{equation}
    \rel(a, b, c)
    = (-1)^{p} \rel(a, c, b)
    = (-1)^{p} \varepsilon(a, b)\, \varepsilon(c, \mathcal{N}).
  \end{equation}
  This equation is useful when $p \geq 2$ (since we already considered the cases $p = 0$ and $p = 1$).
  If $D \leq 3$ which, combined with $q = 1$ and $p + q \leq D$, leads to $p = 2$ and therefore
  \begin{equation}
    \rel(a, b, c) = \varepsilon(a, b)\, \varepsilon(c, \mathcal{N}).
  \end{equation}
  The exact form of $\rel$ when $D \leq 3$ is summarised in \Cref{def:discrete/exterior_calculus/perpendicular_orientations_up_to_3d},
  which is the definition we use in our applications where we always have $D \leq 3$.

  For completeness we will derive an intrinsic (recursive) expression for $\rel$ in the general case.
  Decompose the cell $b$ into orthogonal components $\mathcal{E}$ and $b'$
  (intersecting at $\mathcal{N}$) so that
  $\mathcal{E}$ is spanned by $e_1$ with orientation $s_{\mathcal{E}} e_1$, and
  $b'$ is spanned by $e_2, ..., e_p$ with orientation $s_{b'} e_2 \wedge ... \wedge e_p$,
  where $s_{\mathcal{E}}, s_{b'} \in \{-1, 1\}$.
  Let $a'$ be spanned by $e_2, ..., e_{p + q}$ with orientation $s_{a'} e_2 \wedge ... \wedge e_{p + q}$, where $s_{a'} \in \{-1, 1\}$.
  Then the equality
  \begin{equation}
    s_a s_b s_c = (s_a s_{\mathcal{E}} s_{a'})\, (s_{a'} s_{b'} s_c)\, (s_b s_{\mathcal{E}} s_{b'})
  \end{equation}
  can be restated as
  \begin{equation}
    \label{eq:discrete/exterior_calculus/perpendicular_orientations_recursive}
    \begin{split}
      \rel(a, b, c)
      & = \rel(a, \mathcal{E}, a')\, \rel(a', b', c)\, \rel(b, \mathcal{E}, b') \\
      & = (\varepsilon(a, a')\, \varepsilon(\mathcal{E}, \mathcal{N}))\, \rel(a', b', c)\, (\varepsilon(b, b')\, \varepsilon(\mathcal{E}, \mathcal{N})) \\
      & = \varepsilon(a, a')\, \varepsilon(b, b')\, \rel(a', b', c).
    \end{split}
  \end{equation}
  This gives the required recursive definition since $\dim b' = p - 1 < p = \dim b$.

  At a first glance it is not obvious that the recursive definition is independent of the decomposition of $b$.
  However, all the expressions are calculated over the subfaces of a single quasi-cubical cell $a$, which we can choose to embed as a cube, and use the vector space orientations as a justification for the correctness of the definition.
\end{discussion}
\begin{remark}
  Note that from the graded commutativity of the wedge product and \Cref{eq:discrete/exterior_calculus/perpendicular_orientations} it follows that
  \begin{equation}
    \label{eq:discrete/exterior_calculus/perpendicular_orientations_graded_commutativity}
    \rel(a, c, b) = (-1)^{p q} \rel(a, b, c).
  \end{equation}
\end{remark}
\begin{definition}
  \label{def:discrete/exterior_calculus/perpendicular_orientations_up_to_3d}
  Let $\mathcal{K}$ be a quasi-cubical mesh of dimension $D \leq 3$ with relative orientations $\varepsilon$,
  $0 \leq p, q$ with $p + q \leq D$,
  $a \in \mathcal{K}_{p + q}$, $(b, c) \in \perp_{p, q} a$,
  $\mathcal{N}$ be the node where $b$ and $c$ intersect.
  Then the \newterm{relative orthogonal orientation} is defined by
  \begin{equation}
    \rel(a, b, c) =
    \begin{cases}
      1, & p = 0\ \text{or}\ q = 0 \\
      \varepsilon(a, c)\, \varepsilon(b, \mathcal{N}), & p = 1\ \text{and}\ q \geq 1 \\
      \varepsilon(a, b)\, \varepsilon(c, \mathcal{N}), & p = 2\ \text{and}\ q = 1
    \end{cases}.
  \end{equation}
\end{definition}
\begin{definition}
  Let $\mathcal{K}$ be a quasi-cubical mesh of dimension $D$ with relative orientations $\varepsilon$.
  The \newterm{quasi-cubical cup product} is a bilinear map
  \begin{equation}
    \smile \colon C^\bullet \mathcal{K} \times C^\bullet \mathcal{K} \to C^\bullet \mathcal{K}, \
    \smile_{p, q} \colon C^p \mathcal{K} \times C^q \mathcal{K} \to C^{p + q} \mathcal{K},\ 0 \leq p, q \leq D,\ p + q \leq D
  \end{equation}
  defined as follows: for any $\sigma \in C^p \mathcal{K}, \tau \in C^q \mathcal{K}$, $a \in \mathcal{K}_{p + q}$,
  \begin{equation}
    (\sigma \smile \tau)(a_\bullet) := \frac{1}{2^{p + q}} \sum_{(b, c) \in \perp_{p, q} a} \rel(a, b, c)\,  \sigma(b_\bullet)\, \tau(c_\bullet).
  \end{equation}
\end{definition}
\begin{proposition}
  Let $\mathcal{K}$ be a quasi-cubical mesh.
  The cup product has the following properties:
  \begin{enumerate}
    \item
      graded commutativity: for any $\sigma \in C^p \mathcal{K}$ and $\tau \in C^q \mathcal{K}$,
      \begin{equation}
        \sigma \smile \tau = (-1)^{p q} \tau \smile \sigma.
      \end{equation}
    \item
      graded Leibniz rule: for any $\sigma \in C^p \mathcal{K}$ and $\tau \in C^\bullet \mathcal{K}$,
      \begin{equation}
        \delta(\sigma \smile \tau) = \delta \sigma \smile \tau + (-1)^p \sigma \smile \delta \tau.
      \end{equation}
  \end{enumerate}
\end{proposition}
\begin{remark}
  The above proposition was proved for cubes in \cite[Theorem 3.2.3]{arnold2012discrete} using a lengthy argument.
  The same proof can be used there since both identities have to be checked on single quasi-cubical cells (we first map cells of $\mathcal{K}$ to cubes).
  However, a straightforward proof can be done using the fact that cubical meshes with cup products can be defined as tensor products of graded-commutative differential graded algebras (the $1$D cubes (segments) together with their cup products).
  The base case (a segment) is trivially checked by brute force.
\end{remark}

\subsection{Metric operations}

\begin{remark}
  In this work we assume that topologically orthogonal $p$-cells of the mesh are also geometrically orthogonal which results in a diagonal inner product matrix. It will be discussed after one of the demonstration examples (\Cref{sec:simulations/neper}) that this choice may not work universally and amendments will be suggested.
\end{remark}
\begin{definition}
  Let $\mathcal{K}$ be a mesh of dimension $D$, representing a spatial body.
  A \newterm{metric on $\mathcal{K}$} is a function
  $\mu \colon \mathcal{K} \to \R^+$ assigning positive measures to cells in $\mathcal{K}$, and is of physical dimension $\length^p$ when applied at $p$-cells.
  We call the pair $(\mathcal{K}, \mu)$ a \newterm{a quasi-cubical Riemannian mesh}.
\end{definition}
\begin{remark}
  Let $\mathcal{K}$ be a mesh embedded in a Riemannian manifold $(K, g)$ via a map $\varphi \colon \mathcal{K} \to \mathcal{P} K$, $a \in \mathcal{K}$.
  We will usually choose $\mu(a)$ to be the measure of $\varphi(a)$, i.e., if $\vol_{\varphi(a)}$ is the volume form on $a$ corresponding to the pullback metric $\tr_{\varphi(a)} g$ on $\varphi(a)$, then
  \begin{equation}
    \mu(a) := \int_{\varphi(a)} \vol_{\varphi(a)}.
  \end{equation}
  $\mu(a)$ is $1$ for $a \in \mathcal{K}_0$, the length of $\varphi(a)$ for $a \in \mathcal{K}_1$, the area of $\varphi(a)$ for $a \in \mathcal{K}_2$, and the volume of $\varphi(a)$ for $a \in \mathcal{K}_3$.
\end{remark}
\begin{definition}
  Let $(\mathcal{K}, \mu)$ be a quasi-cubical Riemannian mesh.
  The \newterm{discrete inner product} on $\mathcal{K}$ is a family of bilinear forms on $C^p\mathcal{K}$ (for $p = 0, ..., D$) such that basis cochains form an orthogonal basis,
  and for any $c \in \mathcal{K}_p$,
  \begin{equation}
    \inner{c^\bullet}{c^\bullet}_p :=
    \frac{1}{2^D \mu(c)} \sum_{(a, b) \in \mathcal{K}_D \times \mathcal{K}_{D - p},\ b \perp_a c} \mu(b).
  \end{equation}
  The physical dimension of $\inner{\cdot}{\cdot}_p$ is $\length^{D - 2 p}$.
\end{definition}
\begin{remark}
  Note that in convex meshes any orthogonal cells $b$ and $c$ have a unique common $D$-cell $a$.
  Hence, the above formula becomes \cite[Equation 2.90]{berbatov2023discrete}:
  \begin{equation}
    \inner{c^\bullet}{c^\bullet}_p :=
    \frac{1}{2^D \mu(c)} \sum_{b \perp c} \mu(b).
  \end{equation}
\end{remark}
\begin{definition}
  Let $\mathcal{K}$ be a quasi-cubical oriented Riemannian mesh of dimension $D$,
  $0 \leq p \leq D$.
  The \newterm{discrete Hodge star operator}
  \begin{equation}
    \star_p \colon C^p \mathcal{K} \to C^{D - p} \mathcal{K}
  \end{equation}
  is the unique operator satisfying the following equation:
  for arbitrary cochains $\pi \in C^p \mathcal{K}$ and $\rho \in C^{D - p} \mathcal{K}$,
  \begin{equation}
    \label{eq:combinatorial/inner_product_with_hodge_star}
     \inner{\star_p \pi}{\rho}_{D - p}
     = (\pi \smile \rho)[\mathcal{K}].
  \end{equation}
  Note, that the right hand side expresses a discrete form of integration, and, hence, has the same form as \Cref{eq:exterior_calculus/inner_product_with_hodge_star}.
  In the standard bases of $C^p \mathcal{K}$ and $C^{D - p} \mathcal{K}$
  the discrete Hodge star $\star_p$ is represented as a sparse matrix with dimension $N_{D - p} \times N_p$,
  where the nonzero entries are those corresponding to pairs of topologically orthogonal cells.

  The physical dimension of $\star_p$ is $\length^{D - 2 p}$, just like in the continuum case.
\end{definition}
\begin{proposition}
  Let $\mathcal{K}$ be an oriented quasi-cubical Riemannian mesh of dimension $D$,
  $p \in \{0, ..., D\}$,
  $\sigma \in C^p \mathcal{K}$,
  $c \in \mathcal{K}_{D - p}$.
  Then
  \begin{equation}
    \begin{split}
      (\star_p \sigma)(c_\bullet)
      & = \frac{1}{\inner{c^\bullet}{c^\bullet}_{D - p}} \sum_{(a, b) \in \mathcal{K}_D \times \mathcal{K}_{D - p},\ b \perp_a c} (b^\bullet \smile c^\bullet)[\mathcal{K}]\, \sigma(b_\bullet) \\
      & = \frac{1}{2^D \inner{c^\bullet}{c^\bullet}_{D - p}} \sum_{(a, b) \in \mathcal{K}_D \times \mathcal{K}_{D - p},\ b \perp_a c} \rel(a, b, c)\, \sigma(b_\bullet).
    \end{split}
  \end{equation}
\end{proposition}
\begin{remark}
  A proof for a similar formulation on a convex mesh is derived at \cite[Equation 16]{berbatov2022diffusion}.
  Note that in \cite{berbatov2022diffusion,berbatov2023discrete} a different convention for Hodge star is used but this one coincides with the standard one, used in exterior algebras.
\end{remark}
\begin{definition}
  Let $\mathcal{K}$ be a compatibly oriented quasi-cubical Riemannian mesh.
  The \newterm{adjoint coboundary operator} $\delta^\star_p$ is the discrete analogue of the codifferential.
  It is defined as the adjoint of $\tilde{\delta}_{p - 1}$ with respect to the inner product (restricted to zero-trace cochains),
  that is, for any $\pi \in C^p_0 \mathcal{K}$ and $\rho \in C^{p - 1}_0 \mathcal{K}$,
  \begin{equation}
    \inner{\pi}{\tilde{\delta}_{p - 1} \rho}_p = \inner{\delta^\star_p \pi}{\rho}_{p - 1}.
  \end{equation}
  The adjoint coboundary operator $\delta^\star_p$ is represented as a sparse matrix with dimension $\abs{C^{p - 1}_0 \mathcal{K}} \times \abs{C^p_0 \mathcal{K}}$ with the same stencil as the boundary operator $\partial_p$ (when restricted to boundary non-boundary cells): only magnitudes of values differ from those of the boundary operator, while the signs are the same.

  The physical dimension of $\delta^\star_p$ is $\length^{-2}$, like in the continuum.
\end{definition}
\begin{proposition}
  Let $\mathcal{K}$ be a compatibly oriented quasi-cubical Riemannian mesh,
  $p \in \{1, ..., D\}$,
  $\sigma \in C^p_0 \mathcal{K}, b \in {\rm interior}(\mathcal{K}_{p - 1})$.
  Then \cite[Equation 14]{berbatov2022diffusion}
  \begin{equation}
    (\delta^\star_p \sigma)(b_\bullet) = \frac{1}{\inner{b^\bullet}{b^\bullet}_{p - 1}} \sum_{a \succdot b} \varepsilon(a, b)\, \inner{a^\bullet}{a^\bullet}_p\, \sigma(a_\bullet).
  \end{equation}
\end{proposition}
\begin{proposition}
  Let $\mathcal{K}$ be a compatibly oriented quasi-cubical Riemannian mesh of dimension $D$.
  Then
  \begin{equation}
    \delta^\star_{D - p} \circ \star_p = (-1)^{p + 1} \star_{p + 1} \circ \tilde{\delta}_p.
  \end{equation}
\end{proposition}
\begin{proof}
  The proof is essentially the same as the proof of \Cref{eq:exterior_calculus/codifferential_formula}, given by \Cref{eq:exterior_calculus/codifferential_formula_proof}.
\end{proof}
\begin{remark}
  We are ready to reformulate the exterior calculus models in a discrete setting.
  Before doing that, in \Cref{tab:calculus_summary} we summarise both formalisms (continuous and discrete) and how they relate to one another.
  The central column, ``Correspondence'', describes the level of accuracy between a continuous and a discrete notion.
  More precisely, some relationships are exact by definition, namely those relating objects by discretising them.
  For relationships between operators, a commutative diagram approach is used -- an operator is exact if first discretising it, and then applying it on discrete objects gives the same result as first applying it at continuous objects and then discretising them.

  For instance, all topological operators except the wedge/cup products correspond exactly.
  The accuracy of the correspondence between the wedge product and the cup product on a simplicial mesh is discussed in detail in \cite{wilson2007cochain};
  one of the crucial differences is the non-associativity of the cup product.
  For metric operators, all relationships are approximate, which is inescapable.
  However, our definitions are based on the following principles:
  \begin{enumerate}
    \item
      the Hodge star and adjoint coboundary operators have analogues definitions to their continuous counterparts, so that they share some common algebraic expressions with them;
    \item
      the basis cochains are orthogonal with respect to the inner product so that discrete Hodge star and adjoint coboundary are local operators;
    \item
      the discrete counterparts of the volume form and the identity function are discrete Hodge-dual;
    \item
      for regular cubical grids of mesh size $h > 0$, the discrete operators have optimal accuracy ($O(h^2)$ in the interior) when viewed as approximations of their respective continuous counterparts.
  \end{enumerate}
  That being said, the discrete calculus presented here (CMC) is fully intrinsic and self-contained, i.e., it should be considered as independent of the continuous one. The smooth exterior calculus is primarily a source of inspiration on how to build the discrete one.
\end{remark}
\begin{table}[!ht]
  \caption{Summary of Exterior calculus and Combinatorial mesh calculus}
  \label{tab:calculus_summary}
  \centering
  \begin{tabular}{|l|l|l||l||l|l|l|}
    \hline
    \multicolumn{3}{|c||}{Exterior calculus}
    & Correspondence
    & \multicolumn{3}{|c|}{Combinatorial mesh calculus} \\
    \hline
    Quantity/Operator & Expression & Dimension
    &
    & Quantity/Operator & Expression & Dimension \\
    \hline
    Manifold & $M$ & $1$
    & Exact
    & Mesh & $\mathcal{M}$ & $1$ \\
    Submanifold & $V$ & $1$
    & Exact
    & Cell & $a$ & $1$ \\
    Vector field & $X$ & Generic
    & --
    & -- & -- & -- \\
    Differential form & $\omega$ & Generic
    & Exact
    & Cochain & $\sigma$ & Generic \\
    Boundary & $\partial$ & $1$
    & Exact
    & Boundary operator & $\partial$ & $1$ \\
    Exterior derivative & $d$ & $1$
    & Exact
    & Coboundary operator & $\delta$ & $1$ \\
    Integration & $\int_M \omega$ & $1$
    & Exact
    & Discrete integration & $\sigma [\mathcal{M}]$ & $1$ \\
    Trace & $\tr \omega$ & $1$
    & Exact
    & Discrete trace & $\tr \sigma$ & $1$ \\
    Wedge product & $\wedge$ & $1$
    & Approximate
    & Cup product & $\smile$ & $1$ \\
    Metric tensor & $g^*_p$ & $\length^{- 2 p}$
    & --
    & -- & -- & -- \\
    Inner product & $\inner{\cdot}{\cdot}_p$ & $\length^{D - 2 p}$
    & Approximate
    & Discrete inner product & $\inner{\cdot}{\cdot}_p$ & $\length^{D - 2 p}$ \\
    Hodge star & $\star_p$ & $\length^{D - 2 p}$
    & Approximate
    & Discrete Hodge star & $\star_p$ & $\length^{D - 2 p}$ \\
    Codifferential & $d^{\star}_p$ & $\length^{-2}$
    & Approximate
    & Adjoint coboundary operator & $\delta^{\star}_p$ & $\length^{-2}$ \\
    \hline
  \end{tabular}
\end{table}

\section{Combinatorial mesh calculus variational formulations}
\label{sec:discrete}

\begin{remark}
  The exterior calculus formulations in \Cref{sec:continuum} can be translated almost literally into CMC formulations.
  Because discrete Hodge star operators are generally not invertible, primal
  and mixed formulations are not equivalent at the discrete level.
  In the present work, each formulation is therefore treated as a distinct
  discrete model rather than as an algebraically equivalent representation of
  the same problem.
\end{remark}
\begin{notation}[Parameters participating in the CMC model for transport phenomena]
  \label{notation:discrete/parameters}
  Before imposing the (strong) model, let us summarise all the parameters that will be used in all subsequent reformulations.
  Until the end of this section let:
  \begin{itemize}
    \item
      $D$ be a positive integer (space dimension);
    \item
      $\mathcal{M}$ be a $D$-dimensional oriented Riemannian mesh;
    \item
      $\mathcal{K}$ be the Forman subdivision of $\mathcal{M}$;
    \item
      $\Gamma_{\Dirichlet}, \Gamma_{\Neumann}$ form a partition of $\partial \mathcal{K}$ into Dirichlet and Neumann boundary
      (discrete analogue of \Cref{eq:continuum/boundary_decomposition}).
    \item
      $t_0 [\time] \in \R$ be the initial time;
    \item
      $I := [t_0, \infty)$;
  \end{itemize}
  Input parameters, initial and boundary conditions are given in \Cref{tab:discrete/parameters}.
  The unknowns are given in \Cref{tab:discrete/unknowns}.
  $S_0$ and $\widetilde{S_0}$ are defined as in
  \Cref{eq:continuum/amount_potential_constant} and
  \Cref{eq:continuum/dual_amount_potential_constant} respectively.
\end{notation}
\begin{table}[!ht]
  \caption{Transport phenomena parameters in the CMC formulation}
  \label{tab:discrete/parameters}
  \centering
  \begin{tabular}{|l|l|l|l|l|l|}
    \hline
    Quantity
    & Symbol
    & Domain
    & Dimension
    & Pseudo-object?
    & Continuum analogue \topStrut \\[2pt]
    \hline
    \hline
    Initial potential
    & $u_0$
    & $C^0 \mathcal{K}$
    & $\potential$
    & No
    & \Cref{eq:continuum/initial_potential_definition} \topStrut \\[2pt]
    \hline
    Initial amount
    & $Q_0$
    & $C^D \mathcal{K}$
    & $\amount$
    & Yes
    & \Cref{eq:continuum/initial_amount_definition} \topStrut \\[2pt]
    \hline
    Internal production rate
    & $f$
    & $\mathcal{C}^\infty(I, C^D \mathcal{K})$
    & $\amount \time^{-1}$
    & Yes
    & \Cref{eq:continuum/internal_production_rate_definition} \topStrut \\[2pt]
    \hline
    Prescribed flow rate
    & $g_\Neumann$
    & $\mathcal{C}^\infty(I, C^{D - 1} \Gamma_\Neumann)$
    & $\amount \time^{-1}$
    & Yes
    & \Cref{eq:continuum/prescribed_neumann_condition_definition} \topStrut \\[2pt]
    \hline
    Prescribed potential
    & $g_\Dirichlet$
    & $\mathcal{C}^\infty(I, C^0 \Gamma_\Dirichlet)$
    & $\potential$
    & No
    & \Cref{eq:continuum/prescribed_dirichlet_condition_definition} \topStrut \\[2pt]
    \hline
    Volumetric flow rate
    & $v$
    & $\mathcal{C}^\infty(I, C^{D - 1} \mathcal{K})$
    & $\length^D \time^{-1}$
    & Yes
    & \Cref{eq:continuum/volumetric_flow_rate_definition} \topStrut \\[2pt]
    \hline
    Capacity
    & $\pi$
    & $\mathcal{C}^\infty(I, C^D \mathcal{K}) \to \mathcal{C}^\infty(I, C^D \mathcal{K})$
    & $\amount \potential^{-1} \length^{-D}$
    & No
    & \Cref{eq:continuum/dual_capacity_definition} \topStrut \\[2pt]
    \hline
    Dual capacity
    & $\widetilde{\pi}$
    & $\mathcal{C}^\infty(I, C^0 \mathcal{K}) \to \mathcal{C}^\infty(I, C^0 \mathcal{K})$
    & $\amount \potential^{-1} \length^{-D}$
    & No
    & \Cref{eq:continuum/capacity_definition} \topStrut \\[2pt]
    \hline
    Conductivity
    & $\kappa$
    & $\mathcal{C}^\infty(I, C^{D - 1} \mathcal{K}) \to \mathcal{C}^\infty(I, C^{D - 1} \mathcal{K})$
    & $\amount \potential^{-1} \length^{2 - D} \time^{-1}$
    & No
    & \Cref{eq:continuum/conductivity_definition} \topStrut \\[2pt]
    \hline
    Dual conductivity
    & $\widetilde{\kappa}$
    & $\mathcal{C}^\infty(I, C^1 \mathcal{K}) \to \mathcal{C}^\infty(I, C^1 \mathcal{K})$
    & $\amount \potential^{-1} \length^{2 - D} \time^{-1}$
    & No
    & \Cref{eq:continuum/dual_conductivity_definition} \topStrut \\[2pt]
    \hline
  \end{tabular}
\end{table}
\begin{table}[!ht]
  \caption{Transport phenomena unknowns in the CMC formulation}
  \label{tab:discrete/unknowns}
  \centering
  \begin{tabular}{|l|l|l|l|l|l|}
    \hline
    Quantity
    & Symbol
    & Domain
    & Dimension
    & Pseudo-object?
    & Continuum reference \topStrut \\[2pt]
    \hline
    \hline
    Amount
    & $Q$
    & $\mathcal{C}^\infty(I, C^D \mathcal{K})$
    & $\amount$
    & Yes
    & \Cref{eq:continuum/amount_definition} \topStrut \\[2pt]
    \hline
    Flow rate
    & $q$
    & $\mathcal{C}^\infty(I, C^{D - 1} \mathcal{K})$
    & $\amount \time^{-1}$
    & Yes
    & \Cref{eq:continuum/flow_rate_definition} \topStrut \\[2pt]
    \hline
    Diffusive flow rate
    & $q_\diffusive$
    & $\mathcal{C}^\infty(I, C^{D - 1} \mathcal{K})$
    & $\amount \time^{-1}$
    & Yes
    & \Cref{eq:continuum/diffusive_flow_rate_definition} \topStrut \\[2pt]
    \hline
    Advective flow rate
    & $q_\advective$
    & $\mathcal{C}^\infty(I, C^{D - 1} \mathcal{K})$
    & $\amount \time^{-1}$
    & Yes
    & \Cref{eq:continuum/advective_flow_rate_definition} \topStrut \\[2pt]
    \hline
    Potential
    & $u$
    & $\mathcal{C}^\infty(I, C^0 \mathcal{K})$
    & $\potential$
    & No
    & \Cref{eq:continuum/potential_definition} \topStrut \\[2pt]
    \hline
  \end{tabular}
\end{table}
\begin{remark}
  The input physical quantities ($u_0$, $f$, $g_\Neumann$, $g_\Dirichlet$, and $v$) can be provided either directly in discrete form or as discretised versions of their continuous counterparts via the de Rham map -- see \Cref{discussion:simulations/pre_processing} for further details. Material properties are assigned to the cells of the complex: capacity is associated with $3$-cells (and its dual with $0$-cells), while conductivity is associated with $(D-1)$-cells (and its dual with $1$-cells).   For example, for a $1$-cell $c$, we define $(\tilde{\kappa} q)(c) = \lambda_c\, q(c)$, where $\lambda_c > 0$ is a material-specific coefficient. Analogous relations hold for the operators $\pi$, $\kappa$, and $\tilde{\pi}$.

  Our approach is designed to predict or derive macroscopic (emergent) properties of the complex from given local, cell-wise properties. This enables the analysis of heterogeneous structures with arbitrary geometric and topological complexity. The inverse problem -- determining local properties that give rise to prescribed macroscopic behaviour -- is a compelling and important challenge, especially in the context of materials design with target functionalities. While addressing this inverse problem is beyond the scope of the present work, it could be pursued via various machine learning strategies.

  In this paper, we consider the simpler case of homogeneous and isotropic materials with prescribed macroscopic conductivity and capacity. In such cases, the inverse problem is trivial: all cells associated with conductivity or capacity are assigned the corresponding macroscopic values. This setup enables direct comparison with classical continuum models in our simulation examples.

  It is important to note that, since the discrete Hodge star operators are not invertible, material coefficients and their dual counterparts are not interchangeable. In the following two points, we explain how conductivities and capacities are defined in the discrete setting.
\end{remark}
\begin{discussion}[Specifying discrete conductivities]
  Following \cite{berbatov2022diffusion}, we define conductivity on $1$-cells -- this corresponds to the dual conductivity $\tilde{\kappa}$ in this work. The $1$-cells in $\mathcal{K}$ correspond to intervals $c$ in the physical complex $\mathcal{M}$, where $\dim b = \dim a + 1$ for $[a, b] = c$.

  Conductivities are specified as positive scalar values on non-nodal cells of $\mathcal{M}$. Each $1$-cell $[a, b]$ in $\mathcal{K}$ is assigned the conductivity value of the target cell $b$. The corresponding primal conductivity $\kappa$ is represented as a diagonal matrix, determined via a least-squares fit to satisfy either
  $\star_1 \, \tilde{\kappa} = \kappa \, \star_1$ or
  $\star_{D - 1} \, \kappa = \tilde{\kappa} \, \star_{D - 1}$.
\end{discussion}
\begin{discussion}[Specifying discrete capacities]
  In parallel with the conductivity specification, we begin by assigning the dual capacity $\tilde{\pi}$. Note that $0$-cells in $\mathcal{K}$ correspond to all cells in $\mathcal{M}$.
  Capacities are given as positive scalar values on $p$-cells of $\mathcal{M}$ for $p>0$ only. For any $0$-cell $[a, a]$ in $\mathcal{K}$, we assign the capacity of $a$ if $a$ is not a $0$-cell of $\mathcal{M}$, and zero otherwise.  The corresponding primal capacity $\pi$ is again represented as a diagonal matrix, determined via a least-squares fit to satisfy either
  $\star_0 \, \tilde{\pi} = \pi \, \star_0$ or
  $\star_D \, \pi = \tilde{\pi} \, \star_D$.
\end{discussion}

\begin{formulation}[CMC transient model for transport phenomena]
  By translating \Cref{tab:continuum/equations} into the language of combinatorial meshes, under the assumptions of \Cref{notation:discrete/parameters}, we arrive at the model given in \Cref{tab:discrete/equations}.
\end{formulation}
\begin{table}[!ht]
  \caption{Governing equations for transport phenomena in the CMC formulation}
  \label{tab:discrete/equations}
  \centering
  \begin{tabular}{|ll|l|l|l|}
    \hline
    Equation
    &
    & Dimension
    & Law
    & Continuum reference \topStrut \\[2pt]
    \hline
    \hline
    $\frac{\partial Q}{\partial t}$
    & $= f - \delta_{D - 1} q$
    & $\amount \time^{-1}$
    & Conservation law
    & \Cref{eq:continuum/conservation_law_differential_form} \topStrut \\[2pt]
    \hline
    $\frac{\partial Q}{\partial t}$
    & $= \star_0 \widetilde{\pi} \frac{\partial u}{\partial t} = \pi \star_0 \frac{\partial u}{\partial t}$
    & $\amount \time^{-1}$
    & Relation between potential and amount
    & \Cref{eq:continuum/potential_to_amount} \topStrut \\[2pt]
    \hline
    $q$
    & $= q_\diffusive + q_\advective$
    & $\amount \time^{-1}$
    & Flow rate decomposition
    & \Cref{eq:continuum/flow_rate_decomposition} \topStrut \\[2pt]
    \hline
    $q_\diffusive$
    & $= \kappa d^\star_D \star_0 u$
    & $\amount \time^{-1}$
    & Constitutive law
    & \Cref{eq:continuum/constitutive_law} \topStrut \\[2pt]
    \hline
    $q_\advective$
    & $= (\star_D Q) \smile v$
    & $\amount \time^{-1}$
    & Advective flow rate formula
    & \Cref{eq:continuum/advective_flow_rate_formula} \topStrut \\[2pt]
    \hline
    $\tr_{\Gamma_{\Dirichlet, 0}}(u)$
    & $= g_\Dirichlet$
    & $\potential$
    & Dirichlet boundary condition
    & \Cref{eq:continuum/dirichlet_boundary_condition} \topStrut \\[2pt]
    \hline
    $\tr_{\Gamma_{\Neumann, D - 1}}(q)$
    & $= g_\Neumann$
    & $\amount \time^{-1}$
    & Neumann boundary condition
    & \Cref{eq:continuum/neumann_boundary_condition} \topStrut \\[2pt]
    \hline
    $u(t_0)$
    & $= u_0$
    & $\potential$
    & Initial condition for potential
    & \Cref{eq:continuum/initial_condition_potential} \topStrut \\[2pt]
    \hline
    $Q(t_0)$
    & $= Q_0$
    & $\amount$
    & Initial condition for amount
    & \Cref{eq:continuum/initial_condition_amount} \topStrut \\[2pt]
    \hline
  \end{tabular}
\end{table}

\subsection{Primal weak formulation}
\label{sec:discrete/primal_weak}

\begin{formulation}[CMC transient primal weak formulation for transport phenomena]
  \label{formulation:discrete/primal_weak/transient}
  The following formulation is a discrete version of \Cref{formulation:continuum/primal_weak/transient}.
  Under the assumptions of \Cref{notation:discrete/parameters}
  define the following operators:
  \begin{subequations}
    \label{eq:discrete/primal_weak/operators}
    \begin{alignat}{3}
      & A_{\diffusive} \colon C^0 \mathcal{K} \times \mathcal{C}^\infty(I, C^0 \mathcal{K}) \to \R, \quad
      && A_{\diffusive}(w, u) := \inner{\delta_0 w}{\widetilde{\kappa} \delta_0 u}_{\mathcal{K}, 1} \qquad
      && [\amount \time^{-1} \potential^{-1}], \\
      & A_{\advective} \colon C^0 \mathcal{K} \times \mathcal{C}^\infty(I, C^0 \mathcal{K}) \to \R, \quad
      && A_{\advective}(w, u) := (\delta_0 w \smile (\widetilde{\pi} u \smile v))[\mathcal{K}] \qquad
      && [\amount \time^{-1} \potential^{-1}], \\
      & B \colon C^0 \mathcal{K} \times \mathcal{C}^\infty(I, C^0 \mathcal{K}) \to \R, \quad
      && B(w, u) := \inner{w}{\widetilde{\pi} u}_{\mathcal{K}, 0} \qquad
      && [\amount \potential^{-1}], \\
      & G \colon C^0 \mathcal{K} \to \R, \quad
      && G(w)
        := (\tr_{\Gamma_{\Neumann}} w \smile g_{\Neumann})[\Gamma_{\Neumann}] \qquad
      && [\amount \time^{-1}], \\
      & F_1 \colon C^0 \mathcal{K} \to \R, \quad
      && F_1(w) := (w \smile f)[\mathcal{K}] \qquad
      && [\amount \time^{-1}], \\
      & F_2 \colon C^0 \mathcal{K} \to \R, \quad
      && F_2(w)
        := (\delta_0 w \smile (\widetilde{S_0} \smile v))[\mathcal{K}] \qquad
      && [\amount \time^{-1}].
    \end{alignat}
  \end{subequations}
  The unknown variable is the potential
  $u [\potential] \in \mathcal{C}^\infty(I, C^0 \mathcal{K})$.
  We are solving the following system for $u$:
  \begin{subequations}
    \label{eq:discrete/primal_weak/transient_formulation}
    \begin{alignat}{4}
      & \forall w [\potential] \in \Ker \tr_{\Gamma_{\Dirichlet}, 0}, \quad
      && B(w, \frac{\partial u} {\partial t}) + A_{\diffusive}(w, u) - A_{\advective}(w, u)
      && = F_1(w) + F_2(w) - G(w) \qquad
      && [\amount \time^{-1} \potential], \\
      &
      && \tr_{\Gamma_{\Dirichlet}, 0}(u)
      && = g_{\Dirichlet} \qquad
      && [\potential], \\
      &
      && u(t_0)
      && = u_0 \qquad
      && [\potential].
    \end{alignat}
  \end{subequations}
  The flow rate $q [\amount \time^{-1}] \in \mathcal{C}^\infty(I, C^{D - 1} \mathcal{K})$
  is calculated in the post-processing phase as follows: for any $t \in I,\ c \in \mathcal{K}_{D - 1}$,
  \begin{equation}
    \label{eq:discrete/primal_weak/transient_post_processing}
    q(t, c_\bullet) =
    \begin{cases}
      (-\star_1\widetilde{\kappa} \delta_0 u + (\widetilde{S_0} + \widetilde{\pi} u) \smile v)(t, c_\bullet),
        & c \notin (\Gamma_\Neumann)_{D - 1} \\
      g_{\Neumann}(t, c_\bullet), & c \in (\Gamma_\Neumann)_{D - 1}
    \end{cases}.
  \end{equation}
\end{formulation}

\begin{formulation}[CMC steady-state primal weak formulation for transport phenomena]
  \label{formulation:discrete/primal_weak/steady_state}
  Take the time-independent versions of the assumptions in \Cref{notation:discrete/parameters} and the operators in \Cref{eq:discrete/primal_weak/operators}.
  The unknown variable is the potential
  $u [\potential] \in C^0 \mathcal{K}$.
  We are solving the following system for $u$:
  \begin{subequations}
    \label{eq:discrete/primal_weak/steady_state_formulation}
    \begin{alignat}{4}
      & \forall w [\potential] \in \Ker \tr_{\Gamma_{\Dirichlet}, 0}, \quad
      && A_{\diffusive}(w, u) - A_{\advective}(w, u)
      && = F_1(w) + F_2(w) - G(w) \qquad
      && [\amount \time^{-1} \potential], \\
      &
      && \tr_{\Gamma_{\Dirichlet}, 0}(u)
      && = g_{\Dirichlet} \qquad
      && [\potential].
    \end{alignat}
  \end{subequations}
  The flow rate $q [\amount \time^{-1}] \in C^{D - 1} \mathcal{K}$
  is calculated in the post-processing phase as follows: for any $c \in \mathcal{K}_{D - 1}$,
  \begin{equation}
    \label{eq:discrete/primal_weak/steady_state_post_processing}
    q(c_\bullet) =
    \begin{cases}
      (-\star_1 \widetilde{\kappa} \delta_0 u + \widetilde{Q_0} \smile v)(c_\bullet),
        & c \notin (\Gamma_\Neumann)_{D - 1} \\
      g_{\Neumann}(c_\bullet), & c \in (\Gamma_\Neumann)_{D - 1}
    \end{cases}.
  \end{equation}
\end{formulation}

\subsection{Mixed weak formulation}
\label{subsec:discrete/mixed_weak}

\begin{formulation}[CMC transient mixed weak formulation for transport phenomena]
  \label{formulation:discrete/mixed_weak/transient}
  The following formulation is a discrete version of \Cref{formulation:continuum/mixed_weak/transient}.
  Under the assumptions of \Cref{notation:discrete/parameters}
  define the following operators
  \begin{subequations}
    \label{eq:discrete/mixed_weak/operators}
    \begin{alignat}{3}
      & A \colon C^{D - 1} \mathcal{K} \times \mathcal{C}^\infty(I, C^{D - 1} \mathcal{K}) \to \R,
        \;
      && A(r, s)
        := \inner{r}{\kappa^{-1} s}_{\mathcal{K}, D - 1} \;
      && [\amount^{-1} \time \potential], \\
      & B_{\diffusive} \colon C^D \mathcal{K} \times \mathcal{C}^\infty(I, C^{D - 1} \mathcal{K}) \to \R, \;
      && B_{\diffusive}(\widetilde{w}, r)
        := \inner{d_{D - 1} r}{\widetilde{w}}_{\mathcal{K}, D} \;
      && [\length^{-D}], \\
      & B_{\advective} \colon C^D \mathcal{K} \times \mathcal{C}^\infty(I, C^{D - 1} \mathcal{K}) \to \R, \;
      && B_{\advective}(\widetilde{w}, r)
        := \inner{r}{\kappa^{-1} (\star_D \pi \widetilde{w} \smile v)}_{\mathcal{K}, D - 1} \;
      && [\length^{-D}], \\
      & C \colon C^D \mathcal{K} \times \mathcal{C}^\infty(I, C^D \mathcal{K}) \to \R, \;
      && C(\widetilde{w}, \widetilde{u}) := \inner{\pi \widetilde{u}}{\widetilde{w}}_{\mathcal{K}, D} \;
      && [\amount \length^{-2D} \potential^{-1}], \\
      & G_1 \colon C^{D - 1} \mathcal{K} \to \R, \;
      && G_2(r)
        := (\tr_{\Gamma_{\Dirichlet}, D - 1} r \smile g_{\Dirichlet})[\Gamma_{\Dirichlet}] \;
      && [\potential], \\
      & G_2 \colon C^{D - 1} \mathcal{K} \to \R, \
      && G_2(r) := \inner{r}{\kappa^{-1} (\widetilde{S_0} \smile v)}_{\mathcal{K}, D - 1} \
      && [\potential], \\
      & F \colon C^D \mathcal{K} \to \R, \;
      && F(\widetilde{w}) := \inner{f}{\widetilde{w}}_{\mathcal{K}, D} \;
      && [\amount \time^{-1} \length^{-D}], \\
      & \mathtt{flow\_rate} \colon C^0 \mathcal{K} \to C^{D - 2} \mathcal{K}, \;
      && \mathtt{flow\_rate}(w) := \kappa \delta^\star_D w + (\widetilde{S_0} + \tilde{\pi} w) \smile v\;
      && [\amount \time^{-1} \potential^{-1}].
    \end{alignat}
  \end{subequations}
  The unknown variables are:
  \begin{itemize}
    \item
      $q [\amount \time^{-1}] \in \mathcal{C}^\infty(I, C^{D - 1} \mathcal{K})$ (flow rate);
    \item
      $\widetilde{u} [\potential \length^D] \colon \mathcal{C}^\infty(I, C^D \mathcal{K})$ (dual potential).
  \end{itemize}
  We are solving the following problem for $q$ and $\widetilde{u}$:
  \begin{subequations}
    \label{eq:discrete/mixed_weak/transient_formulation}
    \begin{alignat}{4}
      & \forall r [\amount \time^{-1}] \in \Ker \tr_{\Gamma_{\Neumann}, D - 1}, \;
      && A(r, q) - B_{\diffusive}^T(r, \widetilde{u}) - B_{\advective}^T(r, \widetilde{u})
      && = G_2(r) - G_1(r) \;
      && [\amount \time^{-1} \potential], \\
      & \forall \widetilde{w} [\potential \length^D] \in C^D \mathcal{K}, \;
      && - B_{\diffusive}(\widetilde{w}, q) - C(\widetilde{w}, \frac{\partial \widetilde{u}}{\partial t})
      && = - F(\widetilde{w}) \;
      && [\amount \time^{-1} \potential], \\
      &
      && \tr_{\Gamma_{\Neumann}, D - 1} q
      && = g_{\Neumann} \;
      && [\amount \time^{-1}], \\
      &
      && \widetilde{u}(t_0)
      && = \star_0 u_0 \;
      && [\potential \length^D], \\
      &
      && q(t_0)
      && = \mathtt{flow\_rate}(u_0) \;
      && [\amount \time^{-1}].
    \end{alignat}
  \end{subequations}
  The potential $u [\potential] \in \mathcal{C}^\infty(I, C^0 \mathcal{K})$ is calculated in the post-processing phase as follows: for any $t \in I,\ c \in \mathcal{K}_0$,
  \begin{equation}
    \label{eq:discrete/mixed_weak/transient_post_processing}
    u(t, c_\bullet) :=
    \begin{cases}
      u_0(c_\bullet), & t = t_0 \\
      (\star_D \widetilde{u})(t, c_\bullet), & t > t_0\ \text{and}\ c \notin (\Gamma_{\Dirichlet})_0 \\
      g_{\Dirichlet}(t, c_\bullet), & t > t_0\ \text{and}\ c \in (\Gamma_{\Dirichlet})_0
    \end{cases}.
  \end{equation}
\end{formulation}
\begin{formulation}[CMC steady-state mixed weak formulation for transport phenomena]
  \label{formulation:discrete/mixed_weak/steady_state}
  Take the time-independent versions of the assumptions in \Cref{notation:discrete/parameters} and the operators in \Cref{eq:discrete/mixed_weak/operators}.
  The unknown variables are:
  \begin{itemize}
    \item
      $q [\amount \time^{-1}] \in C^{D - 1} \mathcal{K}$ (flow rate);
    \item
      $\widetilde{u} [\potential \length^D] \in C^D \mathcal{K}$ (dual potential).
  \end{itemize}
  We are solving the following problem for $q$ and $\widetilde{u}$:
  \begin{subequations}
    \label{eq:discrete/mixed_weak/steady_state_formulation}
    \begin{alignat}{4}
      & \forall r [\amount \time^{-1}] \in \Ker \tr_{\Gamma_{\Neumann}, D - 1}, \;
      && A(r, q) - B_{\diffusive}^T(r, \widetilde{u}) - B_{\advective}^T(r, \widetilde{u})
      && = G_2(r) - G_1(r) \;
      && [\amount \time^{-1} \potential], \\
      & \forall \widetilde{w} [\potential \length^D] \in C^D \mathcal{K}, \;
      && - B_{\diffusive}(\widetilde{w}, q)
      && = - F(\widetilde{w}) \;
      && [\amount \time^{-1} \potential], \\
      &
      && \tr_{\Gamma_{\Neumann}, D - 1} q
      && = g_{\Neumann} \;
      && [\amount \time^{-1}].
    \end{alignat}
  \end{subequations}
  The potential $u [\potential] \in C^0 \mathcal{K}$ is calculated in the post-processing phase as follows: for any $c \in \mathcal{K}_0$,
  \begin{equation}
    \label{eq:discrete/mixed_weak/steady_state_post_processing}
    u(c_\bullet) :=
    \begin{cases}
      (\star_D \widetilde{u})(c_\bullet), &  c \notin (\Gamma_{\Dirichlet})_0 \\
      g_{\Dirichlet}(c_\bullet), & c \in (\Gamma_{\Dirichlet})_0
    \end{cases}.
  \end{equation}
\end{formulation}

\subsection{Comparison with other methods}

\begin{discussion}
  Two methods that are notoriously based on differential geometry and meshes (i.e, forms and cochains) are discrete exterior calculus (DEC, \cite{hirani2003discrete}) and finite element exterior calculus (FEEC, \cite{arnold2018finite}).

  Recent work has clarified the close relationship between discrete exterior
  calculus and finite element exterior calculus through generalised Whitney
  forms \cite{guzman2025framework}.
  The present framework aligns with this perspective while emphasising a
  combinatorial organisation of variables and operators, together with a
  physically grounded interpretation and workflow oriented toward multi-physics
  coupling on cell complexes.

  FEEC is a formalism for creating finite element spaces in a structured way based on smooth exterior calculus.
  It is purely a discretisation method and cannot accommodate material properties of microstructural features of different dimensions.
  Polynomial finite elements are generally constructed via polynomial mappings of reference elements and cannot handle general curvilinear meshes, although a large class of such elements can be handled with polynomial approximations \cite{botti2018assesment}.
  Virtual finite elements, on the other hand, can handle general shapes \cite{beirao2019virtual}.

  In CMC the shapes of the cells do not matter for calculations as all the required information is encoded in the face lattice, relative orientations and cell measures -- everything else is canonically represented.
  The shapes may matter in the pre-processing phase if the discrete input comes from a continuous problem.
  The embedding is also required for visualisation.
\end{discussion}

\section{Simulation examples}
\label{sec:simulations}

\begin{discussion}
  Although the aim of this article is to develop intrinsic spatial-discrete formulations for transport phenomena, we will verify the proposed CMC formulations by comparing them to formulations with exact solutions in the continuum, as if the discrete formulations were discretisations.
  Indeed, we first introduced the necessary apparatus of exterior calculus on manifolds and used it to formulate strong, primal weak and mixed weak formulations.
  Then we introduced an analogous apparatus on meshes (CMC), made the connection with exterior calculus, and used it to formulate analogous discrete formulations.
  Since we are not doing a discretisation, we could directly derive strong and weak formulations.
  However, the reasons for our approach are the following:
  \begin{enumerate}
    \item
      exterior calculus is a well established mathematical theory, and so is instructive to relate to it; clearly, it has guided us in the development of the discrete theory;
    \item
      we could pose problems with manufactured solutions using exterior calculus, and pre-process (discretise) them for input to discrete problems.
      We can them compare the discrete solutions with the discretised versions of the continuum solutions.
    \item
      since Hodge stars $\star_p$ and $\star_{D - p}$ do not cancel (up to a sign $(-1)^{p (D - p)}$) in the discrete case, we would have to keep Hodge stars in the bilinear forms.
      Going through exterior calculus formulations and translating the resulting final forms allows us to motivate not using Hodge star.
      Non-invertibility and rank deficiency of discrete Hodge operators are
      common features of structure-preserving discretisations
      \cite{hiptmair2001discrete}, including FEEC, and should be regarded as
      intrinsic rather than pathological.
  \end{enumerate}
\end{discussion}
\begin{discussion}[Pre-processing of continuous input]
  \label{discussion:simulations/pre_processing}
  Consider a continuous steady-state problem
  (\Cref{formulation:continuum/primal_weak/steady_state} or
  \Cref{formulation:continuum/mixed_weak/steady_state})
  defined on a $D$-dimensional manifold $M$,
  with $\partial M = \Gamma_\Dirichlet \cup \Gamma_\Neumann$
  (\Cref{eq:continuum/boundary_decomposition}),
  and input data $(\kappa, f, g_\Dirichlet, g_\Neumann)$
  (time-independent versions from \Cref{tab:continuum/parameters}).
  Consider a mesh $\mathcal{M}$ that realises $M$ through an embedding
  $\varphi \colon \mathcal{M} \to \mathcal{P}(M)$,
  such that all of its boundary cells lie in full in
  $\Gamma_\Dirichlet$ or $\Gamma_\Neumann$,
  i.e., it is not possible for a cell to have non-empty intersection with both of
  $\Gamma_\Dirichlet$ and $\Gamma_\Neumann$,
  while not being on their common boundary.
  Take the Forman decomposition $\mathcal{K}$ of $\mathcal{M}$, and by
  $\Gamma_\Dirichlet^\mathcal{K}$ and $\Gamma_\Neumann^\mathcal{K}$
  denote the discrete versions of
  $\Gamma_\Dirichlet$ and $\Gamma_\Neumann$ respectively.
  Assume that $\kappa$ is constant everywhere, so its discrete counterpart $\kappa^\mathcal{K}$ has the same value on all $(D - 1)$-cells
  (or $\widetilde{\kappa}^\mathcal{K}$ on $1$-cells).
  $f$, $v$, $g_\Dirichlet,$ and $g_\Neumann$ are discretised using the corresponding de Rham maps:
  \begin{equation}
    f^\mathcal{K} = R_{\mathcal{K}, D}(f),\
    v^\mathcal{K} = R_{\mathcal{K}, D - 1}(v),\
    g_\Dirichlet^\mathcal{K} = R_{\Gamma_\Dirichlet^\mathcal{K}, 0}(g_\Dirichlet),\
    g_\Neumann^\mathcal{K} = R_{\Gamma_\Neumann^\mathcal{K}, D - 1}(g_\Neumann).
  \end{equation}
  Hence, we get the following discrete input for
  \Cref{formulation:discrete/primal_weak/steady_state} and
  \Cref{formulation:discrete/mixed_weak/steady_state}:
  \begin{equation}
    (\mathcal{K}, \Gamma_\Dirichlet^\mathcal{K}, \Gamma_\Neumann^\mathcal{K}, \kappa^\mathcal{K}, f^\mathcal{K}, g_\Dirichlet^\mathcal{K}, g_\Neumann^\mathcal{K}).
  \end{equation}
  In the examples, $v^\mathcal{K}$ is assumed to be $0$, i.e., we will not include advective terms.
\end{discussion}
\begin{discussion}[Translating densities to differential forms]
  In this article we have done all continuous modelling in the language of smooth exterior calculus which is better suited to be mimicked in a discrete formulation, and as a source of continuous input.
  We will not discuss in detail the full translation of smooth exterior calculus formulations into the familiar vector calculus ones.
  However, the input parameters, represented as differential forms, may be supplied as densities for which we have to do a continuous-to-continuous pre-processing before doing the continuous-to-discrete pre-processing.

  If internal production rate is supplied as a density, i.e., as $0$-form
  $\tilde{f} [\amount \time^{-1} \length^{- D}]$,
  we define $f := \star_0 \tilde{f}$.

  If the prescribed flow rate is given as a density (a $0$-form)
  $\widetilde{g_\Neumann} [\amount \time^{-1} \length^{1 - D}]$ on $\Gamma_\Neumann$,
  we define
  $g_\Neumann := \star_{\Gamma_\Neumann, 0} \widetilde{g_\Neumann}$.

  Pre-processing vector fields is a bit more involved, e.g., as in the case of the volumetric flow rate $v$ when advection is present.
  Assume instead of $v$ we have the flow velocity vector field ${\bf v}$, usually stated in Cartesian coordinates by its velocity components, i.e.,
  \begin{equation}
    {\bf v} = \sum_{p = 1}^D v_p \frac{\partial}{\partial x_p},\ v_p [\length \time^{-1}].
  \end{equation}
  Because
  $[[\frac{\partial}{\partial x_p}]] = \length^{-1}$
  (since $[[d x^p]] = \length$ and vector fields are dual to forms),
  $[[{\bf v}]] = \time^{-1}$.
  (This may sound counter-intuitive but note that the \emph{components} $\{v_p\}_{p = 1}^{D}$ of ${\bf v}$ have the dimension of speed, $\length \time^{-1}$.)
  Next, define the \newterm{flat operator} (one of the \newterm{musical isomorphisms}, the other one being its inverse, the \newterm{sharp operator}),
  $\flat \colon \mathcal{X} M \to \Omega^1 M$,
  as follows: for any $X, Y \in \mathcal{X} M$,
  \begin{equation}
    (X^\flat)(Y) = g(X, Y),
  \end{equation}
  whose dimension is the same as $g$, i.e, $\length^2$.
  Hence, we get the volumetric flow rate density $\tilde{v} [\length^2 \time^{-1}] := {\bf v}^\flat.$
  Finally, we apply the Hodge star to get the actual volumetric flow rate $v [\length^D \time^{-1}]$, i.e.,
  $v := \star_1 \tilde{v}$.
\end{discussion}
\begin{discussion}[Solving discrete systems]
  We can solve the resulting problem
  (primal: \Cref{eq:discrete/primal_weak/steady_state_formulation},
  or mixed: \Cref{eq:discrete/mixed_weak/steady_state_formulation})
  with standard methods to get
  the potential $u^\mathcal{K} \in C^0 \mathcal{K}$ and
  the flow rate $q^\mathcal{K} \in C^{D - 1} \mathcal{K}$.
  Indeed, the operators and boundary conditions involved have the same forms for all discrete or discretised weak formulations, and so the LHS matrices, RHS vectors, and boundary constraints have the same form
  (the details can be found in \cite[Equations 2.7--2.10, 2.21--2.24]{berbatov2021guide}).
  However, one important feature of the mixed weak CMC formulation presented here is that the matrix $A$ is diagonal, and so the straightforward solution method of variable elimination works without introducing dense matrices, thus transforming a mixed LHS matrix to a a sparse symmetric positive definite one.
  Indeed, the mixed system is represented in the form $A q - B^T u = - g,\ - B q = - f$, and we can express $q$ from the first equation by $q = A^{-1}(-g + B^T u)$, which leads to
  \begin{equation}
    - B A^{-1}(-g + B^T u) = - f \Rightarrow B A^{-1} B^T u = B A^{-1} g + f,
  \end{equation}
  and the matrix $B A^{-1} B^T$ is sparse (since $A$ is diagonal) and positive definite.
  Hence, both primal and mixed formulations can be solved using sparse Cholesky decomposition.

  While similar block-diagonal structures arise in hybridised mixed finite
  element methods (\cite[Section 7.2]{boffi2013mixed},
  \cite[Introduction]{arnold1985mixed}), in the present framework this structure
  follows directly from combinatorial inner products, without the introduction
  of additional hybridisation variables.
\end{discussion}
\begin{discussion}[Post-processing and comparison with exact continuous solutions]
  $u^\mathcal{K}$ is visualised trivially by applying it to the nodes of the mesh and creating a heat map.
  To visualise the $q^\mathcal{K}$ on a $(D - 1)$-cell $S$ we proceed as follows.
  Let $S$ be boundary of $V_-$ and $V_+$ (if $S$ is on the boundary, we assume a ghost cell as the background).
  The direction of the flow is from $V_+$ to $V_-$ if $q^\mathcal{K}(S) > 0$, from $V_-$ to $V_+$ if $q^\mathcal{K}(S) < 0$, and there is no flow if $q^\mathcal{K}(S) = 0$.
  The flow rate is then drawn as an arrow in the direction of the flow, with magnitude given by the heat map value of $q^\mathcal{K}(S)$ (no normalisation, e.g., dividing it by its measure $\mu_{D - 1}(S)$, is done).
  No arrow is drawn if the flow is zero.
  We use a rainbow colour scheme, shown on \Cref{figure:colorbar/horizontal} (representing lowest value in red and highest value in magenta).

  The continuous solution $(u, q)$ is discretised into $(u', q')$ by the de Rham maps $R_0$ and $R_{D - 1}$, and then visualised the same way as $(u^\mathcal{K}, q^\mathcal{K})$.
  For both formulations and both variables we compute global relative errors with respect to the Euclidean norm
  $\norm{(x_1, ..., x_n)}_2 := (x_1^2 + ... + x_n^2)^{1 / 2}$, i.e.,
  \begin{equation}
    u_{\rm rel} := \frac{\norm{u^\mathcal{K} - u'}_2}{\norm{u'}_2},\
    q_{\rm rel} := \frac{\norm{q^\mathcal{K} - q'}_2}{\norm{q'}_2}.
  \end{equation}
\end{discussion}
\begin{figure}
  \centering
  \includegraphics[width=\linewidth, keepaspectratio]{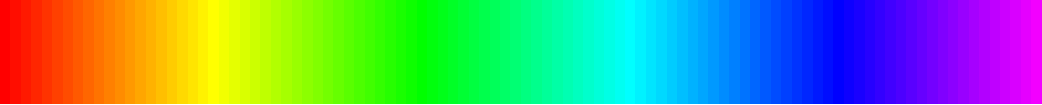}
  \caption{Rainbow colour scheme used for visualisation of potentials and flow rates}
  \label{figure:colorbar/horizontal}
\end{figure}
\begin{discussion}
  We will consider four examples for the steady-state formulation, that will illustrate the CMC formulations for various problems with prescribed exact solutions (we start with desired potential and flow rate and derive the input data that leads to these solutions).
  In all examples the conductivity will be constant across their domains, and there will be no advective terms.
  \begin{enumerate}
    \item
      In \Cref{sec:simulations/cube} we will consider a $3$D example for a simple problem with exact potential on the unit cube.
      Application of the CMC formulations to more complex $3$D domains is left for future modelling of real microstructures.
    \item
      In \Cref{sec:simulations/disk} we will simulate a problem on a disk with a quadratic potential.
      We will be using regular meshes (concentric circles intersected by equally spaced rays), but will illustrate the natural handling of curved cells.
      While classical finite element and isogeometric methods can represent
      curved geometries through isoparametric or spline-based constructions,
      the present approach treats curved cells directly at the level of the cell
      complex, without reference to polynomial reconstruction spaces.
    \item
      The example in \Cref{sec:simulations/hemisphere} is a further complication of the previous one -- this time the whole space is curved, being a hemisphere.
      Again, we will be working with the sphere as as a curved manifold.
      We will calculate the metric tensor and impose input parameters in spherical coordinates, and discretise the input parameters for the discrete formulations.
    \item
      In \Cref{sec:simulations/neper} we will simulate a problem with linear potential (and, hence, a constant flow rate) on a rectangle.
      We will, however test the CMC formulations on an irregular mesh generated by Neper \cite{quey2011large}, \url{https://neper.info}.
  \end{enumerate}
  The code for these examples (and many more) can be found at \url{https://github.com/kipiberbatov/cmc}.
\end{discussion}

\subsection{Quadratic potential on the unit cube}
\label{sec:simulations/cube}
\phantom{T}

\begin{example}
  Let $M = [0, 1]^3$ be the unit cube with the Euclidean metric.
  Consider an exact potential
  \begin{equation}
    u \colon M \to \R,\ u(x, y, z) = x^2 + y^2 + z^2.
  \end{equation}
  Take $\kappa \equiv 2$.
  Then the flow rate is given by
  \begin{equation}
    q = -2 \star_1 d_0 u
    = -2 \star_1(2 x\, d x + 2 y\, d y + 2 z\, d z)
    = -4(x\, d y \wedge d z + y\, d z \wedge d x + z\, d x \wedge d y).
  \end{equation}
  The internal production rate is
  \begin{equation}
    f = d q = -12\, d x \wedge dy \wedge d z.
  \end{equation}
  Let the Dirichlet boundary $\Gamma_\Dirichlet$ consists of the front, back, top, and bottom faces, while the Neumann boundary $\Gamma_\Neumann$ consists of the left and right faces, i.e.,
  \begin{equation}
    \Gamma_\Dirichlet = [0, 1] \times \{0, 1\} \times [0, 1] \cup [0, 1] \times [0, 1] \times \{0, 1\},\ \Gamma_\Neumann = \{0, 1\} \times [0, 1] \times [0, 1].
  \end{equation}
  The prescribed potential and flow rate are equal to
  \begin{equation}
    g_\Dirichlet(x, y, z)
    = (\tr_{\Gamma_\Dirichlet, 0}(u))(x, y, z)
    =
    \begin{cases}
      x^2 + z^2, & y = 0 \\
      x^2 + z^2 + 1, & y = 1 \\
      x^2 + y^2, & z = 0 \\
      x^2 + y^2 + 1, & z = 1
    \end{cases},\
    g_\Neumann = \tr_{\Gamma_\Neumann, 1}(q) =
    \begin{cases}
      0, & x = 0 \\
      -4 d y \wedge d z, & x = 1
    \end{cases}.
  \end{equation}
  The relative errors are:
  0
   for primal weak potential;
  0.0467428
   for mixed weak potential;
  0.129099
   for primal weak flow rate;
  7.2207e-16
   for mixed weak flow rate.

  An interesting feature of the solution to the discrete primal weak formulation on a regular grid is that it is exact for all potentials of degree at most two, achieving this with the lowest amount of degrees of freedom, that is only nodal ones.
  (The reported nonzero error for primal weak flow rate is due to imperfect post-processing on the boundary.)
  Indeed, the same cannot be said for finite differences (because of Neumann boundary conditions have precision only $O(h)$) and lowest order (Lagrange) finite elements (they are exact for products of linear polynomials of the space coordinates).
  Similarly, mixed flow rate perfectly approximates the exact flow rate.
\end{example}

\subsection{Quadratic potential on a 2D disk}
\label{sec:simulations/disk}

\begin{example}
  Let $M = \set{(x, y) \in \R^2}{x^2 + y^2 \leq 1}$ be the unit disk with the Euclidean metric.
  Consider an exact potential
  \begin{equation}
    u \colon M \to \R,\ u(x, y) = x^2 + y^2.
  \end{equation}
  Take $\kappa \equiv 1$.
  Then the flow rate is given by
  \begin{equation}
    \label{eq:simulations/disk/flow_rate}
    q
    = - \star_1 d_0 u
    = - \star_1(2 x\, d x + 2 y\, d y)
    = 2 y\, d x - 2 x\, d y.
  \end{equation}
  The internal production rate is
  \begin{equation}
    f
    = d q
    = d(2 y\, d x - 2 x\, d y)
    = 2\, d y \wedge d x - 2\, d x \wedge d y
    = - 4\, d x \wedge d y.
  \end{equation}
  The boundary of $M$ is the unit circle $\partial M = \set{(x, y) \in \R^2}{x^2 + y^2 = 1}$.
  Let the Dirichlet boundary $\Gamma_\Dirichlet$ be the right half-circle, while and the Neumann boundary $\Gamma_\Neumann$ be the left half-circle, i.e.,
  \begin{equation}
    \Gamma_\Dirichlet = \set{(x, y) \in \partial M}{x \geq 0},\ \Gamma_\Neumann = \set{(x, y) \in \partial M}{x \leq 0}.
  \end{equation}
  The prescribed potential equals to
  \begin{equation}
    g_\Dirichlet(x, y) = (\tr_{\Gamma_\Dirichlet, 0}(u))(x, y) = 1.
  \end{equation}
  We will describe the prescribed flow rate $g_\Neumann$ on $\Gamma_\Neumann$ in polar coordinates, i.e.,
  $(x, y) = (\cos \varphi, \sin \varphi),\ \varphi \in [\pi / 2, 3 \pi / 2]$.
  Let $\widetilde{g_\Neumann}$ be the prescribed flow rate in polar coordinates.
  Substituting $x = \cos \varphi$ and $y = \sin \varphi$ in \Cref{eq:simulations/disk/flow_rate} gives
  \begin{equation}
    \widetilde{g_\Neumann}(\varphi)
    = 2 \sin \varphi\, d(\cos \varphi) - 2 \cos \varphi\, d(\sin \varphi)
    = - 2 \sin^2 \varphi\, d \varphi - 2 \cos^2 \varphi\, d \varphi
    = -2\, d \varphi.
  \end{equation}
  A regular polar mesh $\mathcal{M}$ for $M$ is constructed by dividing domains $[0, 1]$ (of radius) and $[0, 2 \pi]$ (of rotation angle) into $4$ and $3$ equal parts respectively.
  The embedding of the Forman subdivision $\mathcal{K}$ is chosen so that the regularity of the mesh is preserved.
  $\mathcal{M}$ and $\mathcal{K}$ are shown on \Cref{figure:mesh/circular_4_3}.
  Exact solutions for the continuous problem on $M$ (discretised on $\mathcal{K})$ and solutions for the primal and mixed weak discrete formulations on $\mathcal{K}$ are visualised on \Cref{figure:diffusion/disk_quadratic}.
  The relative errors are:
  0.0243588
   for primal weak potential;
  0.0802977
   for mixed weak potential;
  0.0581986
   for primal weak flow rate;
  4.72913e-06
   for mixed weak flow rate.
\end{example}

\begin{figure}[!ht]
  \begin{subfigure}{.45\textwidth}
    \centering
    \includegraphics[width = \linewidth, keepaspectratio]{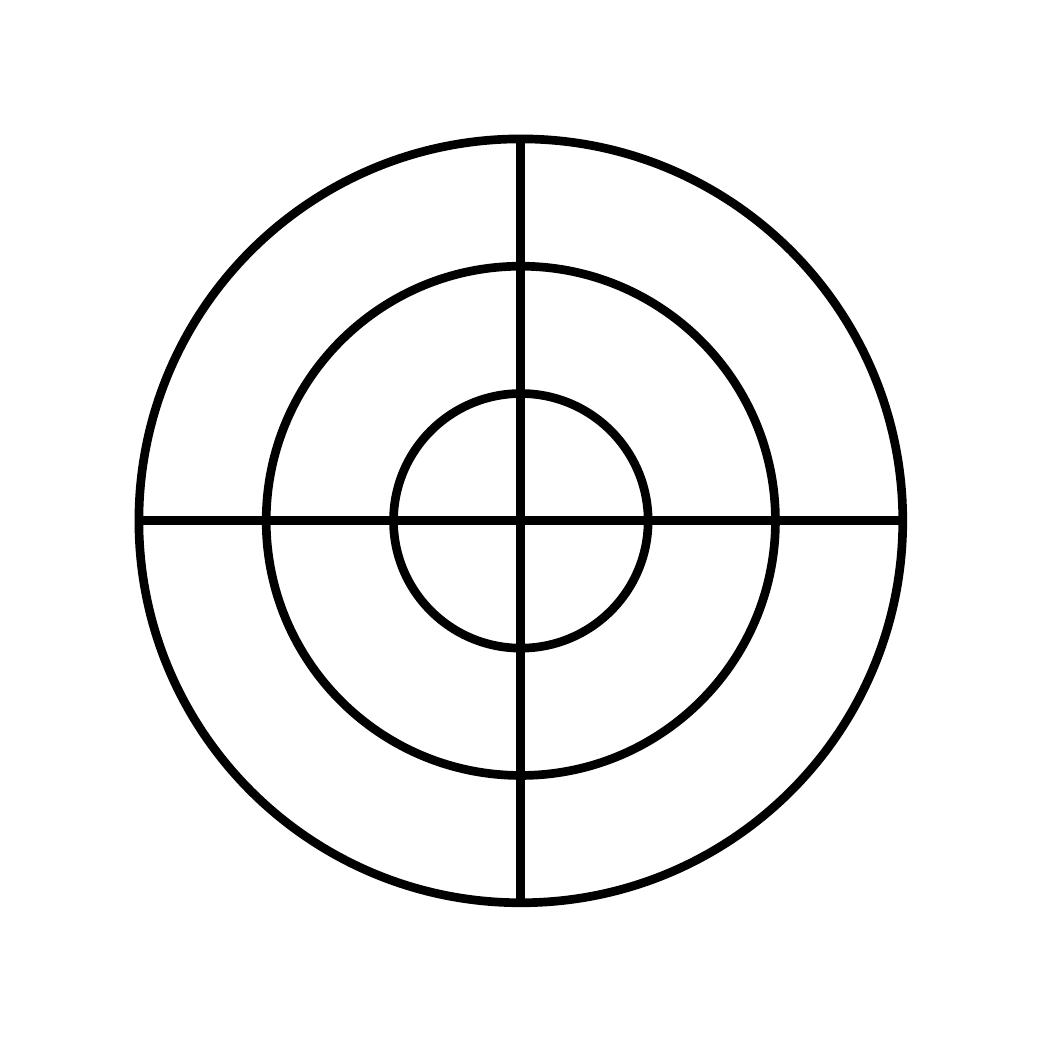}
    \caption{Regular polar mesh $\mathcal{M}$ for a circle}
  \end{subfigure}
  \begin{subfigure}{.45\textwidth}
    \centering
    \includegraphics[width = \linewidth, keepaspectratio]{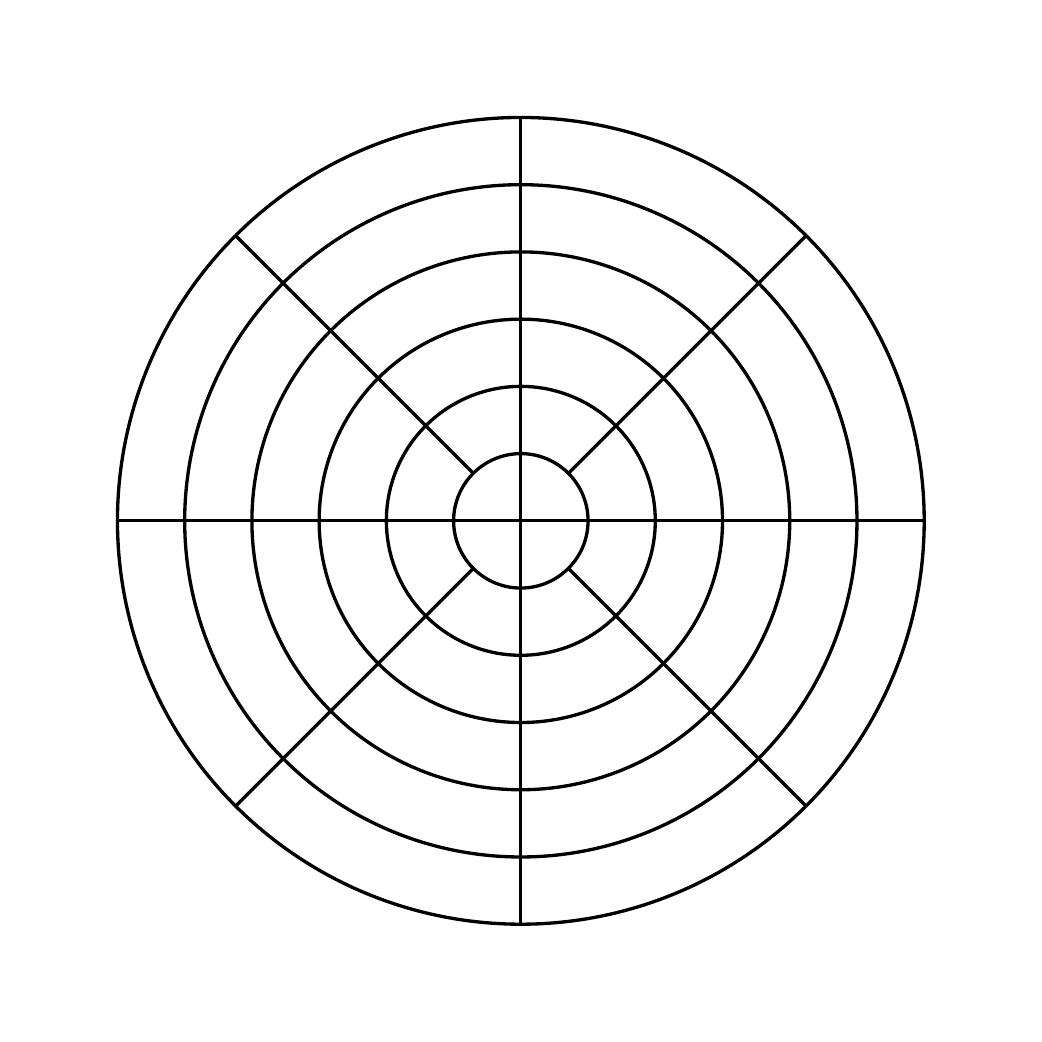}
    \caption{Forman subdivision $\mathcal{K}$ of $\mathcal{M}$}
  \end{subfigure}
  \caption{Regular curvilinear polar mesh on a disk and its Forman subdivision}
  \label{figure:mesh/circular_4_3}
\end{figure}

\begin{figure}[!ht]
  \begin{subfigure}{.3\textwidth}
    \centering
    \includegraphics[width = \linewidth, keepaspectratio]{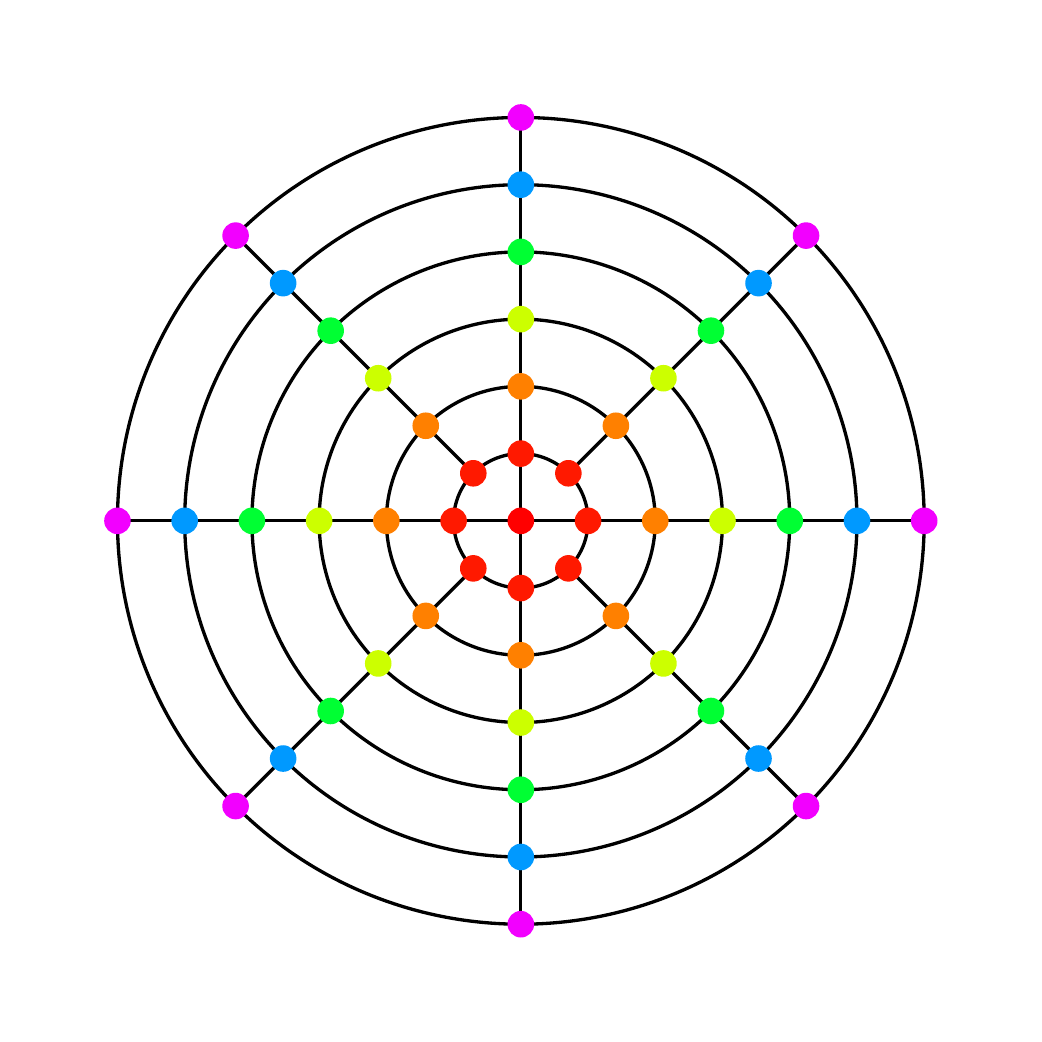}
    \caption{Exact continuous potential}
  \end{subfigure}
  \begin{subfigure}{.3\textwidth}
    \centering
    \includegraphics[width = \linewidth, keepaspectratio]{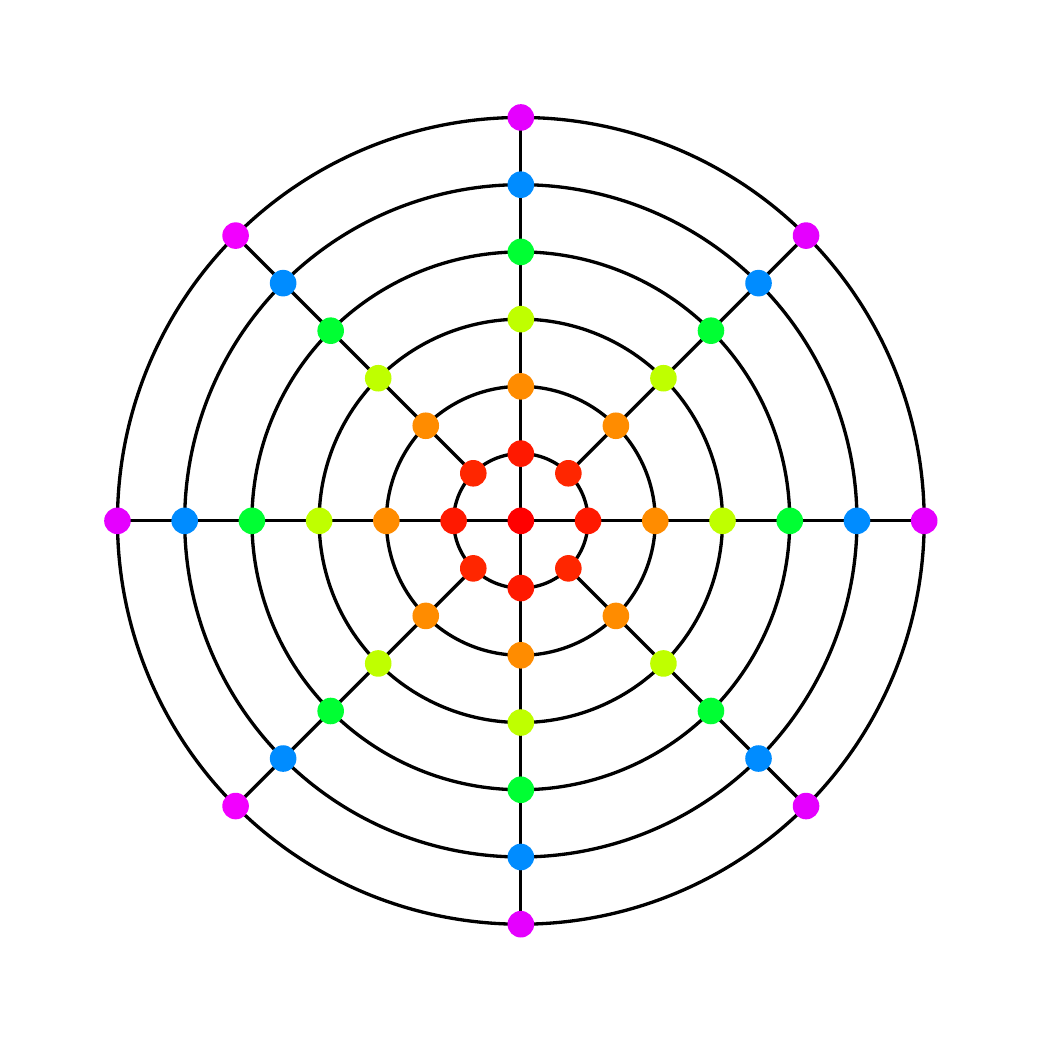}
    \caption{Discrete primal weak potential}
  \end{subfigure}
  \begin{subfigure}{.3\textwidth}
    \centering
    \includegraphics[width = \linewidth, keepaspectratio]{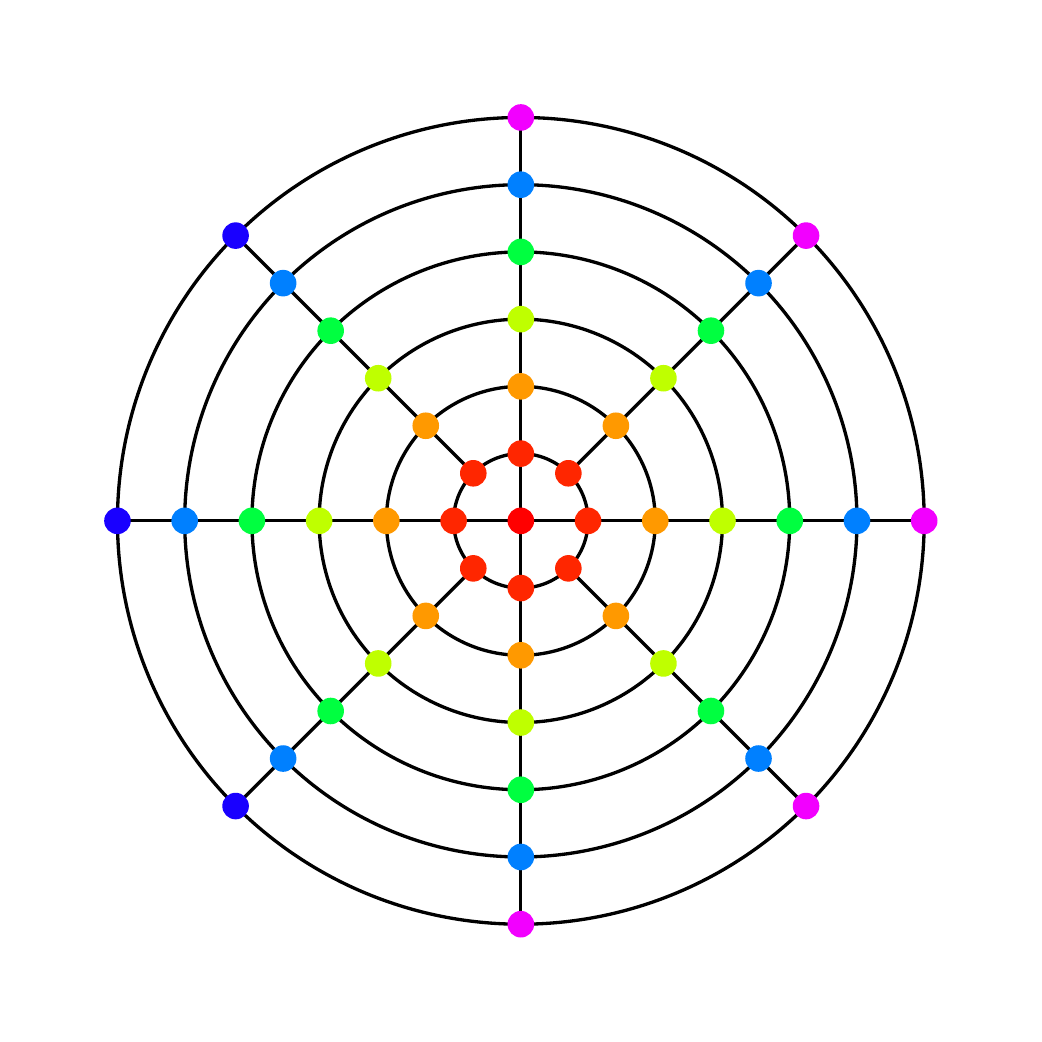}
    \caption{Discrete mixed weak potential}
  \end{subfigure}

  \begin{subfigure}{.3\textwidth}
    \centering
    \includegraphics[width = \linewidth, keepaspectratio]{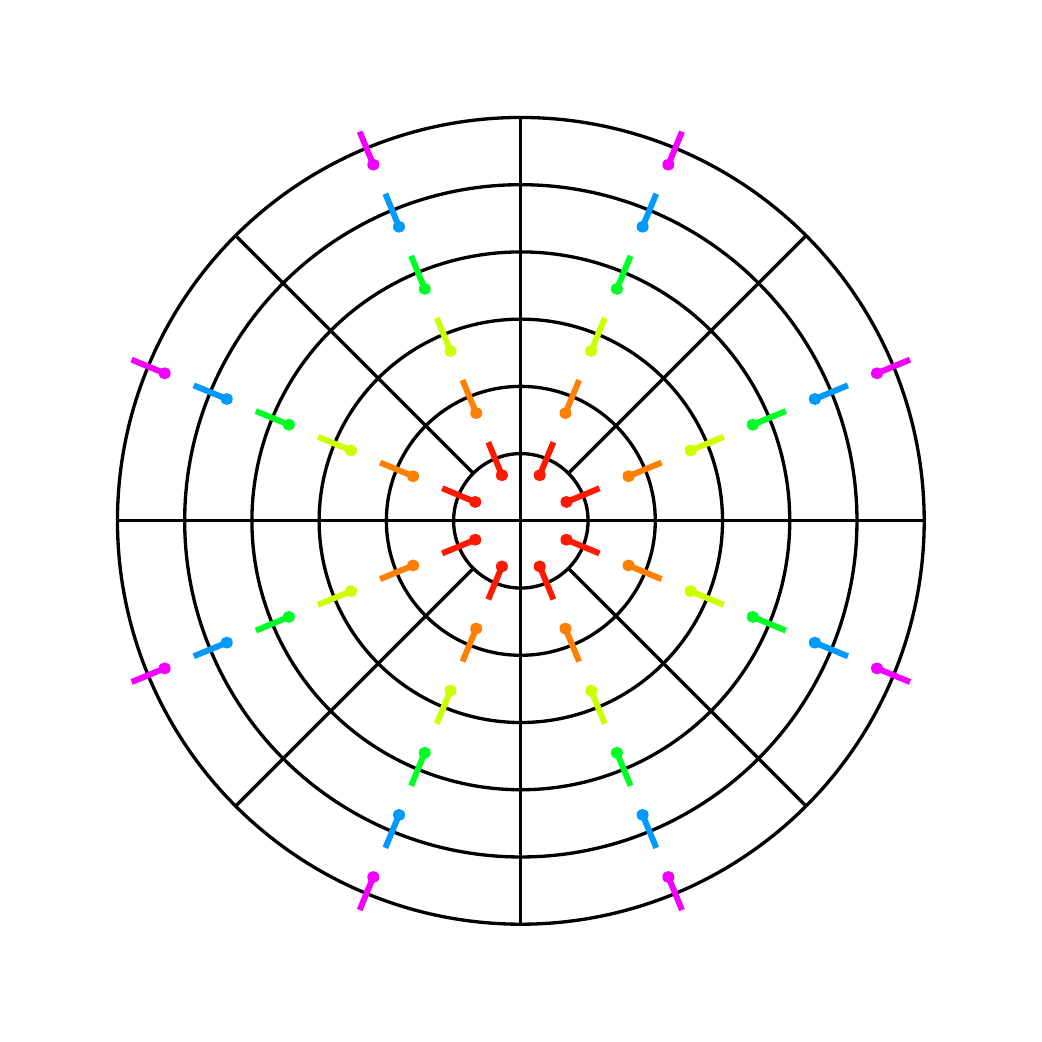}
    \caption{Exact continuous flow rate}
  \end{subfigure}
  \begin{subfigure}{.3\textwidth}
    \centering
    \includegraphics[width = \linewidth, keepaspectratio]{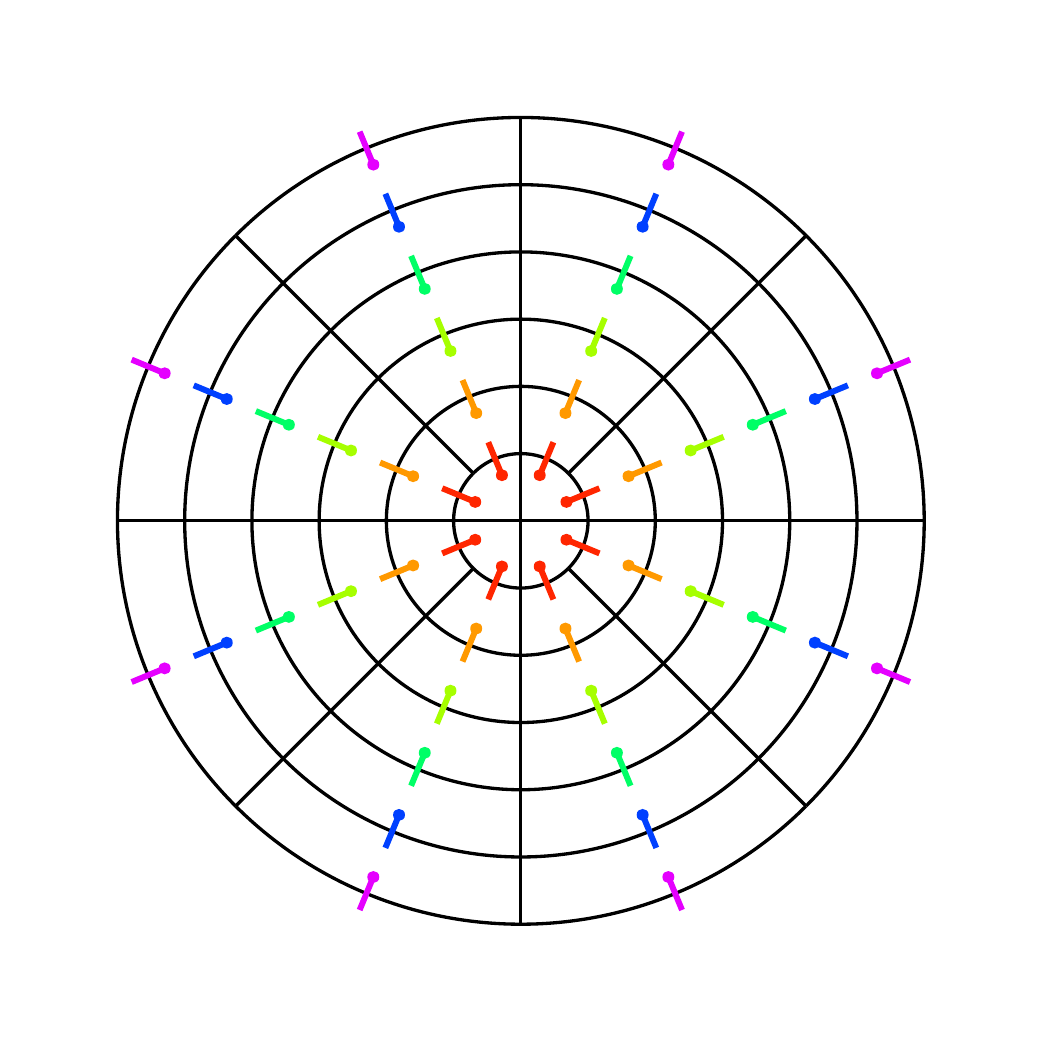}
    \caption{Discrete primal weak flow rate}
  \end{subfigure}
  \begin{subfigure}{.3\textwidth}
    \centering
    \includegraphics[width = \linewidth, keepaspectratio]{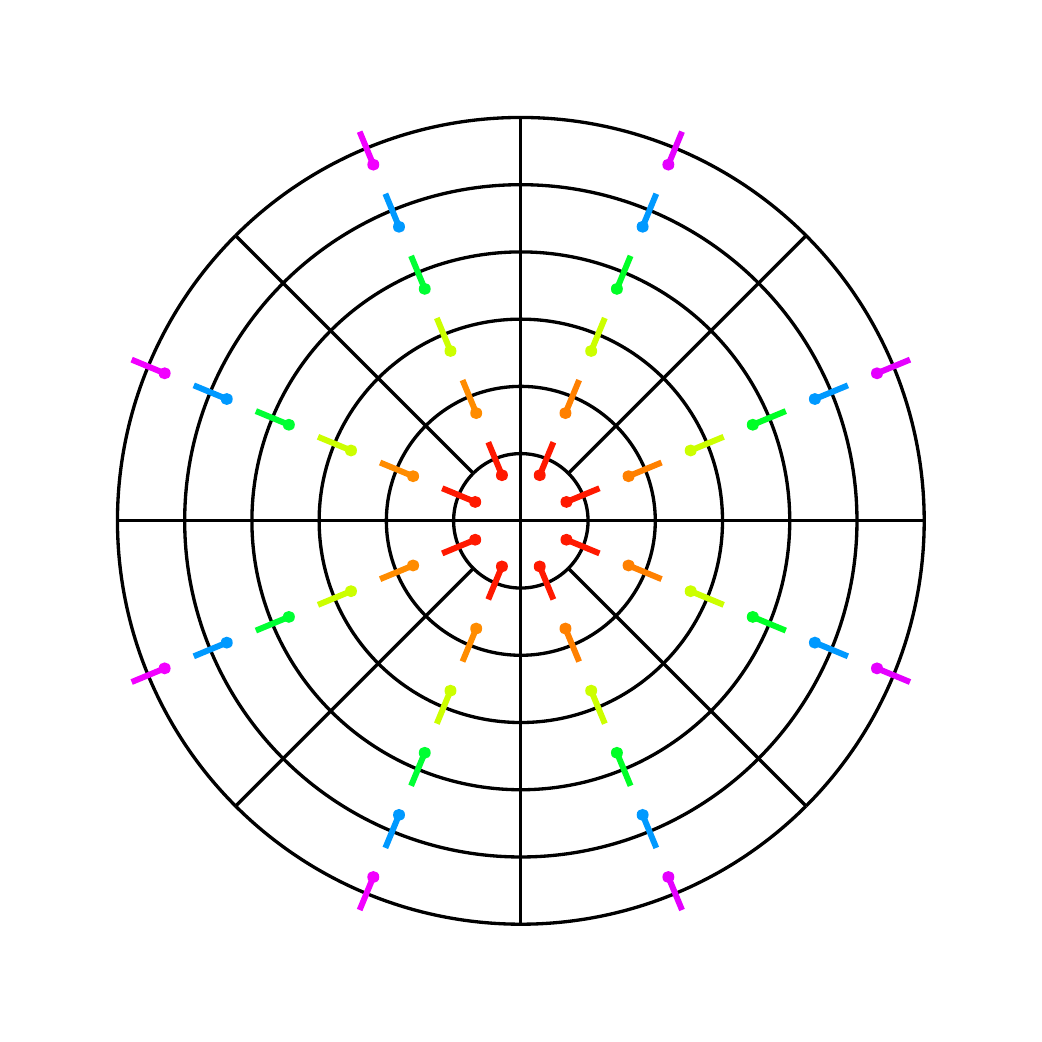}
    \caption{Discrete mixed weak flow rate}
  \end{subfigure}

  \caption{Solutions for diffusion with quadratic potential on a disk with a regular polar mesh $\mathcal{K}$}
  \label{figure:diffusion/disk_quadratic}
\end{figure}

\subsection{Linear (in spherical coordinates) potential on a hemisphere}
\label{sec:simulations/hemisphere}

\begin{discussion}
  Let $M = \set{(x, y, z)^3 \in \R^3}{x^2 + y^2 + z^2 = 1,\ z \geq 0}$ be the northern hemisphere, $g$ be the pullback metric of the Euclidean metric with respect to the inclusion map $\iota \colon M \to \R^3$.
  Introduce the spherical coordinates
  \begin{equation}
    (x, y, z) = (\sin \theta \cos \varphi, \sin \theta \sin \varphi, \cos \theta),\
    0 \leq \theta \leq \pi / 2,\
    0 \leq \varphi \leq 2 \pi.
  \end{equation}
  (Technically, there are singularities for $\theta = 0$ and $\varphi = 2 \pi$, but these will not cause problems.)
  Then,
  \begin{equation}
    \begin{split}
      g
      & = \iota^*(d x \otimes d x + d y \otimes d y + d z \otimes d z) \\
      & = (d(\sin \theta \cos \varphi))^2
        + (d(\sin \theta \sin \varphi))^2
        + (d(\cos \theta))^2 \\
      & = (\cos \theta \cos \varphi\, d \theta
           - \sin \theta \sin \varphi\, d \varphi)^2
          + (\cos \theta \sin \varphi\, d \theta
             + \sin \theta \cos \varphi\, d \varphi)^2
          + (- \sin \theta\, d \theta)^2 \\
      & = (d \theta)^2 + (\sin \theta)^2\, (d \varphi)^2.
    \end{split}
  \end{equation}
  We need the dual metric on forms which is given by
  \begin{equation}
    g^*_1 = \frac{\partial}{\partial \theta} \otimes \frac{\partial}{\partial \theta}
      + \frac{1}{(\sin \theta)^2} \frac{\partial}{\partial \varphi} \otimes \frac{\partial}{\partial \varphi}.
  \end{equation}
  A positively oriented orthonormal basis for this metric is
  \begin{equation}
    (e_1, e_2) = (d \theta,\ \sin \theta\, d \varphi),
  \end{equation}
  and hence $\star_1 e_1 = e_2,\ \star_1 e_2 = - e_1$.
  The volume form is given by
  \begin{equation}
    \vol = e_1 \wedge e_2 = \sin \theta\, d \theta \wedge d \varphi.
  \end{equation}
\end{discussion}
\begin{example}
  Consider an exact potential $u$, given in spherical coordinates as
  \begin{equation}
    u(\theta, \varphi) = \theta.
  \end{equation}
  Take $\kappa \equiv 2$.
  Then the flow rate is given by
  \begin{equation}
    \label{eq:simulations/hemisphere/flow_rate}
    q
    = - \star_1 \kappa d_0 u
    = - 2 \star_1(d \theta)
    = - 2 \sin \theta\, d \varphi.
  \end{equation}
  The internal production rate is
  \begin{equation}
    f
    = d q
    = - 2 \cos \theta\, d \theta \wedge d \varphi.
  \end{equation}
  The boundary of $M$ is the equator $\theta = \pi / 2$, i.e.,
  \begin{equation}
    \partial M = \set{(x, y, 0) \in \R^3}{x^2 + y^2 = 1} = \set{(\cos \varphi, \sin \varphi, 0)}{0 \leq \varphi \leq 2 \pi}.
  \end{equation}
  Let the Dirichlet boundary $\Gamma_\Dirichlet$ be has $y \leq 0$, while and the Neumann boundary $\Gamma_\Neumann$ has $y \geq 0$, i.e.,
  \begin{equation}
    \Gamma_\Dirichlet
    = \set{(\cos \varphi, \sin \varphi, 0)}
    {\pi \leq \varphi \leq 2 \pi},\
    \Gamma_\Neumann = \set{(\cos \varphi, \sin \varphi, 0)}{0 \leq \varphi \leq \pi}.
  \end{equation}
  The prescribed potential equals to
  \begin{equation}
     g_\Dirichlet(\varphi) = u(\pi / 2, \varphi) = 1.
  \end{equation}
  The prescribed flow rate is
  \begin{equation}
    g_\Neumann = \restrict{q}{\theta = \pi / 2} = -2 d \varphi.
  \end{equation}
  Projections on the $xy$-plane of a regular curvilinear mesh $\mathcal{M}$ for $M$ and its Forman subdivision $\mathcal{K}$ are shown on \Cref{figure:mesh/hemisphere_6_6}.
  Exact solutions for the continuous problem on $M$ (discretised on $\mathcal{K})$ and solutions for the primal and mixed weak discrete formulations on $\mathcal{K}$ are visualised on \Cref{figure:diffusion/hemisphere_spherical_linear}.
  The relative errors are:
  0.0190061
   for primal weak potential;
  0.0256953
   for mixed weak potential;
  0.0161111
   for primal weak flow rate;
  0.000889324
   for mixed weak flow rate.
\end{example}

\begin{figure}[!ht]
  \begin{subfigure}{.45\textwidth}
    \centering
    \includegraphics[width = \linewidth, keepaspectratio]{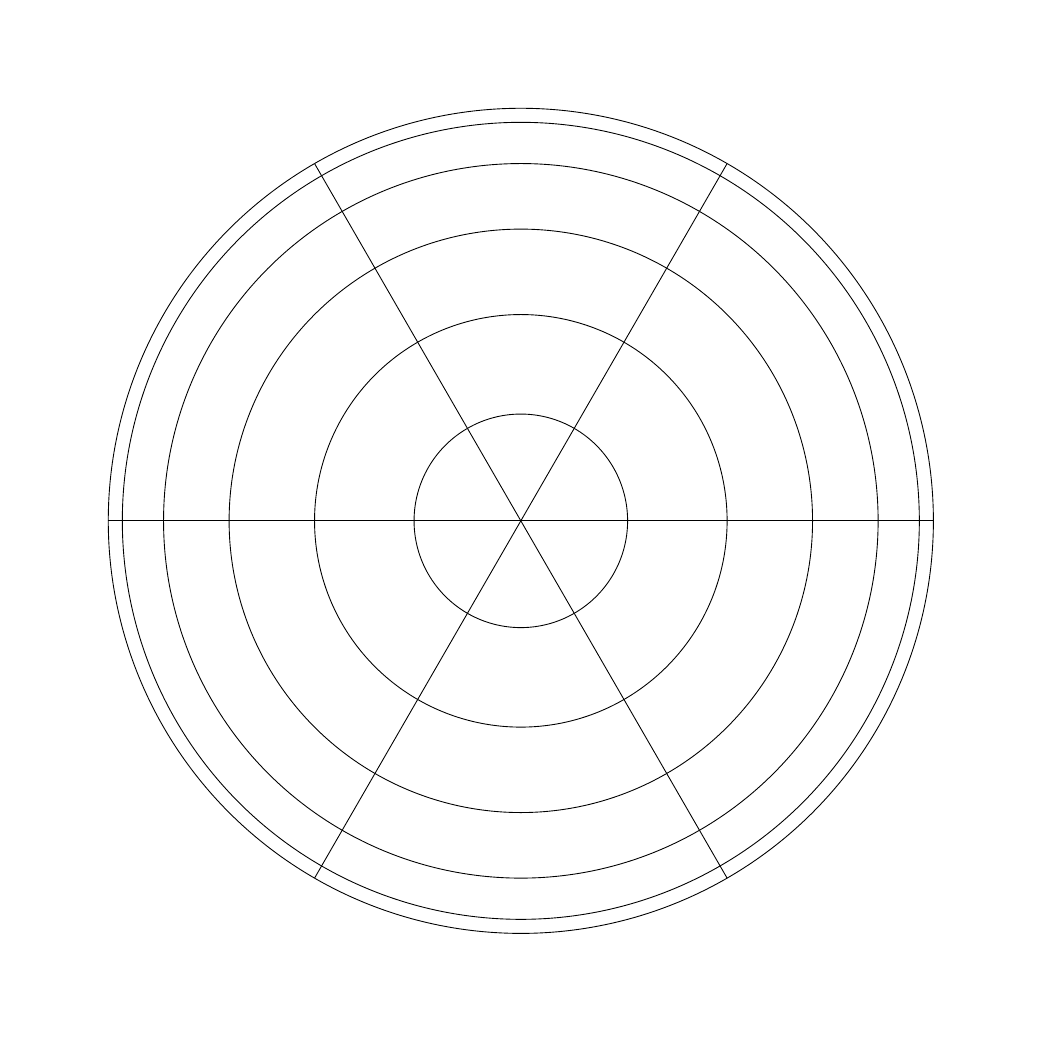}
    \caption{Projection of a regular mesh $\mathcal{M}$ for a hemisphere}
  \end{subfigure}
  \begin{subfigure}{.45\textwidth}
    \centering
    \includegraphics[width = \linewidth, keepaspectratio]{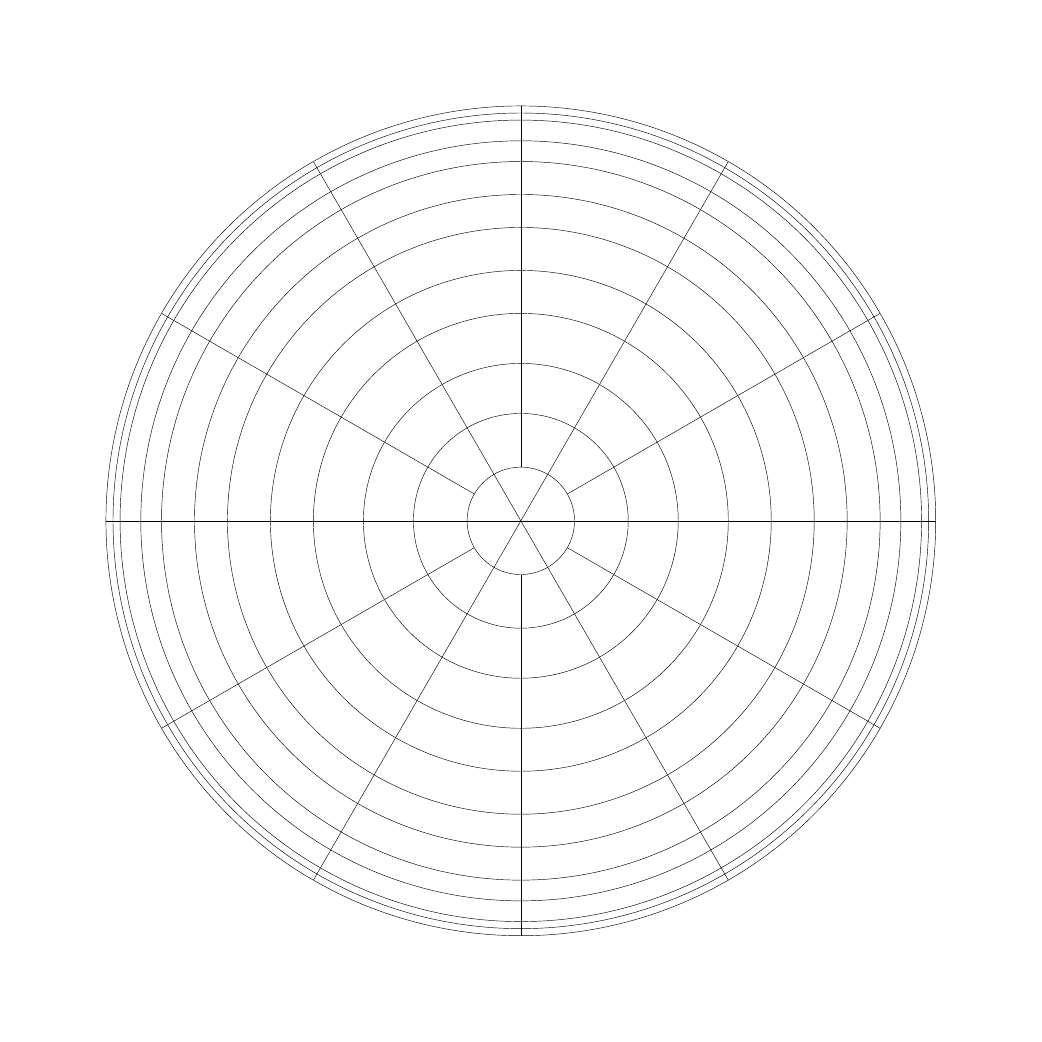}
    \caption{Projection of the Forman subdivision $\mathcal{K}$ of $\mathcal{M}$}
  \end{subfigure}
  \caption{Projections of a regular mesh on a hemisphere and its Forman subdivision}
  \label{figure:mesh/hemisphere_6_6}
\end{figure}

\begin{figure}[!ht]
  \begin{subfigure}{.3\textwidth}
    \centering
    \includegraphics[width = \linewidth, keepaspectratio]{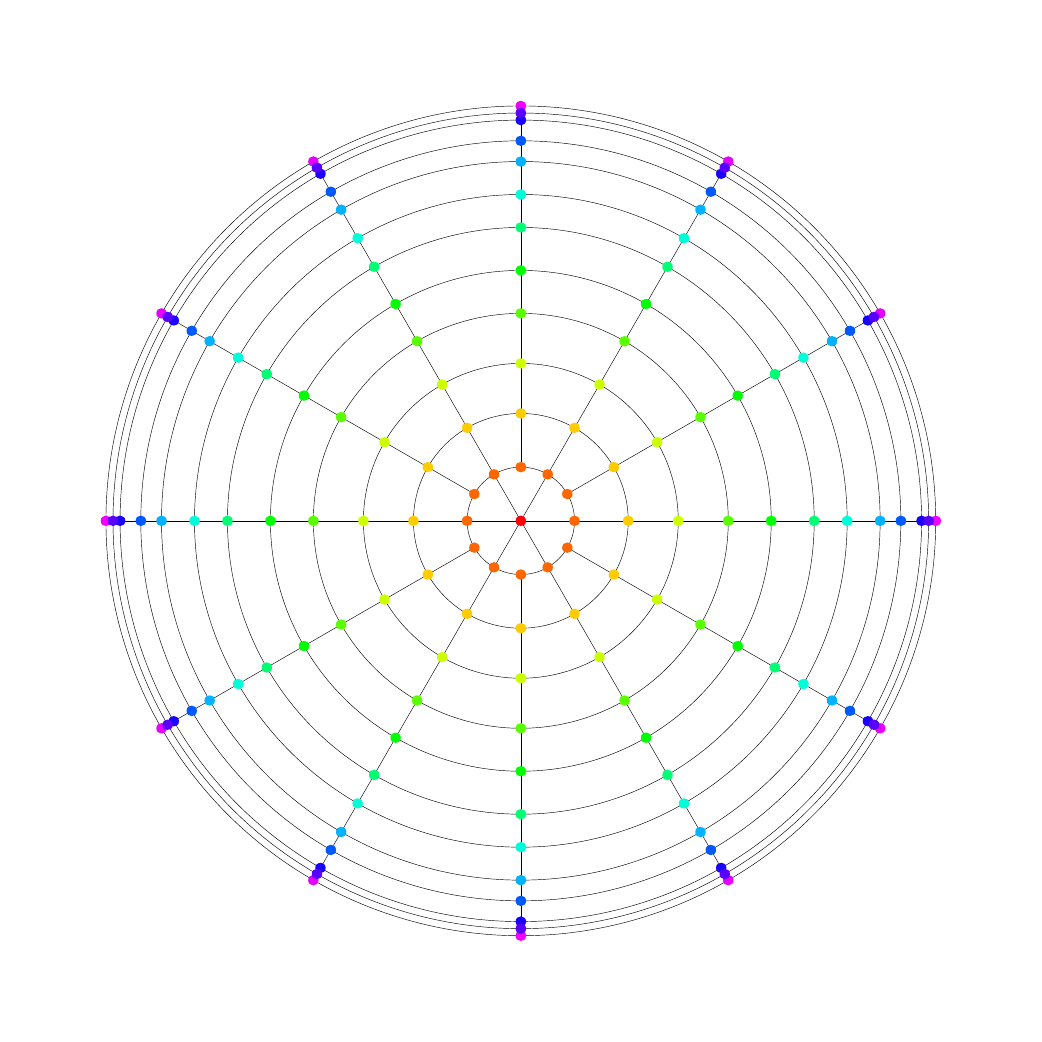}
    \caption{Exact continuous potential}
  \end{subfigure}
  \begin{subfigure}{.3\textwidth}
    \centering
    \includegraphics[width = \linewidth, keepaspectratio]{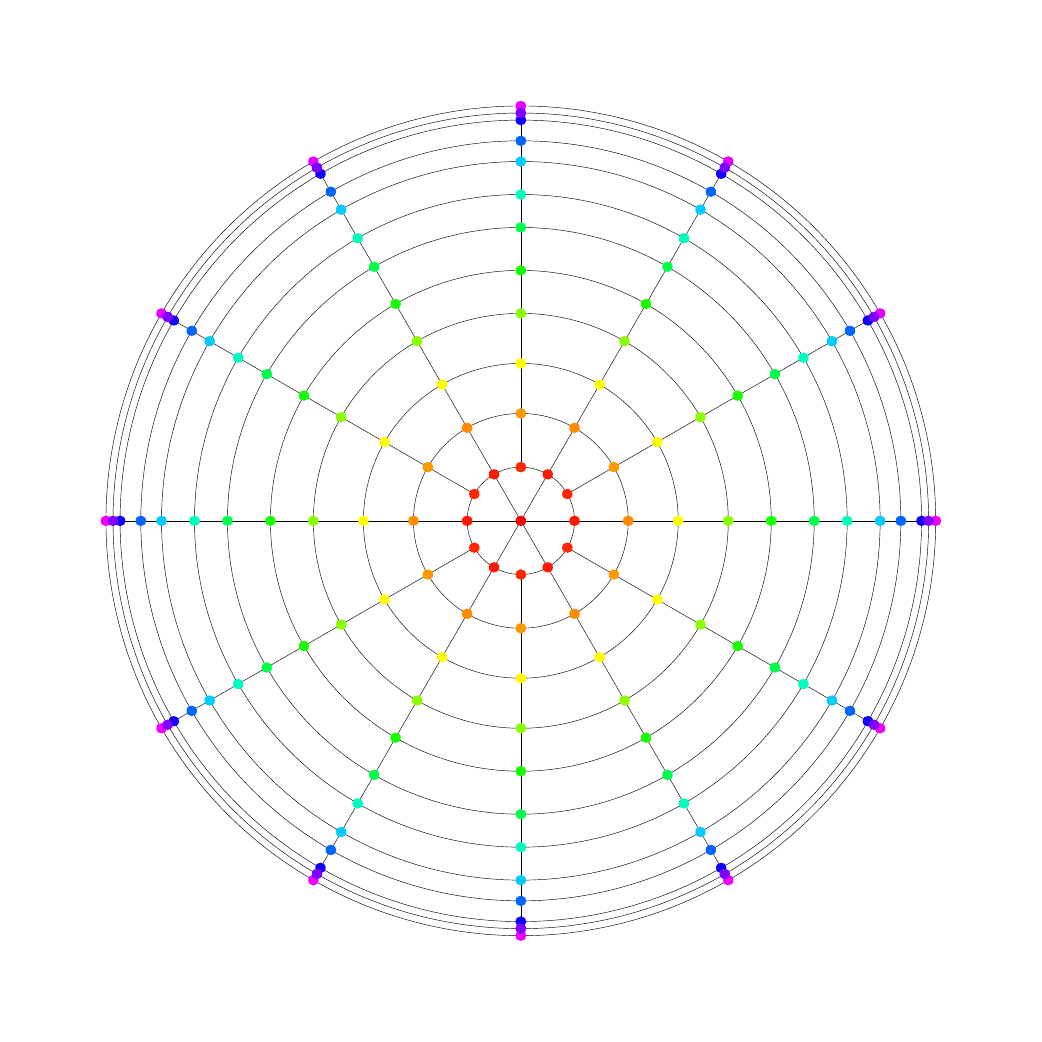}
    \caption{Discrete primal weak potential}
  \end{subfigure}
  \begin{subfigure}{.3\textwidth}
    \centering
    \includegraphics[width = \linewidth, keepaspectratio]{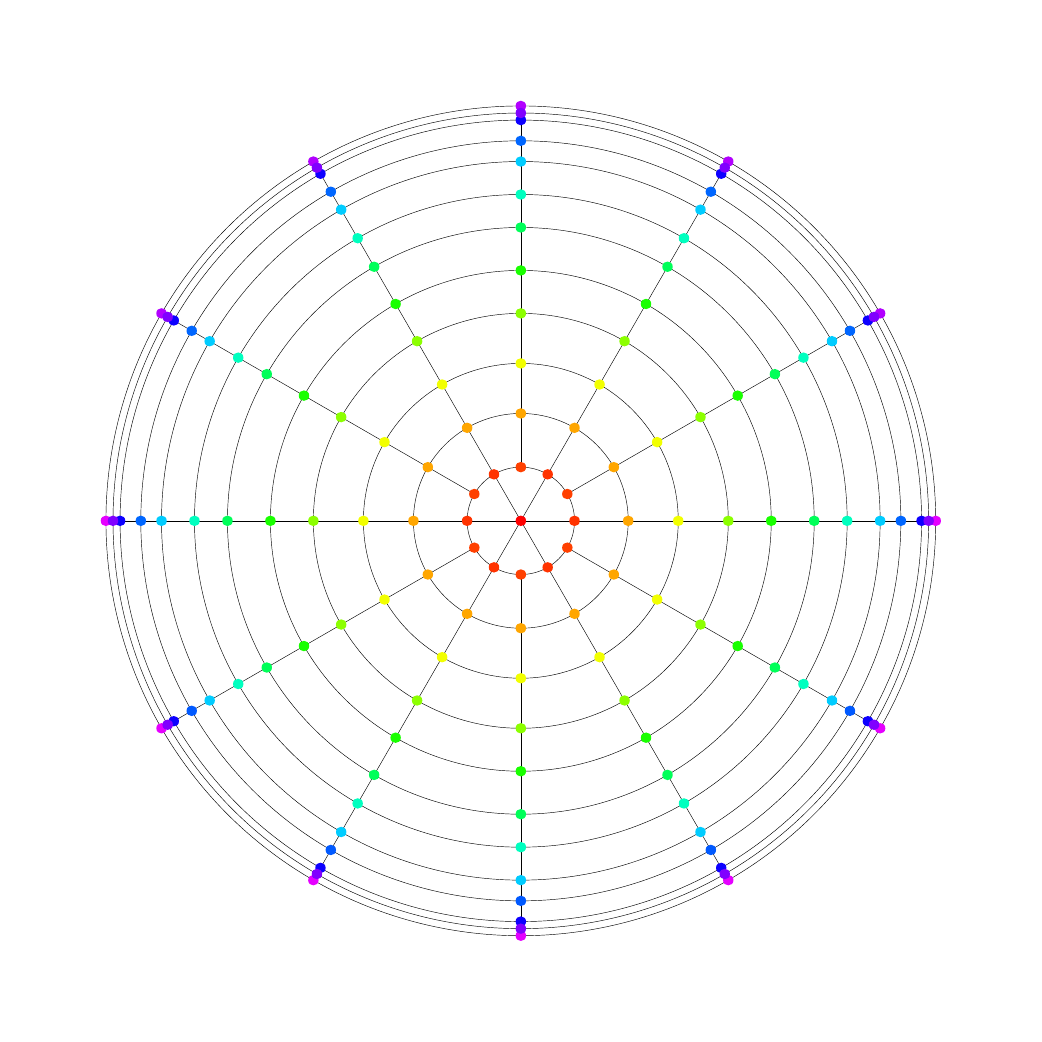}
    \caption{Discrete mixed weak potential}
  \end{subfigure}

  \begin{subfigure}{.3\textwidth}
    \centering
    \includegraphics[width = \linewidth, keepaspectratio]{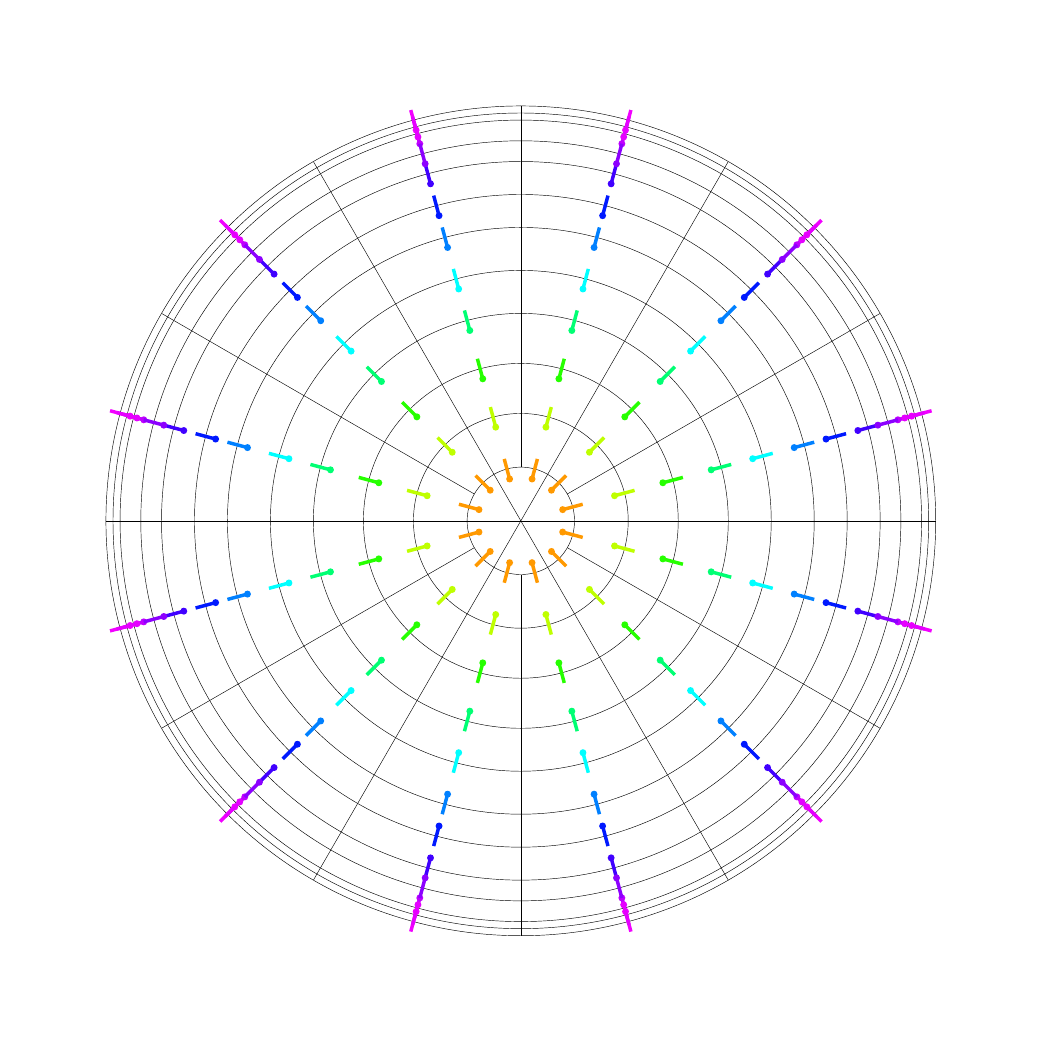}
    \caption{Exact continuous flow rate}
  \end{subfigure}
  \begin{subfigure}{.3\textwidth}
    \centering
    \includegraphics[width = \linewidth, keepaspectratio]{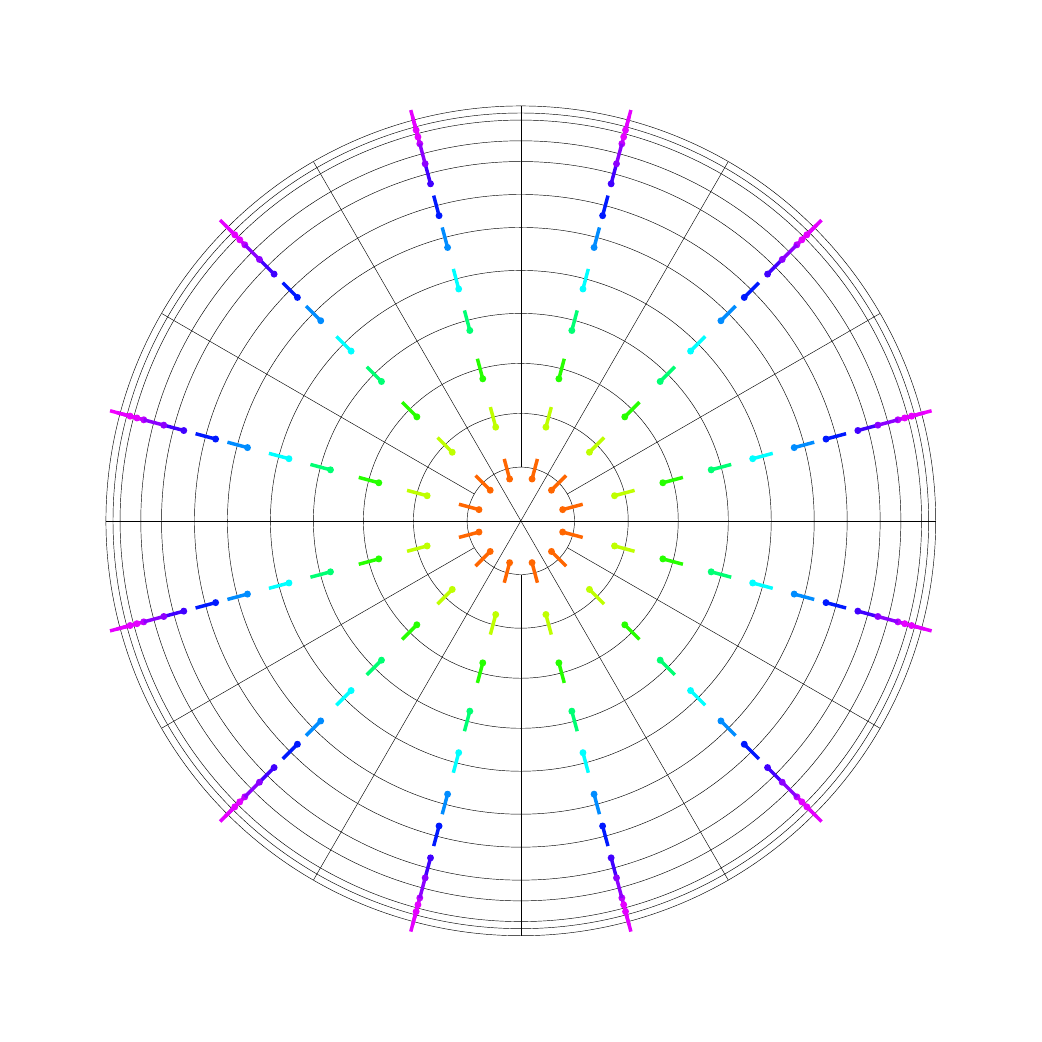}
    \caption{Discrete primal weak flow rate}
  \end{subfigure}
  \begin{subfigure}{.3\textwidth}
    \centering
    \includegraphics[width = \linewidth, keepaspectratio]{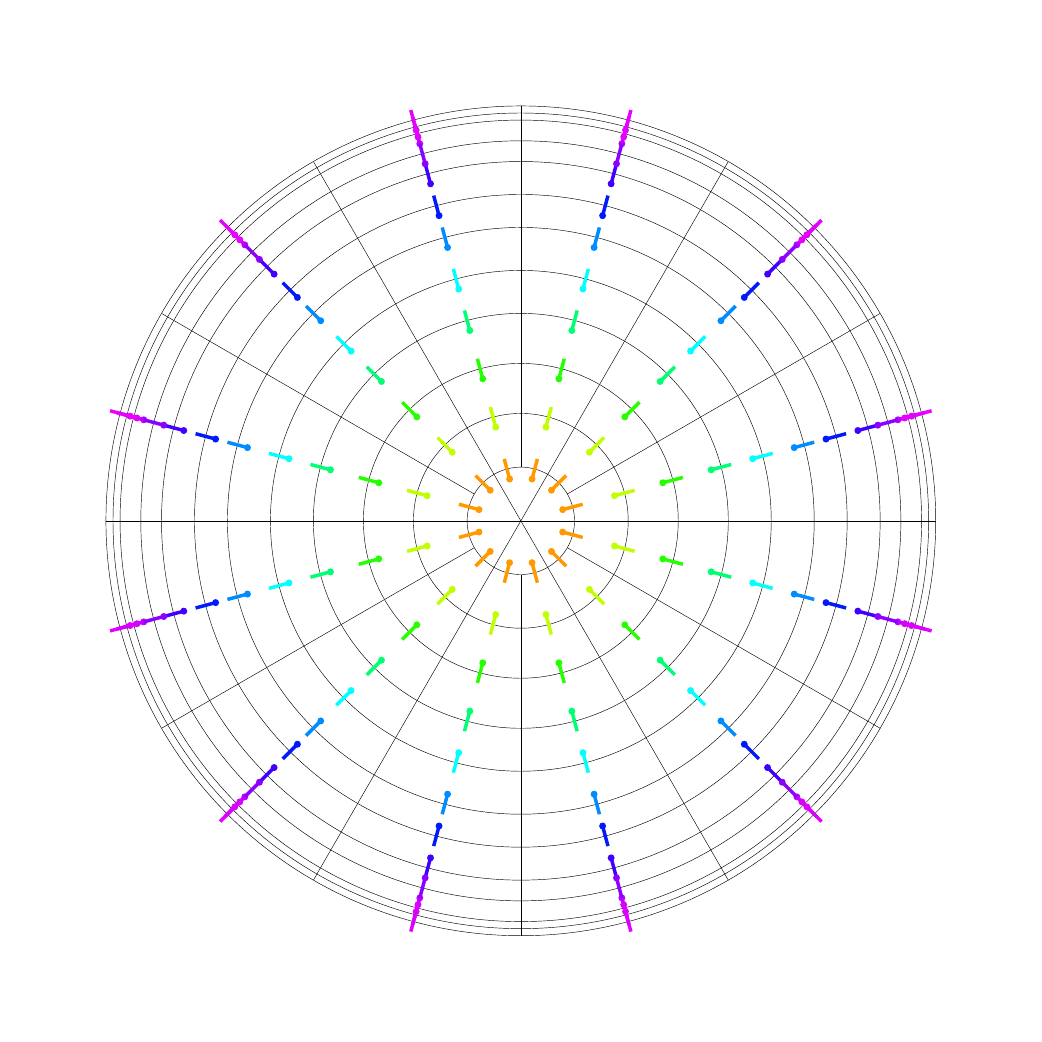}
    \caption{Discrete mixed weak flow rate}
  \end{subfigure}

  \caption{Projections of solutions for diffusion with quadratic potential on a hemisphere with a spherical mesh}
  \label{figure:diffusion/hemisphere_spherical_linear}
\end{figure}

\subsection{Linear potential on an irregular 2D mesh}
\label{sec:simulations/neper}

\begin{example}
  Let $M = [0, 20] \times [0, 15]$ be a rectangle with the Euclidean metric.
  Consider an exact potential
  \begin{equation}
    u \colon M \to \R,\ u(x, y) = 5 x.
  \end{equation}
  Take $\kappa \equiv 6$.
  Then the flow rate is given by
  \begin{equation}
    q = - 6 \star_1 d_0 u = - 6 \star_1(5\, d x) = - 30\, d y.
  \end{equation}
  The internal production rate is
  \begin{equation}
    f = d q = 0.
  \end{equation}
  Let the Dirichlet boundary $\Gamma_\Dirichlet$ consists of the vertical sides, while and the Neumann boundary $\Gamma_\Neumann$ consists of the horizontal sides, i.e.,
  \begin{equation}
    \Gamma_\Dirichlet = \{0, 20\} \times [0, 15],\ \Gamma_\Neumann = [0, 20] \times \{0, 15\}.
  \end{equation}
  The prescribed potential and flow rate are equal to
  \begin{equation}
    g_\Dirichlet(x, y)
    = (\tr_{\Gamma_\Dirichlet, 0}(u))(x, y)
    \begin{cases}
      0, & x = 0 \\
      100, & x = 20
    \end{cases},\
    g_\Neumann = \tr_{\Gamma_\Neumann, 1}(q) = 0.
  \end{equation}
  An irregular mesh $\mathcal{M}$ of $10$ polygons for $M$ and its Forman subdivision $\mathcal{K}$ are shown on \Cref{figure:mesh/2d_10_grains}.
  Exact solutions for the continuous problem on $M$ (discretised on $\mathcal{K})$ and solutions for the primal and mixed weak discrete formulations on $\mathcal{K}$ are visualised on \Cref{figure:diffusion/rectangle_linear}.
  The relative errors are:
  0.0942374
   for primal weak potential;
  0.132704
   for mixed weak potential;
  0.397656
   for primal weak flow rate;
  0.329299
   for mixed weak flow rate.
\end{example}

\begin{figure}[!ht]
  \begin{subfigure}{.45\textwidth}
    \centering
    \includegraphics[width = \linewidth, keepaspectratio]{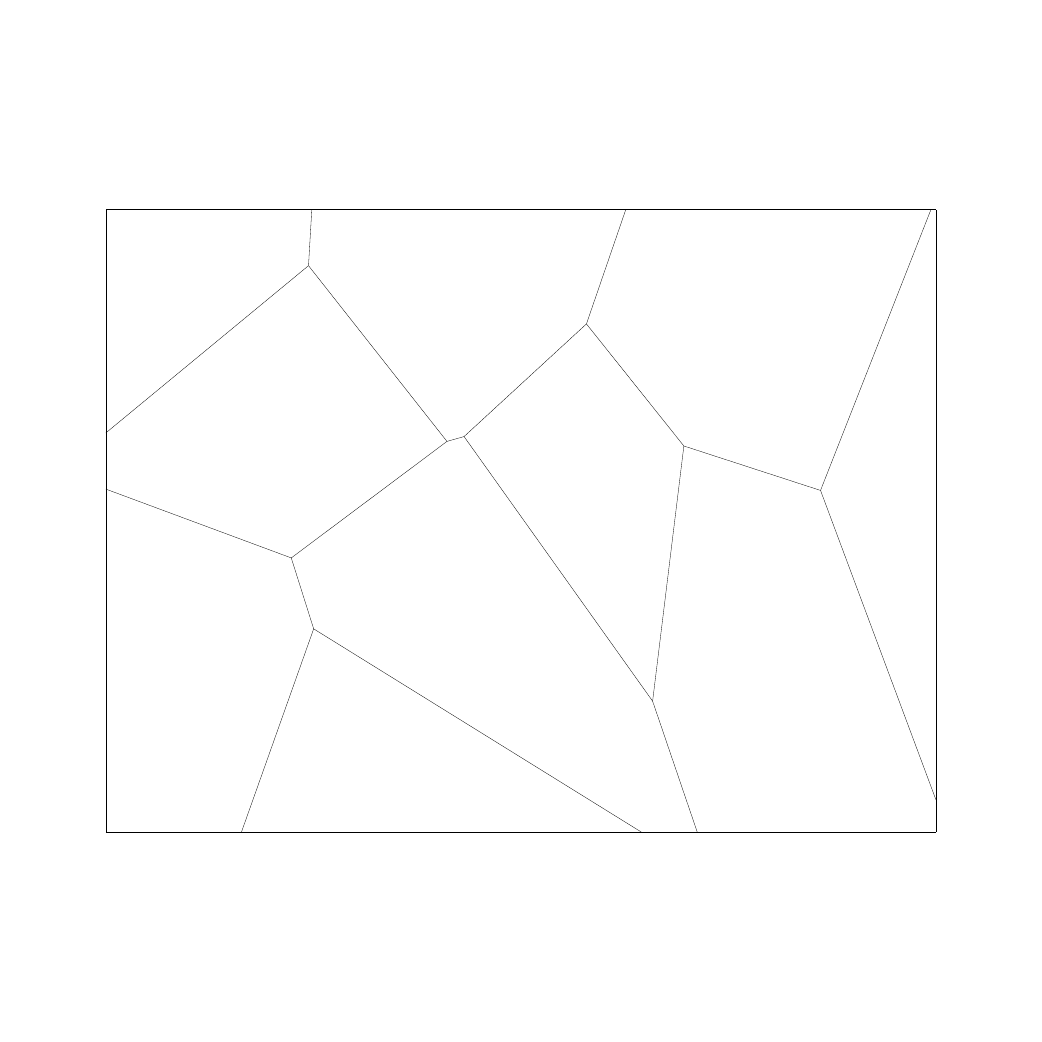}
    \caption{Irregular mesh $\mathcal{M}$ (produced by Neper)}
  \end{subfigure}
  \begin{subfigure}{.45\textwidth}
    \centering
    \includegraphics[width = \linewidth, keepaspectratio]{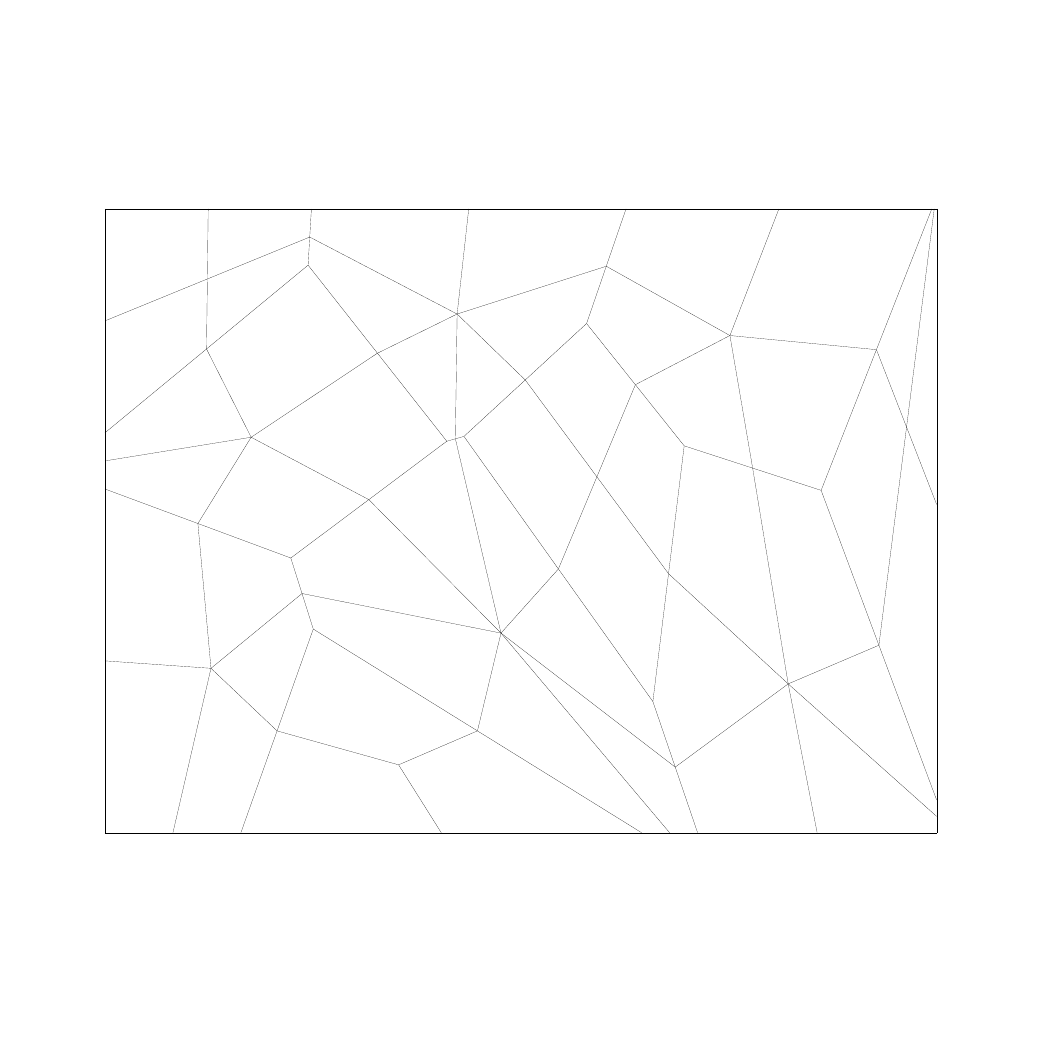}
    \caption{Forman subdivision $\mathcal{K}$ of $\mathcal{M}$}
  \end{subfigure}
  \caption{Irregular mesh on a rectangle and its Forman subdivision}
  \label{figure:mesh/2d_10_grains}
\end{figure}

\begin{figure}[!ht]
  \begin{subfigure}{.3\textwidth}
    \centering
    \includegraphics[width = \linewidth, keepaspectratio]{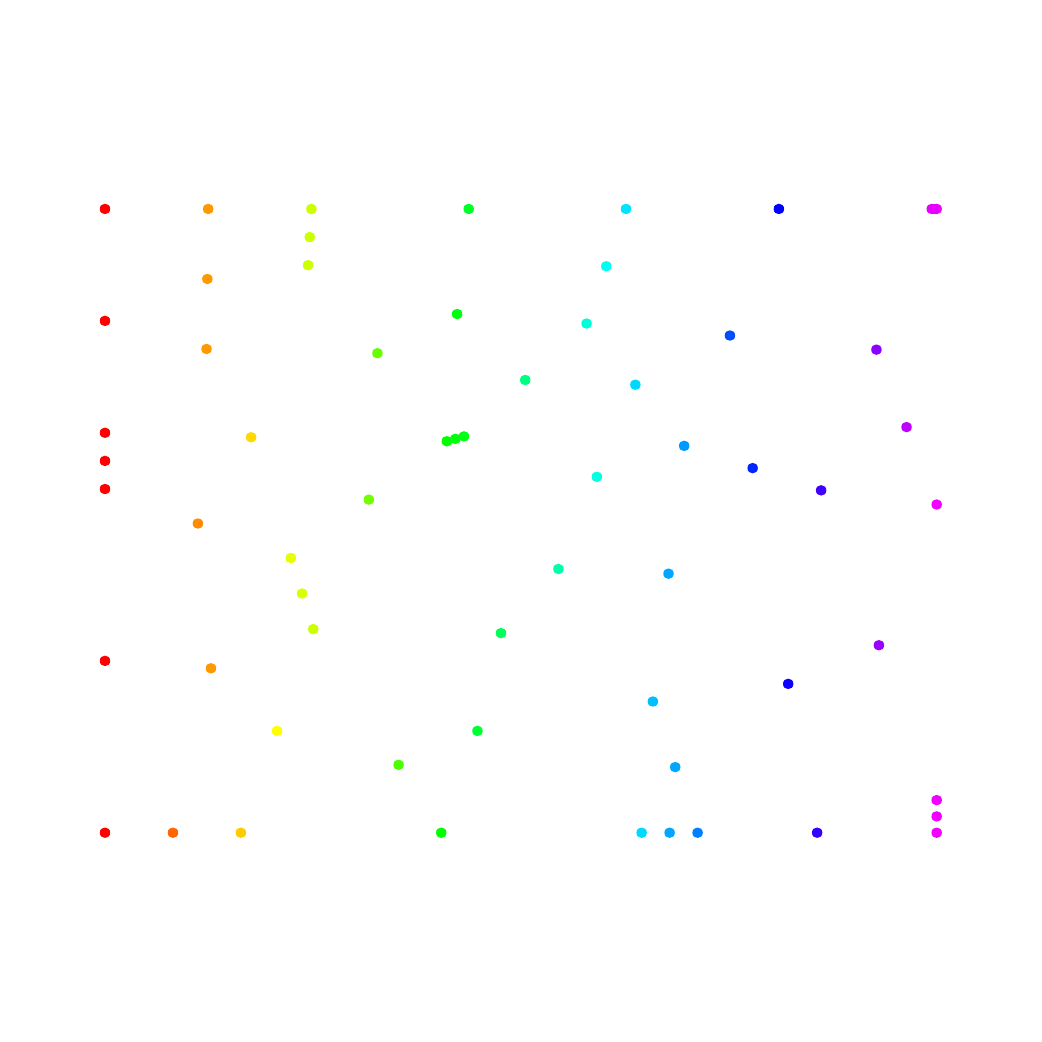}
    \caption{Exact continuous potential}
  \end{subfigure}
  \begin{subfigure}{.3\textwidth}
    \centering
    \includegraphics[width = \linewidth, keepaspectratio]{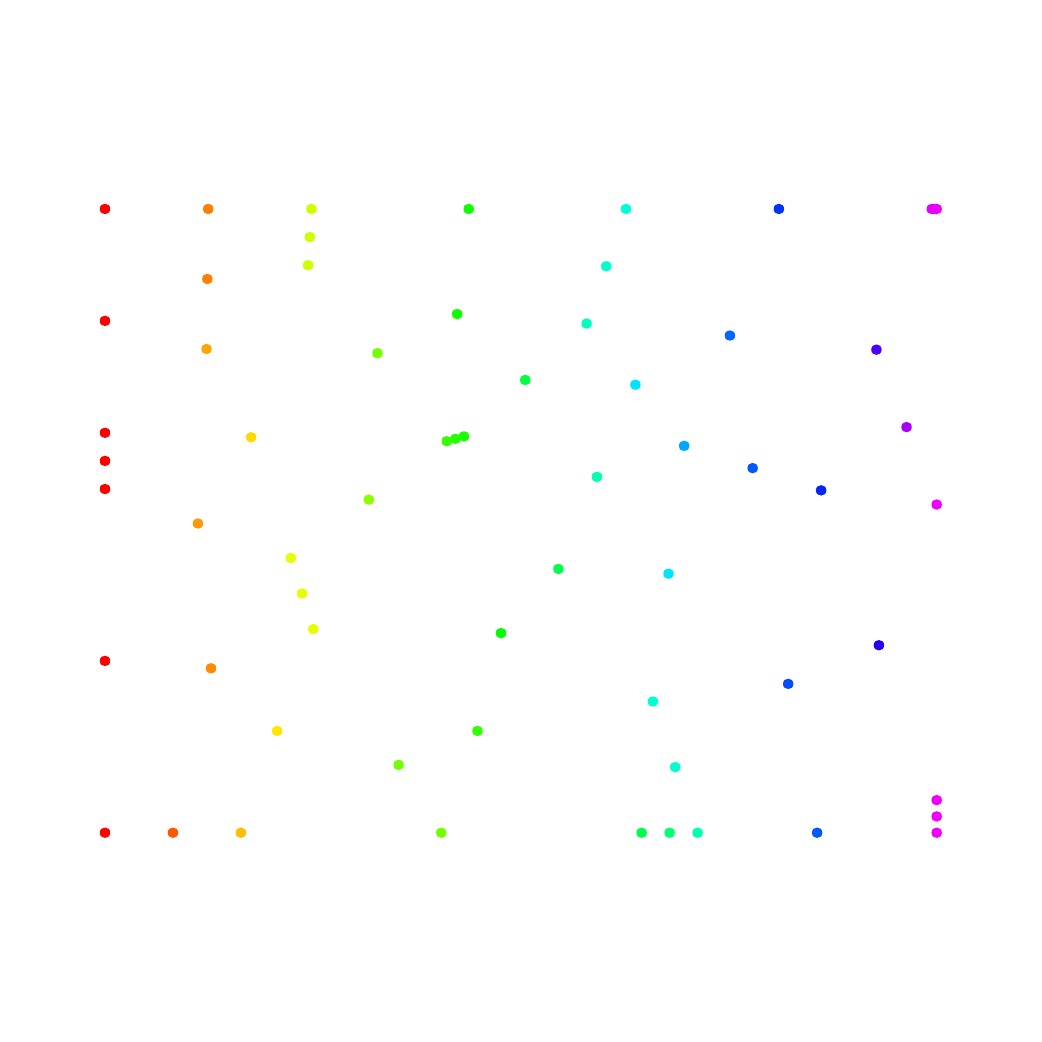}
    \caption{Discrete primal weak potential}
  \end{subfigure}
  \begin{subfigure}{.3\textwidth}
    \centering
    \includegraphics[width = \linewidth, keepaspectratio]{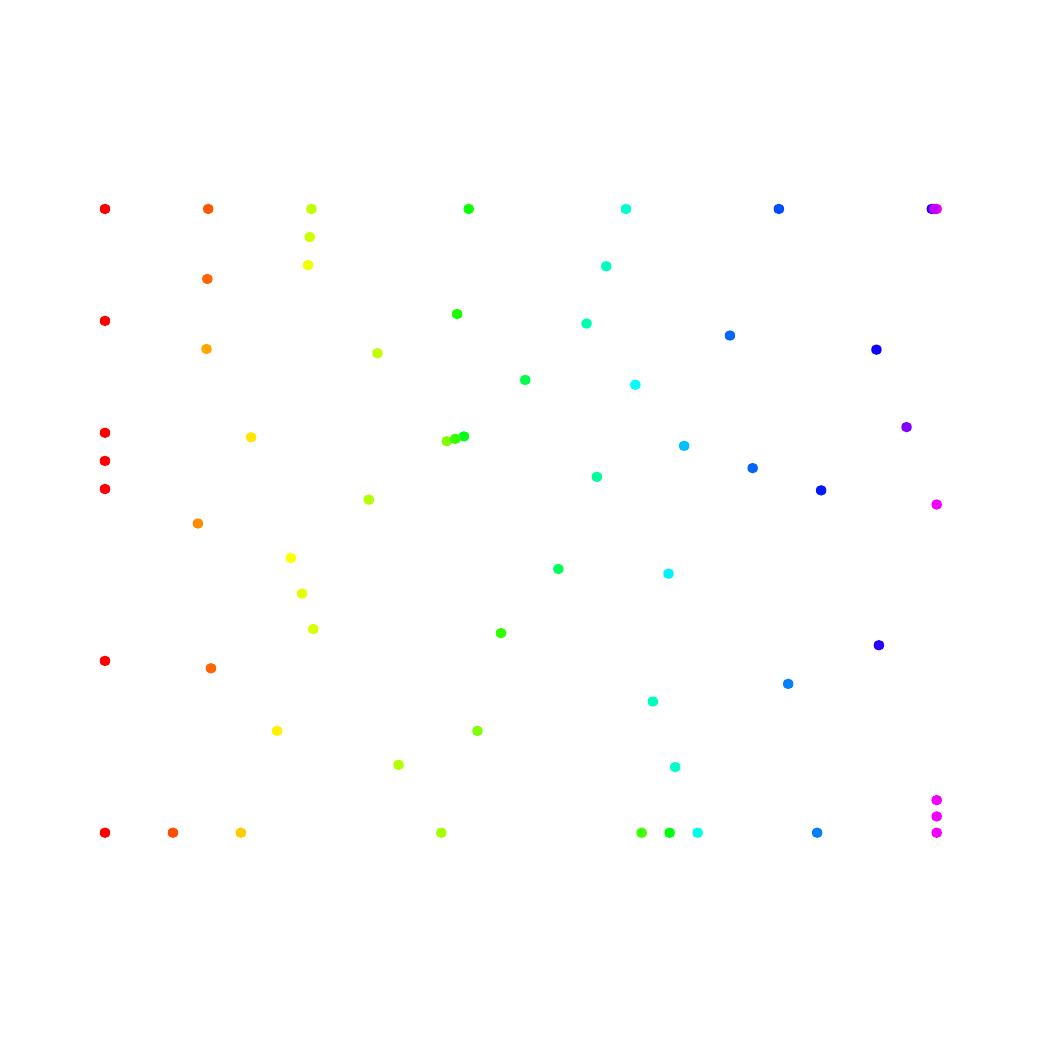}
    \caption{Discrete mixed weak potential}
  \end{subfigure}

  \begin{subfigure}{.3\textwidth}
    \centering
    \includegraphics[width = \linewidth, keepaspectratio]{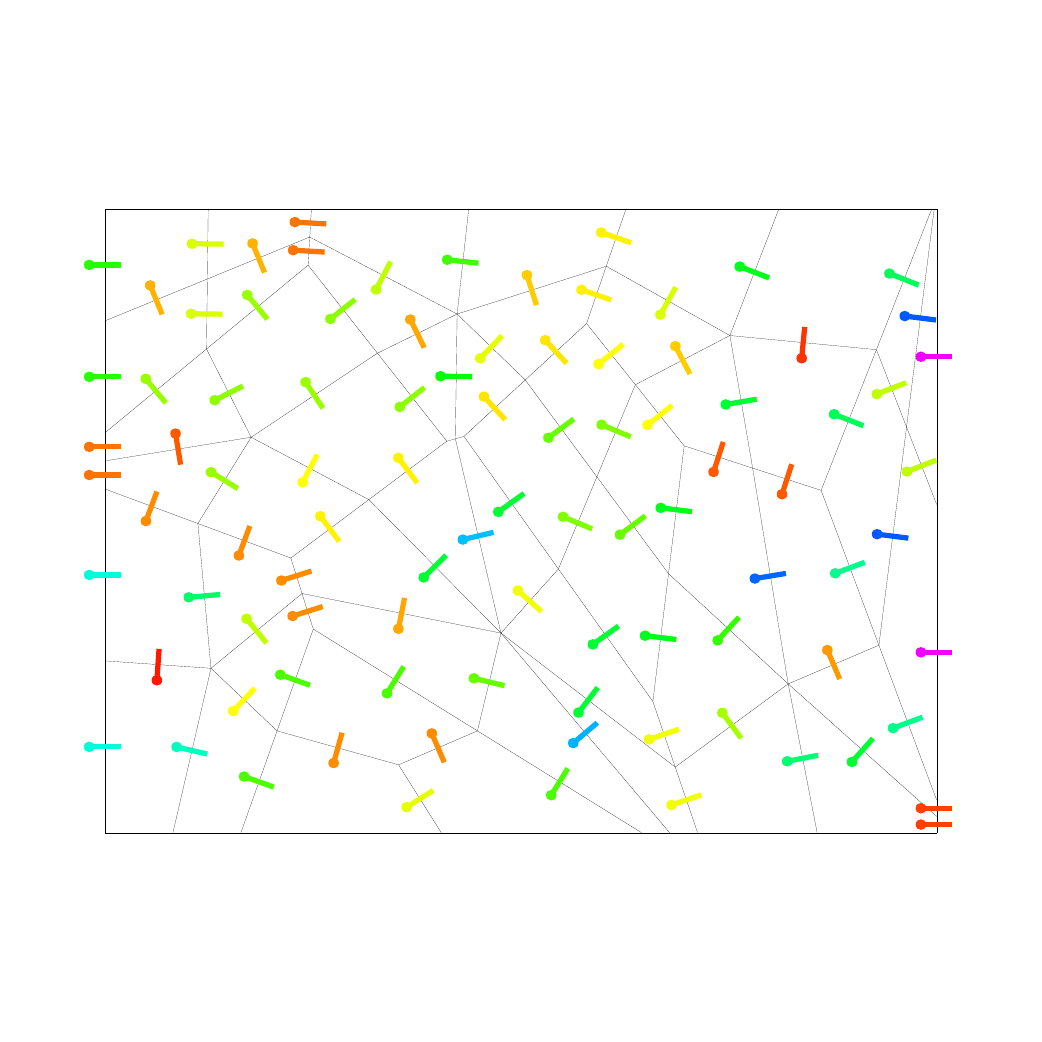}
    \caption{Exact continuous flow rate}
  \end{subfigure}
  \begin{subfigure}{.3\textwidth}
    \centering
    \includegraphics[width = \linewidth, keepaspectratio]{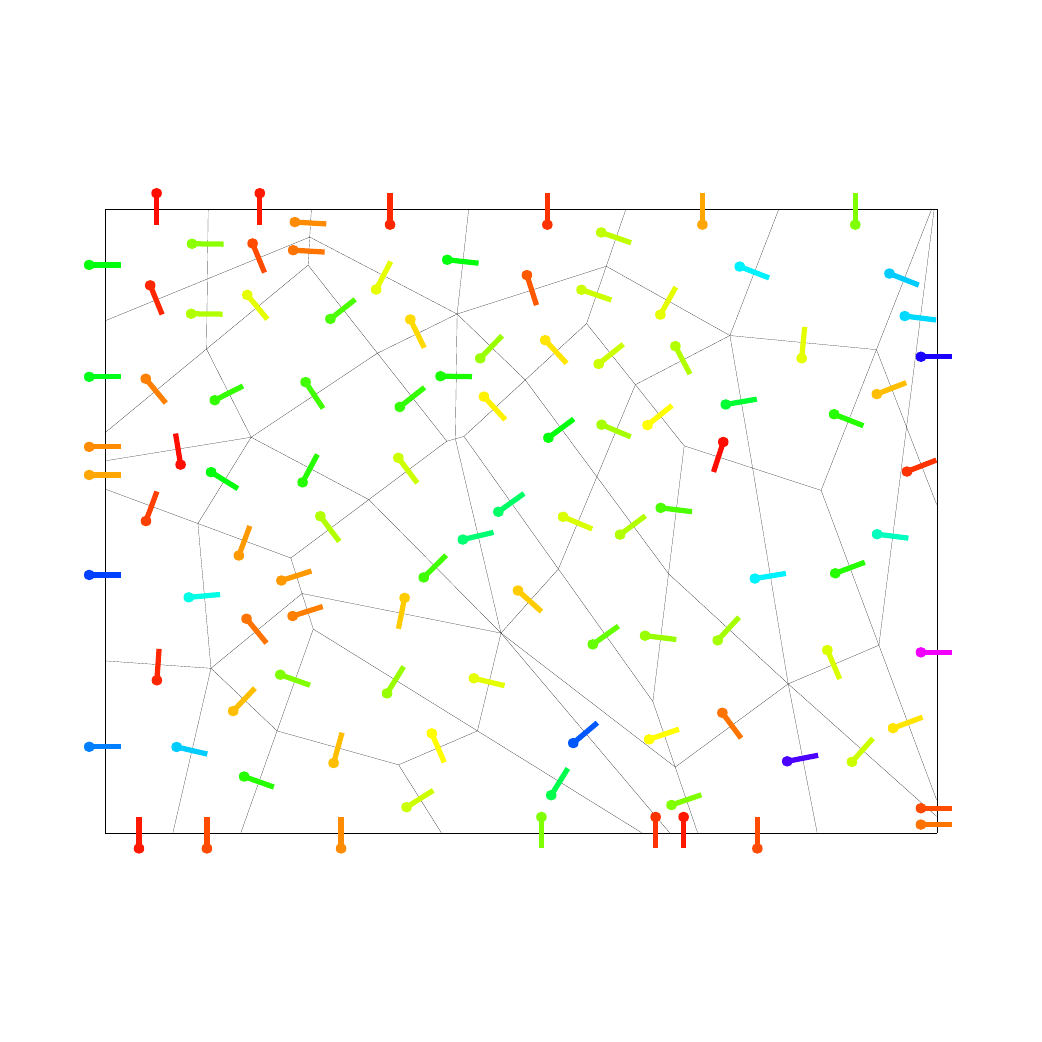}
    \caption{Discrete primal weak flow rate}
  \end{subfigure}
  \begin{subfigure}{.3\textwidth}
    \centering
    \includegraphics[width = \linewidth, keepaspectratio]{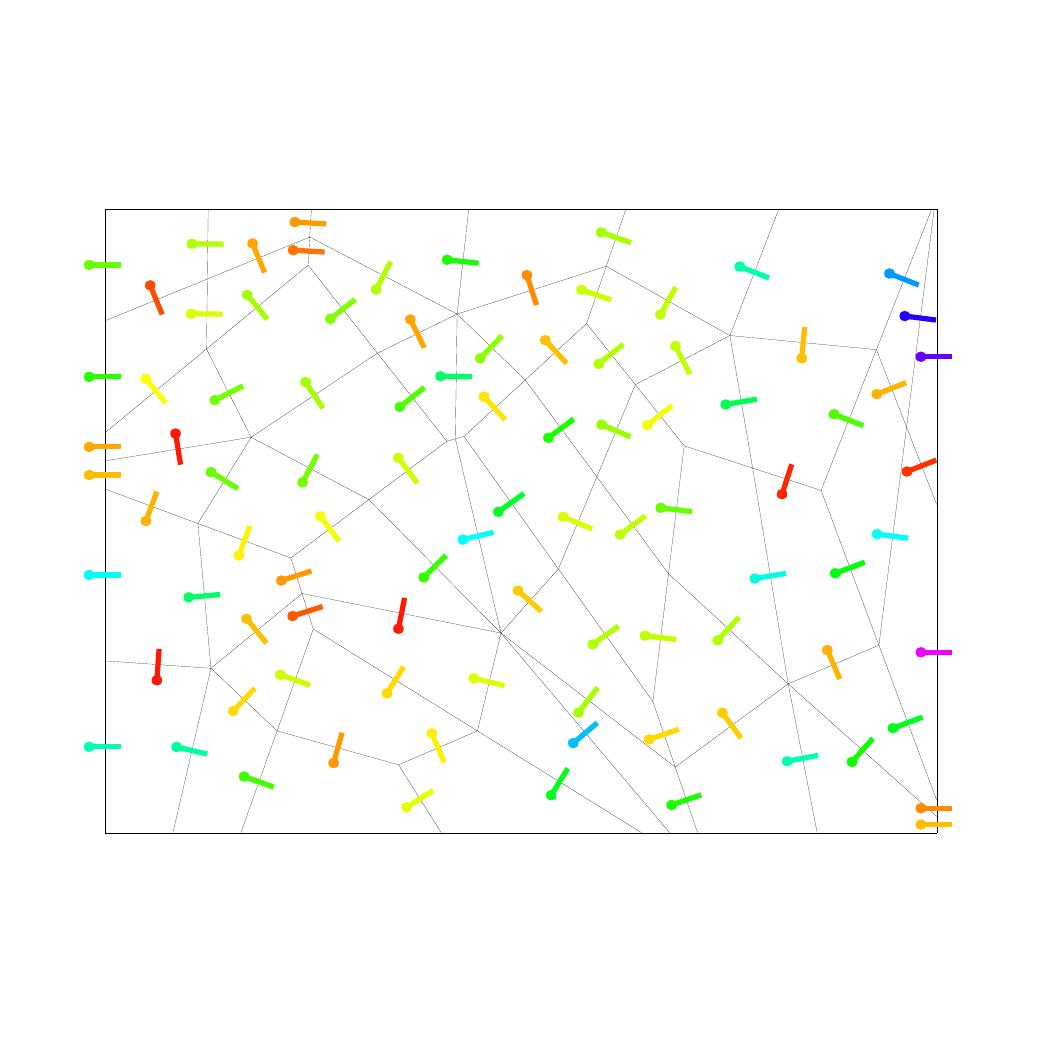}
    \caption{Discrete mixed weak flow rate}
  \end{subfigure}

  \caption{Solutions for diffusion with linear potential on a rectangle with irregular mesh $\mathcal{K}$}
  \label{figure:diffusion/rectangle_linear}
\end{figure}
\begin{discussion}
  In the first three examples we observed small relative errors for both potentials and flow rates.
  A common feature of those examples was that in the respective embedded meshes the topologically orthogonal sides were also geometrically orthogonal. This is not the case for the irregular mesh in the last example, suggesting that non-orthogonality of the embedded mesh is the reason for the high relative errors. Moreover, our tests on a problem with a parallelogram grid showed similarly high errors.
  As with other structure-preserving discretisations, the numerical accuracy of
  the proposed formulation depends on mesh geometry, with skewed or
  non-orthogonal meshes leading to reduced performance, as observed in this
  example.
  A compelling reason for this is that the inner product on $1$-cells does not produce optimal Hodge star when topologically orthogonal edges are not geometrically orthogonal.
  This can be mended potentially by using non-diagonal inner products, i.e., by introducing additional measures for pairs of $p$-cells having a common boundary $(p-1)$-cell.
  Such an inner product will lead to non-local Hodge star (see \cite[Section 8]{wilson2007cochain}) and adjoint coboundary operators but will still be computationally efficient.
  Indeed, the resulting linear systems for both primal and mixed formulations use only inner products, and can be expressed as sparse linear systems (although elimination of variables in the mixed formulation will not be effective).
  Regarding post-processing in the primal weak formulation, it can be reduced to solving a linear system with a sparse positive definite matrix, the matrix of the inner product.

  The definition and usage of inner products with non-orthogonal basis cochains (but still local) is a topic of an ongoing research.
\end{discussion}

\section{Conclusions}

We have developed a new approach to modelling transport phenomena that bridges the gap between discrete microstructural topology and continuum physics. The Combinatorial Mesh Calculus (CMC) presented here provides the first variational formulations -- both primal and mixed -- for conservation laws on cell complexes, establishing a mathematical framework that is more faithful to material physics.

\textit{Theoretical significance}: Our work demonstrates that transport phenomena can be formulated directly on discrete topological structures without reference to smooth manifolds. The exact preservation of conservation laws, discrete-continuous correspondences, and the natural treatment of multi-dimensional transport pathways represent a departure from approximating smooth equations to working intrinsically with the discrete topology of materials.

\textit{Computational impact}: The diagonal structure of key matrices in the mixed formulation enables efficient solution strategies. Our examples show that CMC handles curved geometries and irregular meshes naturally. While the examples presented use manufactured solutions for verification, they demonstrate the method's readiness for complex microstructural problems where different cell dimensions require distinct material properties.

\textit{Future directions}: The CMC framework opens several research avenues:
\begin{enumerate}
    \item Extension to momentum conservation for modelling deformation and defect dynamics (work in progress)
    \item Coupling with experimental microstructural data through direct topology extraction
    \item Multi-physics formulations leveraging the unified treatment of different conservation laws
    \item Development of adaptive refinement strategies based on topological rather than geometric criteria
\end{enumerate}

\textit{Broader implications}: As materials science increasingly reveals the importance of microstructural features across scales, modelling approaches must evolve beyond the continuum-discrete dichotomy. CMC represents a mathematical framework aligned with physical reality -- materials are neither collections of points nor smooth continua but complex assemblies of interacting components with distinct topological dimensions. By providing rigorous mathematics for this intermediate scale, we enable predictive modelling of phenomena from grain boundary engineering to defect-mediated transport.
The computational efficiency demonstrated here, combined with the method's intrinsic handling of complex geometries, positions CMC as a practical tool for materials design and analysis. The open source implementation ensures accessibility for both theoretical development and practical applications. This work establishes the foundation for a new generation of structure-aware computational methods in materials science.

\section*{CRediT authorship contrinution statement}

Kiprian Berbatov:
  Writing -- original draft,
  Writing -- review \& editing,
  Conceptualization,
  Methodology,
  Investigation,
  Software,
  Validation,
  Visualization.
Andrey P. Jivkov:
  Writing -- original draft,
  Writing -- review \& editing,
  Conceptualization,
  Methodology,
  Funding acquisition,
  Project administration.

\section*{Data availability}

Code repository for this project can be found at
\url{https://github.com/kipiberbatov/cmc}.

\section*{Declaration of competing interest}

The authors declare no conflict of interests.

\section*{Acknowledgements}

Jivkov acknowledges the financial support from the Engineering and Physical Sciences Research Council, UK, via grant EP/N026136/1.


\begin{thebibliography}{10}
\expandafter\ifx\csname url\endcsname\relax
  \def\url#1{\texttt{#1}}\fi
\expandafter\ifx\csname urlprefix\endcsname\relax\def\urlprefix{URL }\fi
\expandafter\ifx\csname href\endcsname\relax
  \def\href#1#2{#2} \def\path#1{#1}\fi

\bibitem{kumar2023MC}
C.~Kumar, M.~Singh, R.~Krishna,
  \href{https://www.taylorfrancis.com/books/mono/10.1201/9781003340546/advanced-materials-characterization-ch-sateesh-kumar-muralidhar-singh-ram-krishna}{Advanced
  Materials Characterisation: Basic Principles, Novel Applications, and Future
  Directions}, CRC Press, 2023.
\newblock \href {https://doi.org/10.1201/9781003340546}
  {\path{doi:10.1201/9781003340546}}.
\newline\urlprefix\url{https://www.taylorfrancis.com/books/mono/10.1201/9781003340546/advanced-materials-characterization-ch-sateesh-kumar-muralidhar-singh-ram-krishna}

\bibitem{hayashi2024MC}
K.~Hayashi, (Ed),
  \href{https://link.springer.com/book/10.1007/978-981-99-5235-9}{Hyperordered
  Structures in Materials: Disorder in Order and Order within Disorder},
  Springer Nature Singapore, 2024.
\newblock \href {https://doi.org/10.1007/978-981-99-5235-9}
  {\path{doi:10.1007/978-981-99-5235-9}}.
\newline\urlprefix\url{https://link.springer.com/book/10.1007/978-981-99-5235-9}

\bibitem{stock2020XCT}
S.~Stock,
  \href{https://www.taylorfrancis.com/books/mono/10.1201/9781420058772/microcomputed-tomography-stuart-stock}{Microcomputed
  Tomography: Methodology and Applications}, CRC Press, 2020.
\newblock \href {https://doi.org/10.1201/9781420058772}
  {\path{doi:10.1201/9781420058772}}.
\newline\urlprefix\url{https://www.taylorfrancis.com/books/mono/10.1201/9781420058772/microcomputed-tomography-stuart-stock}

\bibitem{woodruff2021SCT}
D.~Woodruff,
  \href{https://www.cambridge.org/core/books/synchrotron-radiation/554022D36075749FAFC51E0618E35936}{Synchrotron
  Radiation: Sources and Applications to the Structural and Electronic
  Properties of Materials}, Cambridge University Press, 2021.
\newblock \href {https://doi.org/10.1017/9781316995747}
  {\path{doi:10.1017/9781316995747}}.
\newline\urlprefix\url{https://www.cambridge.org/core/books/synchrotron-radiation/554022D36075749FAFC51E0618E35936}

\bibitem{hunklinger2022SSP}
S.~Hunklinger, C.~Enss,
  \href{https://www.degruyterbrill.com/document/doi/10.1515/9783110666502/html}{Solid
  State Physics}, De Gruyter, 2022.
\newblock \href {https://doi.org/10.1515/9783110666502}
  {\path{doi:10.1515/9783110666502}}.
\newline\urlprefix\url{https://www.degruyterbrill.com/document/doi/10.1515/9783110666502/html}

\bibitem{trovalusci2016MSM}
P.~Trovalusci, (Ed),
  \href{https://link.springer.com/book/10.1007/978-3-319-21494-8}{Materials
  with Internal Structure: Multiscale and Multifield Modeling and Simulation},
  Springer International Publishing Switzerland, 2016.
\newblock \href {https://doi.org/10.1007/978-3-319-21494-8}
  {\path{doi:10.1007/978-3-319-21494-8}}.
\newline\urlprefix\url{https://link.springer.com/book/10.1007/978-3-319-21494-8}

\bibitem{brancherie2017MSM}
D.~Brancherie, P.~Feissel, S.~Bouvier, A.~Ibrahimbegovic, (Eds),
  \href{https://onlinelibrary.wiley.com/doi/book/10.1002/9781119476757}{From
  Microstructure Investigation to Multiscale Modelling: Bridging the Gap}, John
  Whiley \& Sons, 2017.
\newblock \href {https://doi.org/10.1002/9781119476757}
  {\path{doi:10.1002/9781119476757}}.
\newline\urlprefix\url{https://onlinelibrary.wiley.com/doi/book/10.1002/9781119476757}

\bibitem{burczynski2022MSM}
T.~Burczynski, M.~Pietrzyk, K.~W, L.~Madej, A.~Mrozek, L.~Rauch,
  \href{https://onlinelibrary.wiley.com/doi/book/10.1002/9781118536445}{Multiscale
  Modelling and Optimisation of Materials and Structures}, John Wiley \& Sons,
  2022.
\newblock \href {https://doi.org/10.1002/9781118536445}
  {\path{doi:10.1002/9781118536445}}.
\newline\urlprefix\url{https://onlinelibrary.wiley.com/doi/book/10.1002/9781118536445}

\bibitem{baniassadi2023MCM}
M.~Baniassadi, M.~Baghani, Y.~Remond,
  \href{https://www.sciencedirect.com/book/monograph/9780443189913/applied-micromechanics-of-complex-microstructures}{Applied
  Micromechanics of Complex Microstructures: Computational Modelling and
  Numerical Characterisation}, Elsevier, 2023.
\newblock \href {https://doi.org/10.1016/C2022-0-01129-6}
  {\path{doi:10.1016/C2022-0-01129-6}}.
\newline\urlprefix\url{https://www.sciencedirect.com/book/monograph/9780443189913/applied-micromechanics-of-complex-microstructures}

\bibitem{morro2023Continuum}
A.~Morro, C.~Giorgi,
  \href{https://link.springer.com/book/10.1007/978-3-031-20814-0}{Mathematical
  Modelling of Continuum Physics}, Springer Nature Switzerland AG, 2023.
\newblock \href {https://doi.org/10.1007/978-3-031-20814-0}
  {\path{doi:10.1007/978-3-031-20814-0}}.
\newline\urlprefix\url{https://link.springer.com/book/10.1007/978-3-031-20814-0}

\bibitem{brenner2008FEM}
S.~C. Brenner, L.~R. Scott,
  \href{https://link.springer.com/book/10.1007/978-0-387-75934-0}{The
  Mathematical Theory of Finite Element Methods}, Springer Science+Business
  Media, 2008.
\newblock \href {https://doi.org/10.1007/978-0-387-75934-0}
  {\path{doi:10.1007/978-0-387-75934-0}}.
\newline\urlprefix\url{https://link.springer.com/book/10.1007/978-0-387-75934-0}

\bibitem{thomas1995FDM}
J.~W. Thomas,
  \href{https://link.springer.com/book/10.1007/978-1-4899-7278-1}{Numerical
  Partial Differential Equations: Finite Difference Methods}, Springer, 1995.
\newblock \href {https://doi.org/10.1007/978-1-4899-7278-1}
  {\path{doi:10.1007/978-1-4899-7278-1}}.
\newline\urlprefix\url{https://link.springer.com/book/10.1007/978-1-4899-7278-1}

\bibitem{leveque2012FVM}
R.~J. LeVeque,
  \href{https://www.cambridge.org/core/books/finite-volume-methods-for-hyperbolic-problems/97D5D1ACB1926DA1D4D52EAD6909E2B9}{Finite
  Volume Methods for Hyperbolic Problems}, Cambridge University Press, 2012.
\newblock \href {https://doi.org/10.1017/CBO9780511791253}
  {\path{doi:10.1017/CBO9780511791253}}.
\newline\urlprefix\url{https://www.cambridge.org/core/books/finite-volume-methods-for-hyperbolic-problems/97D5D1ACB1926DA1D4D52EAD6909E2B9}

\bibitem{ruiz2001CZM}
G.~Ruiz, A.~Pandolfi, M.~Ortiz,
  \href{https://onlinelibrary.wiley.com/doi/10.1002/nme.273}{Three-dimensional
  cohesive modeling of dynamic mixed-mode fracture}, International Journal for
  Numerical Methods in Engineering 52~(1-2) (2001) 97 – 120.
\newblock \href {https://doi.org/10.1002/nme.273} {\path{doi:10.1002/nme.273}}.
\newline\urlprefix\url{https://onlinelibrary.wiley.com/doi/10.1002/nme.273}

\bibitem{park2011CZM}
K.~Park, G.~H. Paulino,
  \href{https://asmedigitalcollection.asme.org/appliedmechanicsreviews/article/64/6/060802/370063/Cohesive-Zone-Models-A-Critical-Review-of-Traction}{Cohesive
  zone models: A critical review of traction-separation relationships across
  fracture surfaces}, Applied Mechanics Reviews 64 (2011) 061002.
\newblock \href {https://doi.org/10.1115/1.4023110}
  {\path{doi:10.1115/1.4023110}}.
\newline\urlprefix\url{https://asmedigitalcollection.asme.org/appliedmechanicsreviews/article/64/6/060802/370063/Cohesive-Zone-Models-A-Critical-Review-of-Traction}

\bibitem{mohammadi2012XFEM}
S.~Mohammadi,
  \href{https://onlinelibrary.wiley.com/doi/book/10.1002/9781118443378}{XFEM
  Fracture Analysis of Composites}, John Wiley \& Sons, 2012.
\newblock \href {https://doi.org/10.1002/9781118443378}
  {\path{doi:10.1002/9781118443378}}.
\newline\urlprefix\url{https://onlinelibrary.wiley.com/doi/book/10.1002/9781118443378}

\bibitem{vellwock2024XFEM}
A.~E. Vellwock, F.~Libonati,
  \href{https://www.mdpi.com/1996-1944/17/3/745}{Xfem for composites,
  biological, and bioinspired materials: A review}, Materials 17 (2024) 745.
\newblock \href {https://doi.org/10.3390/ma17030745}
  {\path{doi:10.3390/ma17030745}}.
\newline\urlprefix\url{https://www.mdpi.com/1996-1944/17/3/745}

\bibitem{chen2022phase}
L.-Q. Chen, Y.~Zhao,
  \href{https://www.sciencedirect.com/science/article/pii/S007964252100092X}{From
  classical thermodynamics to phase-field method}, Progress in Materials
  Science 124 (2022) 100868.
\newblock \href {https://doi.org/10.1016/j.pmatsci.2021.100868}
  {\path{doi:10.1016/j.pmatsci.2021.100868}}.
\newline\urlprefix\url{https://www.sciencedirect.com/science/article/pii/S007964252100092X}

\bibitem{chen2024phase}
L.-Q. Chen, N.~Moelans,
  \href{https://link.springer.com/article/10.1557/s43577-024-00724-7}{Phase-field
  method of materials microstructures and properties}, MRS Bulletin 49 (2022)
  551–--555.
\newblock \href {https://doi.org/10.1557/s43577-024-00724-7}
  {\path{doi:10.1557/s43577-024-00724-7}}.
\newline\urlprefix\url{https://link.springer.com/article/10.1557/s43577-024-00724-7}

\bibitem{santamaria2023MD}
R.~Santamaria,
  \href{https://link.springer.com/book/10.1007/978-3-031-37042-7}{Molecular
  Dynamics}, Springer Nature Switzerland AG, 2023.
\newblock \href {https://doi.org/10.1007/978-3-031-37042-7}
  {\path{doi:10.1007/978-3-031-37042-7}}.
\newline\urlprefix\url{https://link.springer.com/book/10.1007/978-3-031-37042-7}

\bibitem{madenci2014PD}
E.~Madenci, E.~Oterkus,
  \href{https://link.springer.com/book/10.1007/978-1-4614-8465-3}{Peridynamic
  Theory and Its Applications}, Springer Science+Business Media, 2014.
\newblock \href {https://doi.org/10.1007/978-1-4614-8465-3}
  {\path{doi:10.1007/978-1-4614-8465-3}}.
\newline\urlprefix\url{https://link.springer.com/book/10.1007/978-1-4614-8465-3}

\bibitem{filho2019SPH}
C.~A. D.~F. Filho,
  \href{https://link.springer.com/book/10.1007/978-3-030-00773-7}{Smoothed
  Particle Hydrodynamics}, Springer Nature Switzerland AG, 2019.
\newblock \href {https://doi.org/10.1007/978-3-030-00773-7}
  {\path{doi:10.1007/978-3-030-00773-7}}.
\newline\urlprefix\url{https://link.springer.com/book/10.1007/978-3-030-00773-7}

\bibitem{jebahi2015DEM}
M.~Jebahi, D.~André, I.~Terreros, I.~Iordanoff,
  \href{https://onlinelibrary.wiley.com/doi/book/10.1002/9781119103042}{Discrete
  Element Method to Model 3D Continuous Materials}, John Wiley \& Sons, Inc.,
  2015.
\newblock \href {https://doi.org/10.1002/9781119103042}
  {\path{doi:10.1002/9781119103042}}.
\newline\urlprefix\url{https://onlinelibrary.wiley.com/doi/book/10.1002/9781119103042}

\bibitem{kozlov2008}
D.~Kozlov,
  \href{https://link.springer.com/book/10.1007/978-3-540-71962-5}{Combinatorial
  Algebraic Topology}, Springer, 2008.
\newblock \href {https://doi.org/10.1007/978-3-540-71962-5}
  {\path{doi:10.1007/978-3-540-71962-5}}.
\newline\urlprefix\url{https://link.springer.com/book/10.1007/978-3-540-71962-5}

\bibitem{knudson2022}
K.~Knudson,
  \href{https://www.degruyterbrill.com/document/doi/10.1515/9783111014852/html}{Algebraic
  Topology: A Toolkit}, De Gruyter, 2022.
\newblock \href {https://doi.org/10.1515/9783111014852}
  {\path{doi:10.1515/9783111014852}}.
\newline\urlprefix\url{https://www.degruyterbrill.com/document/doi/10.1515/9783111014852/html}

\bibitem{berbatov2022diffusion}
K.~Berbatov, P.~D. Boom, A.~L. Hazel, A.~P. Jivkov,
  \href{https://www.sciencedirect.com/science/article/pii/S0307904X22002657}{Diffusion
  in multi-dimensional solids using {Forman}'s combinatorial differential
  forms}, Applied Mathematical Modelling 110 (2022) 172--192.
\newblock \href {https://doi.org/10.1016/j.apm.2022.05.043}
  {\path{doi:10.1016/j.apm.2022.05.043}}.
\newline\urlprefix\url{https://www.sciencedirect.com/science/article/pii/S0307904X22002657}

\bibitem{forman2002combinatorial}
R.~Forman,
  \href{https://www.worldscientific.com/doi/abs/10.1142/S0129167X02001265}{Combinatorial
  {Novikov}-{Morse} theory}, International Journal of Mathematics 13~(4) (2002)
  333--368.
\newblock \href {https://doi.org/10.1142/S0129167X02001265}
  {\path{doi:10.1142/S0129167X02001265}}.
\newline\urlprefix\url{https://www.worldscientific.com/doi/abs/10.1142/S0129167X02001265}

\bibitem{lee2012introduction}
J.~M. Lee,
  \href{https://link.springer.com/book/10.1007/978-1-4419-9982-5}{Introduction
  to smooth manifolds}, Vol. 218 of Graduate Texts in Mathematics, Springer,
  2012.
\newblock \href {https://doi.org/https://doi.org/10.1007/978-1-4419-9982-5}
  {\path{doi:https://doi.org/10.1007/978-1-4419-9982-5}}.
\newline\urlprefix\url{https://link.springer.com/book/10.1007/978-1-4419-9982-5}

\bibitem{arnold2018finite}
D.~Arnold,
  \href{https://epubs.siam.org/doi/book/10.1137/1.9781611975543}{Finite Element
  Exterior Calculus}, Vol.~93 of CBMS-NSF Regional Conference Series in Applied
  Mathematics, SIAM, 2018.
\newblock \href {https://doi.org/10.1137/1.9781611975543}
  {\path{doi:10.1137/1.9781611975543}}.
\newline\urlprefix\url{https://epubs.siam.org/doi/book/10.1137/1.9781611975543}

\bibitem{berbatov2023discrete}
K.~Berbatov,
  \href{https://pure.manchester.ac.uk/ws/portalfiles/portal/261212897/FULL_TEXT.PDF}{Discrete
  approaches to mechanics and physics of solids}, Ph.d. thesis, The University
  of Manchester (2023).
\newline\urlprefix\url{https://pure.manchester.ac.uk/ws/portalfiles/portal/261212897/FULL_TEXT.PDF}

\bibitem{berbatov2026cmc}
K.~Berbatov, M.~Azeem, A.~Jivkov, C.~Liu,
  \href{https://zenodo.org/records/18325324}{kipiberbatov/cmc: Initial version
  for referencing (v0.1)} (2026).
\newblock \href {https://doi.org/10.5281/zenodo.18325324}
  {\path{doi:10.5281/zenodo.18325324}}.
\newline\urlprefix\url{https://zenodo.org/records/18325324}

\bibitem{ziegler1995lectures}
G.~M. Ziegler,
  \href{https://link.springer.com/book/10.1007%2F978-1-4613-8431-1}{Lectures on
  polytopes}, Vol. 152, Springer, 1995.
\newblock \href {https://doi.org/https://doi.org/10.1007/978-1-4613-8431-1}
  {\path{doi:https://doi.org/10.1007/978-1-4613-8431-1}}.
\newline\urlprefix\url{https://link.springer.com/book/10.1007%2F978-1-4613-8431-1}

\bibitem{bjorner1984posets}
A.~Bj{\"o}rner,
  \href{https://www.sciencedirect.com/science/article/pii/S0195669884800128}{Posets,
  regular cw complexes and {Bruhat} order}, European Journal of Combinatorics
  5~(1) (1984) 7--16.
\newblock \href {https://doi.org/10.1016/S0195-6698(84)80012-8}
  {\path{doi:10.1016/S0195-6698(84)80012-8}}.
\newline\urlprefix\url{https://www.sciencedirect.com/science/article/pii/S0195669884800128}

\bibitem{wilson2007cochain}
S.~O. Wilson,
  \href{https://www.sciencedirect.com/science/article/pii/S0166864107000314}{Cochain
  algebra on manifolds and convergence under refinement}, Topology and its
  Applications 154~(9) (2007) 1898--1920.
\newblock \href {https://doi.org/10.1016/j.topol.2007.01.017}
  {\path{doi:10.1016/j.topol.2007.01.017}}.
\newline\urlprefix\url{https://www.sciencedirect.com/science/article/pii/S0166864107000314}

\bibitem{arnold2012discrete}
R.~F. Arnold, \href{https://vtechworks.lib.vt.edu/handle/10919/27485}{The
  discrete {Hodge} star operator and {Poincar{\'e}} duality}, Ph.d. thesis,
  Virginia Tech (2012).
\newline\urlprefix\url{https://vtechworks.lib.vt.edu/handle/10919/27485}

\bibitem{ptackova2017discrete}
L.~Pt{\'a}{\v{c}}kov{\'a},
  \href{https://www.visgraf.impa.br/Data/RefBib/PS_PDF/student-phd-2017-02-lenka-ptackova/FinalThesis.pdf}{A
  discrete wedge product on polygonal pseudomanifolds}, Ph.d. thesis, Instituto
  de Matem\'atica Pura e Aplicada (2017).
\newline\urlprefix\url{https://www.visgraf.impa.br/Data/RefBib/PS_PDF/student-phd-2017-02-lenka-ptackova/FinalThesis.pdf}

\bibitem{hirani2003discrete}
A.~Hirani, \href{https://thesis.library.caltech.edu/1885/}{Discrete exterior
  calculus}, Ph.d. thesis, California Institute of Technology (2003).
\newblock \href {https://doi.org/10.7907/ZHY8-V329}
  {\path{doi:10.7907/ZHY8-V329}}.
\newline\urlprefix\url{https://thesis.library.caltech.edu/1885/}

\bibitem{guzman2025framework}
J.~Guzm{\'a}n, P.~Potu, \href{https://arxiv.org/abs/2505.08934}{A framework for
  analysis of dec approximations to hodge-laplacian problems using generalized
  whitney forms}, arXiv preprint arXiv:2505.08934 (2025).
\newblock \href {https://doi.org/10.48550/arXiv.2505.08934}
  {\path{doi:10.48550/arXiv.2505.08934}}.
\newline\urlprefix\url{https://arxiv.org/abs/2505.08934}

\bibitem{botti2018assesment}
L.~Botti, D.~A. Di~Pietro,
  \href{https://www.sciencedirect.com/science/article/pii/S0021999118303176}{Assessment
  of hybrid high-order methods on curved meshes and comparison with
  discontinuous galerkin methods}, Journal of Computational Physics 370 (2018)
  58--84.
\newblock \href {https://doi.org/10.1016/j.jcp.2018.05.017}
  {\path{doi:10.1016/j.jcp.2018.05.017}}.
\newline\urlprefix\url{https://www.sciencedirect.com/science/article/pii/S0021999118303176}

\bibitem{beirao2019virtual}
L.~Beir{\~a}o~da Veiga, A.~Russo, G.~Vacca,
  \href{https://www.esaim-m2an.org/articles/m2an/abs/2019/02/m2an170204/m2an170204.html}{The
  virtual element method with curved edges}, ESAIM: Mathematical Modelling and
  Numerical Analysis 53~(2) (2019) 375--404.
\newblock \href {https://doi.org/10.1051/m2an/2018052}
  {\path{doi:10.1051/m2an/2018052}}.
\newline\urlprefix\url{https://www.esaim-m2an.org/articles/m2an/abs/2019/02/m2an170204/m2an170204.html}

\bibitem{hiptmair2001discrete}
R.~Hiptmair,
  \href{https://link.springer.com/article/10.1007/s002110100295}{Discrete hodge
  operators}, Numerische Mathematik 90~(2) (2001) 265--289.
\newblock \href {https://doi.org/10.1007/s002110100295}
  {\path{doi:10.1007/s002110100295}}.
\newline\urlprefix\url{https://link.springer.com/article/10.1007/s002110100295}

\bibitem{berbatov2021guide}
K.~Berbatov, B.~Lazarov, A.~Jivkov,
  \href{https://www.sciencedirect.com/science/article/pii/S0168927421002014}{A
  guide to the finite and virtual element methods for elasticity}, Applied
  Numerical Mathematics 169 (2021) 351--395.
\newblock \href {https://doi.org/https://doi.org/10.1016/j.apnum.2021.07.010}
  {\path{doi:https://doi.org/10.1016/j.apnum.2021.07.010}}.
\newline\urlprefix\url{https://www.sciencedirect.com/science/article/pii/S0168927421002014}

\bibitem{boffi2013mixed}
D.~Boffi, F.~Brezzi, M.~Fortin,
  \href{https://link.springer.com/book/10.1007/978-3-642-36519-5}{Mixed finite
  element methods and applications}, Vol.~44 of Springer Series in
  Computational Mathematics, Springer, 2013.
\newblock \href {https://doi.org/10.1007/978-3-642-36519-5}
  {\path{doi:10.1007/978-3-642-36519-5}}.
\newline\urlprefix\url{https://link.springer.com/book/10.1007/978-3-642-36519-5}

\bibitem{arnold1985mixed}
D.~N. Arnold, F.~Brezzi,
  \href{https://www.numdam.org/item/M2AN_1985__19_1_7_0/}{Mixed and
  nonconforming finite element methods: implementation, postprocessing and
  error estimates}, ESAIM: Mathematical Modelling and Numerical Analysis 19~(1)
  (1985) 7--32.
\newline\urlprefix\url{https://www.numdam.org/item/M2AN_1985__19_1_7_0/}

\bibitem{quey2011large}
R.~Quey, P.~R. Dawson, F.~Barbe,
  \href{https://www.sciencedirect.com/science/article/pii/S004578251100003X}{Large-scale
  3d random polycrystals for the finite element method: Generation, meshing and
  remeshing}, Computer Methods in Applied Mechanics and Engineering 200~(17-20)
  (2011) 1729--1745.
\newblock \href {https://doi.org/10.1016/j.cma.2011.01.002}
  {\path{doi:10.1016/j.cma.2011.01.002}}.
\newline\urlprefix\url{https://www.sciencedirect.com/science/article/pii/S004578251100003X}

\end{thebibliography}
\end{document}